\documentclass[a4paper,10pt]{article}

\newif\ifextendedversion
\extendedversiontrue

\usepackage[top=3cm,bottom=3cm,left=2.8cm,right=2.8cm]{geometry}
\usepackage{times}
\usepackage{bbm}
\usepackage{microtype}

\usepackage{latexsym}
\usepackage{amsmath,amsthm,amssymb,stmaryrd}
\usepackage{url}
\usepackage{misc}
\usepackage{multirow}
\usepackage{xcolor}
\usepackage{xspace}

\usepackage[numbers]{natbib}

\usepackage{graphicx}
\graphicspath{{../}}

\usepackage{tikz}
\usetikzlibrary{shapes,decorations,backgrounds}
\usetikzlibrary{arrows,calc,fit}

\makeatletter
\def\input@path{{../}}
\makeatother

\newcommand{\nb}[1]{\textcolor{red}{\bf!}%
  \marginpar[\parbox{15mm}{\raggedleft\scriptsize\textcolor{red}{#1}}]%
  {\parbox{15mm}{\raggedright\scriptsize\textcolor{red}{#1}}}}
\renewcommand{\nb}[1]{}

\newtheorem{theorem}{Theorem}
\newtheorem{definition}[theorem]{Definition}
\newtheorem{example}[theorem]{Example}

\newtheorem{lemma}[theorem]{Lemma}

\newcommand{\A}{\mathcal{A}}

\newcommand{\I}{\mathcal{I}}
\newcommand{\J}{\mathcal{J}}

\newcommand{\K}{\mathcal{K}}
\renewcommand{\L}{\mathcal{L}}
\newcommand{\M}{\mathcal{M}}

\newcommand{\R}{\mathcal{R}}
\renewcommand{\S}{\mathcal{S}}
\newcommand{\T}{\mathcal{T}}

\newcommand{\ALC}{\ensuremath{\mathcal{ALC}}\xspace}

\newcommand{\ELUf}{\ensuremath{\mathcal{ELU}_{\mathit{rhs}}}\xspace}

\newcommand{\hALC}{\textsl{Horn-}\ensuremath{\mathcal{ALC}}\xspace}

\newcommand{\EL}{\ensuremath{\mathcal{EL}}\xspace}

\newcommand{\ExpTime}{\textsc{ExpTime}\xspace}
\newcommand{\NExpTime}{\textsc{NExpTime}\xspace}

\renewcommand{\succ}{\mathsf{succ}}
\newcommand{\sig}{\mathsf{sig}}
\newcommand{\ind}{\mathsf{ind}}

\newcommand{\type}{\boldsymbol{t}}

\newcommand{\q}{\boldsymbol q}

\newcommand{\TWAPA}{2APTA\xspace}
\newcommand{\TWAPAs}{2APTAs\xspace}
\newcommand{\TWATA}{2ATA\xspace}
\newcommand{\TWATAs}{2ATAs\xspace}
\newcommand{\TWABA}{2ABTA\xspace}
\newcommand{\TWABAs}{2ABTAs\xspace}

\newcommand{\Start}{\textit{Start}}
\newcommand{\Row}{\textit{Row}}
\newcommand{\End}{\textit{End}}
\newcommand{\first}{\textit{first}}
\newcommand{\halt}{\textit{halt}}

\newcommand{\Mod}{\boldsymbol{M}}

\newcommand{\RCQ}{rCQ\xspace}
\newcommand{\RUCQ}{rUCQ\xspace}

\newcommand{\avec}[1]{\boldsymbol{#1}}

\newcommand{\types}{\mathsf{type}}
\newcommand{\compl}{\mathsf{compl}}

\newcommand{\tup}[1]{\langle #1\rangle}

\newcommand{\ISA}{\sqsubseteq}

\newcommand{\AND}{\sqcap}
\newcommand{\OR}{\sqcup}

\newcommand{\SOME}[2]{\ensuremath{\exists #1 . #2}}

\newcommand{\NOT}{\neg}

\newcommand{\dom}[1][\I]{\Delta^{#1}}
\newcommand{\Int}[2][\I]{#2^{#1}}

\def\spoiler[#1,#2]{(#1 {\,\mapsto\,} #2)}
\def\dupl[#1,#2,#3]{({\text{\footnotesize $#1$}}, \, #2 {\,\leadsto\,} #3)}
\def\spoilernp[#1,#2]{#1\,{\mapsto\,} #2}
\def\duplnp[#1,#2,#3]{{\text{\footnotesize $#1$}}, \, #2 {\,\leadsto\,} #3}
\def\spoilers[#1,#2,#3]{({\text{\footnotesize $#1$}}, \, #2 {\,\mapsto\,} #3)}
\def\spoilersnp[#1,#2,#3]{{\text{\footnotesize $#1$}}, \, #2 {\,\mapsto\,} #3}
\def\spoilerb[#1,#2, #3]{(#1 {\,\mapsto\,} #2, \, {\text{\footnotesize $#3$}})}
\def\spoilerbnp[#1,#2, #3]{#1 {\,\mapsto\,} #2, \,  {\text{\footnotesize $#3$}}}

\tikzset{ %
  point/.style={thick,circle,draw=black,minimum size=1.3mm,inner sep=0pt},%
  constant/.style={fill=black},%
  rowind/.style={fill=gray},%
  cwitness/.style={circle,draw=black,minimum size=3.5mm,inner sep=0pt,
    label=center:{$\land$}},%
  dwitness/.style={minimum width=1.1cm, minimum height=0.3cm},%
  homomorphism/.style={line width=0.1cm,-latex},%
  role/.style={-latex, semithick},%
  gen/.style={role, decorate, decoration={snake, amplitude=0.3mm,segment
      length=2mm, post length=1mm}},%
  subtree/.style={isosceles triangle, draw, very thick, subtreecolor,
    anchor=north, outer sep=0.1, isosceles triangle stretches, minimum
    width=1cm, minimum height=0.6cm, shape border rotate=90}, %
  backward/.style={rectangle, draw=black, fill=gray!50, minimum size=1mm, inner
    ysep=4pt, inner xsep=7pt, outer sep=0.05cm, rounded corners=2mm},%
  sbound/.style={rectangle, draw=black, fill=gray!10, minimum size=1mm, inner
    ysep=4pt, inner xsep=7pt, outer sep=0.05cm},%
  infinite/.style={rectangle, draw=black, fill=gray!25, minimum size=1mm, inner
    ysep=4pt, inner xsep=7pt, outer sep=0.05cm, minimum width=1.5cm},%
  trans/.style={-stealth', semithick},%
}

\colorlet{subtreecolor}{black!70}

\def \tikzdots[#1]{
  \begin{scope}[shift={#1}]
    \node at (0,0.15) {$\cdot$}; \node {$\cdot$}; \node at (0,-0.15) {$\cdot$};
  \end{scope}
}

\begin{document}
\title{Query-Based Entailment and Inseparability for \ALC Ontologies (Full Version)}
\author{Elena Botoeva\\
Free University of Bozen-Bolzano\\
{\normalsize\url{botoeva@inf.unibz.it}}\\
\and 
Carsten Lutz\\
University of Bremen\\
{\normalsize\url{clu@uni-bremen.de}}\\
\and 
Vladislav Ryzhikov\\
Free University of Bozen-Bolzano\\
{\normalsize\url{ryzhikov@inf.unibz.it}}\\
\and 
Frank Wolter\\
University of Liverpool\\
{\normalsize\url{wolter@liverpool.ac.uk}}\\
\and 
Michael Zakharyaschev\\
Birkbeck, University of London\\
{\normalsize\url{michael@dcs.bbk.ac.uk}}
}

\date{}

\maketitle

\begin{abstract}
  We investigate the problem whether two $\mathcal{ALC}$ knowledge
  bases are indistinguishable by queries over a given vocabulary. We
  give model-theoretic criteria in terms of (partial) homomorphisms and products
	and prove that this problem is
  undecidable for conjunctive queries (CQs) but 2\ExpTime-complete 
	for UCQs (unions of CQs). The same results hold if CQs are replaced by rooted CQs. We also consider the
	problem whether two
  $\mathcal{ALC}$ TBoxes give the same answers to any query in a given
  vocabulary over all ABoxes, and show that for CQs this problem is
  undecidable, too, but becomes decidable and 2\ExpTime-complete in
  \hALC, and even \ExpTime-complete in \hALC when restricted to
  (unions of) rooted CQs.
\end{abstract}
%

\ifextendedversion
\begin{figure*}%
\centering\small
\parbox{0.48\linewidth}{%
\begin{tabular}{|@{~}r@{~}|@{~}c@{~}|@{~}c@{~}|c@{~}|@{~}c@{~}|}\hline
Queries
& \ALC
& \begin{tabular}{@{}c@{}}\hALC\\to \ALC\end{tabular}
& \begin{tabular}{@{}c@{}}\ALC to\\\hALC\end{tabular}
& \hALC
\\\hline
CQ   & \multicolumn{2}{c|}{undecidable} & \multirow{2}{*}{\!$\le \! 2\ExpTime$} & \multirow{2}{*}{=$\ExpTime^{(\star)}$} \\\cline{1-3}
UCQ   & \multicolumn{2}{c|}{=2\ExpTime} &  &  \\\hline
\RCQ & \multicolumn{2}{c|}{undecidable} & \multirow{2}{*}{\!$\le \! 2\ExpTime$} & \multirow{2}{*}{=$\ExpTime^{(\star)}$} \\\cline{1-3}
\RUCQ  & \multicolumn{2}{c|}{=$2\ExpTime$} & & \\\hline
\end{tabular}
\caption{KB query entailment.}%
\label{table:kb}}%
\quad~~
\begin{minipage}{0.48\linewidth}%
\begin{tabular}{|@{~}r@{~}|@{~}c@{~}|@{~}c@{~}|c@{~}|@{~}c@{~}|}\hline
Queries
& \ALC
& \begin{tabular}{@{}c@{}}\hALC\\to \ALC\end{tabular}
& \begin{tabular}{@{}c@{}}\ALC to\\\hALC\end{tabular}
& \hALC
\\\hline
CQ   & \multicolumn{2}{c|}{undecidable} & \multirow{2}{*}{?} & \multirow{2}{*}{=$2\ExpTime$} \\\cline{1-3}
UCQ   & \multicolumn{2}{c|}{?} &  &  \\\hline
\RCQ & \multicolumn{2}{c|}{undecidable} & \multirow{2}{*}{=\ExpTime} & \multirow{2}{*}{=\ExpTime} \\\cline{1-3}
\RUCQ  & \multicolumn{2}{c|}{?} & & \\ \hline
\end{tabular}
\caption{TBox query entailment.}%
\label{table:tbox}%
\end{minipage}%
\end{figure*}%
\fi


\section{Introduction}

In recent years, data access using description logic (DL) TBoxes has
become one of the most important applications of
DLs~\cite{PLCD*08,DBLP:conf/rweb/BienvenuO15}, where the underlying
idea is to use a TBox to specify semantics and background knowledge
for the data (stored in an ABox), and thereby derive more complete
query answers.  A major research effort has led to the development of
efficient algorithms and tools for a number of DLs ranging from
DL-Lite \cite{CDLLR07,DBLP:conf/semweb/Rodriguez-MuroKZ13} via more
expressive Horn DLs such as
Horn-\ALC~\cite{DBLP:conf/aaai/EiterOSTX12,DBLP:journals/ws/TrivelaSCS15}
to DLs with all Boolean constructors such as
\ALC~\cite{DBLP:journals/jair/KolliaG13,DBLP:journals/jair/ZhouGNKH15}.

While query answering with DLs is now well-developed, this is much less the
case for reasoning services that support ontology engineering and target query
answering as an application.  In ontology versioning, for example, one would
like to know whether two versions of an ontology give the same answers to all
queries formulated over a given vocabulary of interest, which means that the
newer version can safely replace the older one 
\cite{DBLP:journals/jair/KonevL0W12}. Similarly, if one wants to know whether an ontology can be safely
  replaced by a smaller subset (module), it is the answers to all queries
  that should be preserved \cite{KWZ10}. In this context, the fundamental
relationship between ontologies is thus not whether they are logically
equivalent (have the same models), but whether they give the same answers to
any relevant query. The resulting entailment problem can be formalized in two
ways, with different applications. First, given a class $\mathcal{Q}$ of
queries, knowledge bases (KBs) $\K_{1}$ and
$\K_{2}$, and a signature $\Sigma$ of
relevant concept and role names, we say that $\K_{1}$
\emph{$\Sigma$-$\mathcal{Q}$-entails} $\K_{2}$ if the answers to any
$\Sigma$-query in $\mathcal{Q}$ over $\K_{2}$ are contained in the answers to
the same query over $\K_{1}$. Further, $\K_{1}$ and $\K_{2}$ are
\emph{$\Sigma$-$\mathcal{Q}$-inseparable} if they $\Sigma$-$\mathcal{Q}$-entail
each other.  Note that a KB includes an ABox, and thus this notion of
entailment is appropriate if the data is known and does not change frequently.
Applications include data-oriented KB versioning and KB module extraction, KB
forgetting \cite{DBLP:journals/ci/WangWTPA14}, and knowledge
exchange~\cite{ArenasBCR13}.

If the data is not known or changes frequently, it is not KBs that
should be compared, but TBoxes. Given a pair
$\Theta=(\Sigma_{1},\Sigma_{2})$ that specifies a relevant signature
$\Sigma_{1}$ for ABoxes and $\Sigma_{2}$ for queries, we say that a
TBox $\T_{1}$ \emph{$\Theta$-$\mathcal{Q}$-entails} a TBox $\T_{2}$
if, for every $\Sigma_{1}$-ABox $\A$, the KB $(\T_{1},\A)$
$\Sigma_{2}$-$\mathcal{Q}$-entails $(\T_{2},\A)$.  $\T_{1}$ and
$\T_{2}$ are \emph{$\Theta$-$\mathcal{Q}$-inseparable} if they
$\Theta$-$\mathcal{Q}$-entail each other.  Applications include
data-oriented TBox versioning, TBox modularization and TBox
forgetting~\cite{KWZ10}.

In this paper, we concentrate on the most important choices for \Qmc,
conjunctive queries (CQs) and unions thereof (UCQs); we also consider
the practically relevant classes of rooted CQs (\RCQ{s}) and UCQs
(\RUCQ{s}), in which every variable is connected to an answer variable.
So far, CQ-entailment has been studied for Horn DL
KBs~\cite{BotoevaKRWZ14}, $\mathcal{EL}$
TBoxes~\cite{DBLP:journals/jsc/LutzW10,DBLP:journals/jair/KonevL0W12},
DL-Lite TBoxes~\cite{KontchakovPSSSWZ09}, and also for OBDA
specifications, that is, DL-Lite TBoxes with
mappings~\cite{DBLP:conf/dlog/BienvenuR15}.  No results are available
for non-Horn DLs (neither in the KB nor in the TBox case) and for
expressive Horn DLs in the TBox case. In particular, query entailment
in non-Horn DLs has had the reputation of being a technically challenging
problem.

This paper makes a first breakthrough into understanding query
entailment and inseparability in these cases, with the main results
summarized in Figures~\ref{table:kb} and~\ref{table:tbox} (those marked with $(\star)$ are from \cite{BotoevaKRWZ14}).
Three of them came as a real surprise to us. First, it turned out that CQ- and rCQ-entailment between \ALC KBs is undecidable, even when the first KB is formulated in $\hALC$ (in fact, $\EL$) and without any signature restriction. This should be contrasted with the
decidability of subsumption-based entailment between \ALC TBoxes \cite{GhiLuWo-06}
and of CQ-entailment between \hALC KBs \cite{BotoevaKRWZ14}.
%
%
The second surprising result is that entailment between \ALC KBs
becomes decidable when CQs are replaced with UCQs or \RUCQ{s}. In
fact, we show that entailment is 2\ExpTime-complete for both UCQs and
\RUCQ{s}.  For \ALC TBoxes, CQ- and rCQ-entailment are undecidable as
well. We obtain decidability for \hALC TBoxes (where CQ- und
UCQ-entailments coincide) using the fact that non-entailment is always
witnessed by tree-shaped
ABoxes. 
As another surprise, CQ-entailment of \hALC TBoxes is
2\ExpTime-complete while rCQ-entailment is only
\ExpTime-complete. This should be contrasted with the \EL case, where
both problems are \ExpTime-complete
\cite{DBLP:journals/jsc/LutzW10}. All upper bounds and most lower
bounds hold also for inseparability in place of entailment.
A model-theoretic foundation for these results is a characterization  of query entailment between KBs and TBoxes in terms of (partial) homomorphisms, which, in particular, enables the use
of tree automata techniques to establish the upper bounds in
Figs.~\ref{table:kb} and~\ref{table:tbox}.
%
%

\section{Preliminaries}

Fix lists of \emph{individual names} $a_i$, \emph{concept names} $A_i$, and \emph{role names} $R_i$, for $i < \omega$. $\ALC$-\emph{concepts}, $C$, are defined by the grammar
\begin{equation*}
  C \ \ ::=\ \  A_i \ \ \mid \ \ \top \ \ \mid \ \
  \neg C \ \ \mid \ \ C_1 \sqcap C_2 \  \ \mid \ \
  \exists R_i.C.
\end{equation*}
We use $\bot$, $C_1 \sqcup C_2$ and $\forall R.C$ as abbreviations for $\neg \top$, $\neg(\neg C_1\sqcap \neg C_2)$ and $\neg \exists R.\neg C$, respectively. A \emph{concept inclusion} (CI) takes the form $C \sqsubseteq D$, where $C$ and $D$ are concepts.  An $\ALC$ \emph{TBox} is a finite set of CIs.
In a \emph{\hALC TBox}, no concept of the form $\neg C$ occurs negatively and
no $\exists R.\neg C$ occurs positively~\cite{HuMS05,Kaza09}. An \emph{$\EL$
  TBox} does not contain $\neg$ at all.
An \emph{ABox}, $\A$, is a finite set of \emph{assertions} of the form
$A_k(a_i)$ or $R_k(a_i,a_j)$; $\ind(\A)$ is the set of individual names in
$\A$.  Taken together, $\T$ and $\A$ form a \emph{knowledge base} (KB)
$\K=(\T,\A)$; we set $\ind(\K)=\ind(\A)$.

The semantics 
is defined as usual based on interpretations $\I =
(\Delta^\mathcal{I}, \cdot^\mathcal{I})$ that comply with the
\emph{standard name assumption} in the sense that $a_i^\mathcal{I}
=a_i$~\cite{BCMNP03}. We write $\I \models \alpha$ if an
inclusion or assertion $\alpha$ is true in~$\I$. If
$\I \models \alpha$, for all $\alpha \in \T \cup \A$, then we call
$\I$ a \emph{model} of $\K$ and write $\I\models \K$. $\K$ is
\emph{consistent} if it has a model; we then also say that \Amc
\emph{is consistent with} \Tmc. $\K \models \alpha$ means that
$\I \models \alpha$ for all $\I\models \K$.

A \emph{conjunctive query} (CQ) $\q(\avec{x})$ is a formula $\exists \avec{y}\, \varphi(\avec{x}, \avec{y})$, where $\varphi$ is a conjunction of atoms of the form $A_k(z_1)$ or $R_k(z_1,z_2)$ with $z_i$ in $\avec{x}, \avec{y}$; the variables in $\avec{x}$ are the \emph{answer variables} of $\q(\avec{x})$.  We call $\q$ \emph{rooted} (\RCQ) if every $y \in \avec{y}$ is connected to some $x \in \avec{x}$ by a path in the graph whose nodes are the variables in $\q$ and edges are the pairs $\{u,v\}$ with $R(u,v)\in\q$, for some $R$.
A \emph{union of CQs} (UCQ) is a disjunction $\q(\avec{x}) = \bigvee_i \q_i(\avec{x})$ of CQs $\q_i(\avec{x})$ with the same \emph{answer variables} $\avec{x}$; it is \emph{rooted} (\RUCQ) if all $\q_i$ are rooted.

A tuple $\avec{a}$ in $\ind (\K)$ is a \emph{certain answer to a UCQ $\q(\avec{x})$ over} $\K = (\T,\A)$ if $\I \models \q(\avec{a})$ for all $\I \models \K$; in this case we write $\K \models \q(\avec{a})$. If $\avec{x} = \emptyset$, the answer to $\q$ is `yes' if $\K \models \q$ and `no' otherwise. The problem of checking whether a tuple is a certain answer to a given (U)CQ over a given \ALC KB is known to be \ExpTime-complete for combined complexity~\cite{Lutz-IJCAR08}. The \ExpTime{} lower bound actually holds for \hALC~\cite{DBLP:journals/tocl/KrotzschRH13}.

A set $\Mod$ of models of a KB $\K$ is called \emph{complete for} $\K$ if, for
every UCQ $\q(\avec{x})$, we have $\K \models \q(\avec{a})$ iff $\I \models
\q(\avec{a})$ for all $\I\in \Mod$.  We call an interpretation $\I$ a \emph{ditree interpretation}
if the directed graph $G_\I$ with nodes $d \in \Delta^{\I}$ and edges $(d,e)\in
R^{\I}$, for some $R$, is a tree and $R^{\I}\cap S^{\I}=\emptyset$, for any
distinct roles $R$ and $S$.  $\I$ has \emph{outdegree} $n$ if $G_{\I}$ has
outdegree $n$.  A model $\I$ of a KB $\K=(\T,\A)$ is \emph{forest-shaped} if
$\I$ is the disjoint union of ditree interpretations $\I_{a}$ with root $a$, for
$a\in \ind(\A)$, extended with all $R(a,b)\in \A$. The \emph{outdegree} of $\I$
is the maximum outdegree of the $\I_{a}$. It is well known that the class
$\Mod^{\it fo}_\K$ of all forest-shaped models of an \ALC KB $\K$ of outdegree
bounded by $|\T|$ is complete for $\K$~\cite{Lutz-IJCAR08}. If $\K$ is a \hALC
KB, then a single member $\I_{\K}$ of $\Mod^{\it fo}_\K$ is complete for $\K$.
$\I_{\K}$ is constructed using the standard chase procedure and called the
\emph{canonical model of $\K$}.

A \emph{signature}, $\Sigma$, is a set of concept and role names. By a
$\Sigma$-concept, $\Sigma$-CQ, etc.\ we understand any concept, CQ,
etc.\ constructed using the names from $\Sigma$. We say that $\Sigma$
is
\emph{full} if it contains all concept and role names.
%
A model $\I$ of a KB $\K$ is \emph{$\Sigma$-connected} if, for any  $u \in \Delta^\I \setminus \ind(\K)$, there is a path $R_1^\I(a,u_1),\dots, R_n^\I(u_n,u)$ with $a \in \ind(\K)$ and the $R_i$ in $\Sigma$.
\begin{definition}\em
Let $\K_{1}$ and $\K_{2}$ be consistent KBs, $\Sigma$ a signature, and $\mathcal{Q}$ one of CQ, \RCQ, UCQ or \RUCQ. We say that $\K_1$ \emph{$\Sigma$-$\mathcal{Q}$-entails} $\K_2$ if
$\K_2\models \q(\avec{a})$ implies $\avec{a} \subseteq \ind(\K_1)$ and \mbox{$\K_1\models \q(\avec{a})$}, for all $\Sigma$-$\mathcal{Q}$ $\q(\avec{x})$ and all tuples $\avec{a}$ in $\ind(\K_2)$.
$\K_1$ and $\K_2$ are \emph{$\Sigma$-$\mathcal{Q}$ inseparable} if
they $\Sigma$-$\mathcal{Q}$ entail each other.
%
%
\end{definition}
As larger classes of queries separate more KBs, $\Sigma$-UCQ inseparability implies all other inseparabilities. The following example shows that, in general, no other implications between the different notions of inseparability hold for $\ALC$.
\begin{example}\label{ex:ex1}\em
  Suppose $\T_{0}=\emptyset$, $\T_{0}'= \{ E \sqsubseteq A \sqcup B\}$ and
  $\Sigma_{0}= \{A,B,E\}$.  Let $\A_{0}=\{ E(a)\}$, $\K_{0}=(\T_{0},\A_{0})$,
  and $\K_{0}'=(\T_{0}',\A_{0})$. Then $\K_{0}$ and $\K_{0}'$ are
  $\Sigma_{0}$-CQ inseparable but not $\Sigma_{0}$-\RUCQ inseparable. In fact,
  $\K_{0}'\models \q(a)$ and $\K_{0}\not\models \q(a)$ for $\q(x)= A(x)\vee B(x)$.

  Now, let $\Sigma_{1}= \{ E,B\}$, $\T_{1}=\emptyset$, and $\T_{1}'= \{ E
  \sqsubseteq \exists R.B\}$.  Let $\A_{1}=\{ E(a)\}$, $\K_{1}=(\T_{1},\A_{1})$,
  and $\K_{1}'=(\T_{1}',\A_{1})$.  Then $\K_{1}$ and $\K_{1}'$ are
  $\Sigma_{1}$-rUCQ inseparable but not $\Sigma_{1}$-CQ inseparable. In fact,
  $\K_{1}'\models \exists x B(x)$ but $\K_{1}\not\models \exists x
  B(x)$. 
\end{example}
\begin{definition}\em
\label{def:tboxinsep}
Let $\T_{1}$ and $\T_{2}$ be TBoxes, $\mathcal{Q}$ one of CQ, \RCQ, UCQ or \RUCQ, and let $\Theta = (\Sigma_1,\Sigma_2)$ be a pair of signatures. We say that $\T_{1}$ \emph{$\Theta$-$\mathcal{Q}$ entails} $\T_{2}$ if, for every $\Sigma_1$-ABox $\A$ that is consistent with both $\T_1$ and $\T_2$, the KB $(\T_1,\A)$ $\Sigma_2$-$\mathcal{Q}$ entails the KB $(\T_2,\A)$. $\T_1$ and $\T_2$ are \emph{$\Theta$-$\mathcal{Q}$ inseparable} if they $\Theta$-$\mathcal{Q}$ entail each other. If $\Sigma_{1}$ is the set of all concept and role names, we say `\emph{full ABox signature $\Sigma_{2}$-$\mathcal{Q}$ entails}' or `\emph{full ABox signature $\Sigma_{2}$-$\mathcal{Q}$ inseparable}'\!.
\end{definition}
We only consider ABoxes that are consistent with both TBoxes because   the problem whether a $\Sigma_{1}$-ABox
consistent with $\T_{2}$ is also consistent with $\T_{1}$ is well understood: it is mutually polynomially reducible with the containment problem for
ontology-mediated queries with CQs of the form $\exists x A(x)$, which is \NExpTime-complete for \ALC and \ExpTime-complete for \hALC~\cite{DBLP:journals/tods/BienvenuCLW14,HLSW-IJCAI16}.
%
\begin{example}\label{andrea}\em
Consider the TBoxes $\T_{0}$ and $\T_{0}'$ from Example~\ref{ex:ex1} and let $\Theta=(\Sigma,\Sigma)$ for $\Sigma=\{R,A,B,E\}$.
Then $\T_{0}$ does not $\Theta$-rCQ entail $\T_{0}'$ as $(\T_{0}',\A) \models \q(a)$ and $(\T_{0},\A) \not\models \q(a)$ for
\ifextendedversion
\begin{center}
\else
\\[2mm]
\fi
\begin{tikzpicture}\scriptsize
  \foreach \name/\x/\y/\conc/\wh in {%
    a/-1.6/0.5//right,%
    b/-0.8/-0.1/A/below, %
    c/0/0.5/E/below, %
    d/1.2/0.5/B/below%
  }{ \node[inner sep=1, outer sep= 0, label=\wh:{$\conc$}] (\name) at (\x,\y)
    {$\name$}; }

  \foreach \from/\to/\wh in {%
    a/b/below, a/c/below, b/c/below, c/d/below%
  }{ \draw[role] (\from) -- node[\wh,sloped] {$R$} (\to); }

  \node[anchor=east] at (-2,0.2) {\normalsize $\A$:};

  \begin{scope}[xshift=4.5cm, yshift=0.2cm]
    \foreach \name/\x/\conc/\wh in {%
      x/-1.2//right,%
      y_1/0/A/below, %
      y_2/1.2/B/below%
    }{ \node[inner sep=1, outer sep=0, label=\wh:{$\conc$}] (\name) at
      (\x,0) {$\name$ }; }

    \foreach \from/\to/\wh in {%
      x/y_1/right, y_1/y_2/right%
    }{ \draw[role] (\from) -- node[below] {$R$} (\to); }

    \node[anchor=east] at (-1.6,0) {\normalsize $\q(x)$:};
  \end{scope}
\end{tikzpicture}
\ifextendedversion
\end{center}
\fi
\end{example}
We observe that $\Theta$-CQ-entailment in the restricted case with $\Theta=(\Sigma,\Sigma)$ has been investigated for $\mathcal{EL}$ TBoxes by Lutz and Wolter~[\citeyear{DBLP:journals/jsc/LutzW10}] and Konev et al.~[\citeyear{DBLP:journals/jair/KonevL0W12}].

As in the KB case, $\Sigma$-UCQ inseparability of $\ALC$ TBoxes implies all other types of inseparability, and Example~\ref{ex:ex1} can be
used to show that no other implications hold in general. The situation changes for \hALC KBs and TBoxes. The following can be proved by observing that a
\hALC KB entails a UCQ iff it entails one of its disjuncts:

\begin{theorem}\label{UCQtoCQ}
Let $\K_{1}$ be an $\ALC$ KB and $\K_{2}$ a \hALC KB. Then
$\K_{1}$ $\Sigma$-UCQ entails $\K_{2}$ iff $\K_{1}$ $\Sigma$-CQ entails $\K_{2}$. The same holds for \RUCQ and \RCQ, and for TBox entailment.
\end{theorem}

%

\section{Model-Theoretic Criteria for \ALC KBs}

We now give model-theoretic criteria for 
$\Sigma$-entailment between KBs. 
The \emph{product} $\prod \avec{\I}$ of a set $\avec{\I}$ of interpretations is defined as usual in model theory~\cite[page~405]{Chang&Keisler90}.
%
Note that, for any CQ $\avec{q}(\avec{x})$ and any tuple $\avec{a}$ of individual names, $\prod \avec{\I} \models \avec{q}(\avec{a})$ iff $\I \models \avec{q}(\avec{a})$ for each $\I \in \avec{\I}$.


Suppose $\I_i$ is an interpretation for a KB $\K_i$, $i=1,2$.  A
function $h \colon \Delta^{\I_2} \to \Delta^{\I_1}$ is called a
\emph{$\Sigma$-homomorphism} if $u \in A^{\Imc_2}$ implies $h(u) \in
A^{\Imc_1}$ and $(u,v) \in R^{\Imc_2}$ implies $(h(u),h(v)) \in
R^{\Imc_1}$ for all $u,v \in \Delta^{\smash{\I_2}}$, $\Sigma$-concept
names $A$, and $\Sigma$-role names $R$, and $h(a) = a$ for all $a \in
\ind(\K_2)$. 
%
It is known from database theory that homomorphisms characterize CQ-containment~\cite{DBLP:conf/stoc/ChandraM77}. For KB $\Sigma$-query entailment, finite partial homomorphisms are required.
We say that $\I_2$ is \emph{$n\Sigma$-homo\-mo\-rphically embeddable into $\I_1$} if, for any subinterpretation $\I_2'$ of $\I_2$ with $|\Delta^{\smash{\I'_2}}| \le n$, there is a $\Sigma$-homomorphism from $\I_2'$ to $\I_1$.
If, additionally, we require $\I_2'$ to be $\Sigma$-connected then $\I_2$ is said to be \emph{con-$n\Sigma$-homo\-mo\-rphically embeddable into $\I_1$}.

\begin{theorem}\label{crit:KB}
Let $\K_1$ and $\K_2$ be \ALC KBs, $\Sigma$ a signature, and let $\Mod_{\!i}$ be complete for $\K_i$, $i=1,2$.
\begin{description}\itemsep=0pt
\item[\rm (1)] $\K_{1}$ $\Sigma$-UCQ entails $\K_2$ iff, for any  $n>0$ and $\I_1\in \Mod_{\!1}$, there exists $\I_2 \in \Mod_{\!2}$ that is $n\Sigma$-homomorphically embeddable into $\I_1$.
	
\item[\rm (2)] $\K_{1}$ $\Sigma$-rUCQ entails $\K_2$ iff, for any  $n>0$ and $\I_1\in \Mod_{\!1}$, there exists $\I_2 \in \Mod_{\!2}$ that is con-$n\Sigma$-homomorphically embeddable into $\I_1$.
	
\item[\rm (3)] $\K_{1}$ $\Sigma$-CQ entails $\K_2$ iff $\prod \Mod_{\!2}$ is $n\Sigma$-homomorphically embeddable into $\prod \Mod_{\!1}$ for any $n>0$.
	
\item[\rm (4)] $\K_{1}$ $\Sigma$-rCQ entails $\K_2$ iff $\prod \Mod_{\!2}$ is con-$n\Sigma$-homo\-mor\-phically embeddable into $\prod \Mod_{\!1}$ for any $n>0$.
\end{description}
\end{theorem}
%
%
%
\begin{proof}
We only show (1). Suppose $\K_{2}\models \q$ but $\K_{1}\not\models \q$. Let $n$ be the number of variables in $\q$. Take $\I_{1}\in \Mod_{\!1}$ such that $\I_{1}\not\models \q$. Then no $\I_{2}\in \Mod_{\!2}$ is $n\Sigma$-homomorphically embeddable into $\I_{1}$.
Conversely, suppose $\I_{1}\in \Mod_{\!1}$ is such that, for some $n$,  no $\I_{2}\in \Mod_{\!2}$ is $n\Sigma$-homomorphically embeddable into $\I_{1}$. We can regard any subinterpretation of any $\I_{2}\in \Mod_{\!2}$ with domain of size $\le n$ as a CQ (with answer variable corresponding to ABox individuals). The disjunction of all such CQs is entailed by $\K_{2}$ but not by $\K_{1}$.
%
\end{proof}

Note that $n\Sigma$-homomorphic embeddability cannot be replaced by $\Sigma$-homomorphic embeddability. For example, in~(1), let $\K_{1}=\K_{2}=(\{\top \sqsubseteq \exists R.\top\},\{A(a)\})$, $\Mod_{\!1}=\{\I_{1}\}$, where $\I_{1}$ is the infinite $R$-chain starting with $a$, and let $\Mod_{\!2}$ contain arbitrary finite $R$-chains starting with $a$ followed by an arbitrary long $R$-cycle. $\Mod_{\!1}$ and $\Mod_{\!2}$ are
both complete for $\K$, but there is no $\Sigma$-homomorphism from any $\I_{2}\in \Mod_{\!2}$ to $\I_{1}$. In Section~\ref{rUCQ}, we show that in some cases
we \emph{can} find characterizations with full $\Sigma$-homomorphisms and use them to present decision procedures for entailment.

If both $\Mod_{\!i}$ are finite and contain only finite interpretations, then Theorem~\ref{crit:KB} provides a decision procedure for KB entailment. This applies, for example, to KBs with acyclic classical TBoxes~\cite{BCMNP03}, and to KBs for which the chase terminates \cite{DBLP:journals/jair/GrauHKKMMW13}.

%


\section{Undecidability for \ALC KBs and TBoxes}

We show that CQ and rCQ-entailment and inseparability for \ALC KBs are undecidable---even if the signature is full and $\K_1$ is a
\hALC (in fact, $\mathcal{EL}$) KB.
We establish the same results for TBoxes except that in the rCQ case, we
leave it open whether the full ABox signature is sufficient for
undecidability. 
%
%
\begin{theorem}\label{thm:undecidability}
\textup{(}i\textup{)} The problem whether a \hALC KB $\Sigma$-$\mathcal{Q}$ entails an $\ALC$ KB is undecidable for $\mathcal{Q} \in \{\text{CQ},\text{\RCQ}\}$.

\textup{(}ii\textup{)} $\Sigma$-$\mathcal{Q}$ inseparability between \hALC and \ALC KBs is undecidable for $\mathcal{Q} \in \{\text{CQ},\text{\RCQ}\}$.

\textup{(}iii\textup{)} Both \textup{(}i\textup{)} and \textup{(}ii\textup{)} hold for the full signature $\Sigma$.
\end{theorem}
\begin{proof}
\begin{figure*}
  \centering
  \scalebox{0.82}{






\begin{tikzpicture}[xscale=1, %
  point/.style={thick,circle,draw=black,fill=white, minimum
    size=1.3mm,inner sep=0pt}%
  ]
  \foreach \al/\x/\y/\lab/\wh/\extra in {%
    a/0.2/0/A/below/constant, %
    start/1/0/\Start/below/, %
    x11/2/0/{I_0}/below/, %
    xN1/4/0/{T_0^{N1}}/below/, %
    row1/5/0/\Start/below/rowind, %
    x12/6/0/{T_1^{12}}/below/, %
    xN2/8/0/{T_1^{N2}}/below/, %
    row2/9/0/{\Row}/above/rowind,%
    x1m/10.5/0/{T_1^{1M\text{-}1}}/below/, %
    xNm/12.5/0/{T_1^{NM\text{-}1}}/below/, %
    rowm/13.5/0/{\Row}/above/rowind,%
    endl/14.5/0.6/{~~\End}/below/,%
    x1M/14.5/-0.6/{T_2^{1M}}/below/, %
    xNM/16.5/-0.6/{T_2^{NM}}/below/, %
    rowM/17.5/-0.6/{\Row}/above/rowind,%
    endr/18.3/-0.6/{\End}/below/%
  }{ \node[point, \extra, label=\wh:{\footnotesize $\lab$}] (\al) at (\x,\y)
    {}; }

  \node[label={[inner ysep=1]above:{\footnotesize$\Row~~~~~~~$}}] at (row1) {};%

  \node[label={[inner xsep=1]left:{\footnotesize$a$}}] at (a) {};%

  \foreach \al/\lab/\wh in {%
    x12/{$I_0$\dots}/1, %
    xN2/{$T_0^{N1}$\dots}/1, %
    x1m/{$T_0^{1M\text{-}2}$\dots}/1, %
    xNm/{$T_0^{NM\text{-}2}$\dots}/1, %
    x1M/{$T_1^{1M\text{-}1}$\dots}/-1, %
    xNM/{$T_1^{NM\text{-}1}$\dots}/-1%
  }{ \node[label={[scale=0.8]center:{\footnotesize\lab}}, yshift=-0.7cm] at
    (\al) {}; }

  \foreach \from/\to/\type/\extra in {%
    start/x11/role, x11/xN1/dotted, xN1/row1/role,%
    row1/x12/role, x12/xN2/dotted, xN2/row2/role,%
    row2/x1m/dotted, x1m/xNm/dotted, xNm/rowm/role,%
    rowm/x1M/role, x1M/xNM/dotted, xNM/rowM/role,%
    rowm/endl/role, rowM/endr/role%
  } {\draw[thick, \type, \extra] (\from) -- (\to);}

  \draw[role] (a) -- node[above] {\footnotesize$P$} (start); %

  \foreach \from/\to/\type in {%
    rowm/endl, rowm/x1M%
  } {\draw[draw=none] (\from) -- node[pos=0.5, inner sep=0] (\from-\to) {} (\to);}

  \draw[red] (rowm-endl) to[bend left] node[left] {\footnotesize $\lor$}
  (rowm-x1M);

  \node[yshift=0.1cm, xshift=-0.5cm] at (endl) {\Large$\I_l$};
  \node[yshift=-0.1cm, xshift=-0.6cm] at (x1M) {\Large$\I_r$};

  \begin{pgfonlayer}{background}
    \foreach \first/\last in {%
      x11/row1, x12/row2, x1m/rowm, x1M/rowM%
    }{ \node[fit=(\first)(\last), rounded corners, fill=gray!20] {}; }
  \end{pgfonlayer}

  \begin{pgfonlayer}{background}
    \def\drawsubtree[#1,#2]{%
      \draw[secondary, thick, dotted] (#1) -- ++(0.5,0.1) (#1) -- ++(0.5,-0.1);%
    }
    
    \colorlet{lightgray}{gray!70}
    \begin{scope}[every node/.style={scale=0.7}, secondary/.style={minimum
        size=0.8mm, draw=lightgray, thin}]\footnotesize%
      \coordinate (xk1) at (2.8,0);%
      \coordinate (xk2) at (7,0);%
      \coordinate (xkm) at (11.5,0);%
      \coordinate (xkM) at (15.5,-0.75);

      \foreach \al/\tile in {%
        x11/T_0, xk1/T_0, xN1/T_0, row1/T_1, x12/T_1, xN2/T_1, xNm/T_1,
        xkM/T_2, xNM/T_2%
      }{ \node[point, secondary, xshift=0.6cm, yshift=0.8cm,
        label={[lightgray, inner sep=1]above:{$\tile$}}] (\al-z1) at (\al) {};

        \draw[role, lightgray, secondary] (\al) -- (\al-z1);

        \drawsubtree[\al-z1,1]
      }

      \foreach \al in {%
        row2, xk2, x1m, xkm, x1M%
      }{ \draw[dotted, lightgray, thick] (\al) -- +(0.5,0.6); }

      \foreach \al/\tile/\ysh in {%
        xN2/T_2/1.6cm%
      }{ \node[point, secondary, fill=lightgray!50, xshift=0.5cm, yshift=\ysh,
        label={[lightgray, inner sep=1]left:{$\Row$}}] (\al-halt) at (\al) {};

        \node[point, secondary, xshift=0.7cm, yshift=0.3cm,
        label={[lightgray, inner sep=1]right:{$\End$}}] (\al-halt-end) at (\al-halt) {};

        \node[point, secondary, xshift=0.7cm, yshift=-0.3cm,
        label={[lightgray, inner sep=1]above:{$\tile$}}] (\al-halt-tile) at (\al-halt) {};

        \foreach \from/\to in {%
          \al/\al-halt, \al-halt/\al-halt-end, \al-halt/\al-halt-tile%
        }{ \draw[role, lightgray, secondary] (\from) -- (\to); }

        \foreach \from/\to in {%
          \al-halt/\al-halt-end, \al-halt/\al-halt-tile%
        }{ \draw[draw=none] (\from) -- node[inner sep=0, outer sep=0, pos=0.55]
          (\from-\to) {} (\to); }

        \draw[lightgray!50!red] (\al-halt-\al-halt-end) to[bend left]
        node[left, inner sep=2, scale=0.7] {$\lor$} (\al-halt-\al-halt-tile);

        \drawsubtree[\al-halt-tile,1]%
      }

      \foreach \al/\ysh in {%
        xNm/1.6cm%
      }{ \node[point, secondary, fill=lightgray!50, xshift=0.5cm, yshift=\ysh,
        label={[lightgray, inner sep=1]left:{$\Row$}}] (\al-row) at (\al) {};

        \draw[role, lightgray, secondary] (\al) -- (\al-row); 

        \drawsubtree[\al-row,1]%
      }

      \foreach \al/\ysh in {%
        xk1/1.6cm%
      }{ \node[point, secondary, fill=lightgray!50, xshift=0.6cm, yshift=\ysh,
        label={[lightgray,text width=0.65cm]left:{$\Row$~~\\$\Start$}}]
        (\al-start) at (\al) {};

        \draw[role, lightgray, secondary] (\al) -- (\al-start); 

        \drawsubtree[\al-start,1]%
      }

      \foreach \al/\tile/\ysh in {%
        xkM/T_2/1.6cm%
      }{ \node[point, secondary, fill=lightgray!50, xshift=0.5cm, yshift=\ysh,
        label={[lightgray, inner sep=1]left:{$\Row$}}] (\al-end-row) at (\al) {};

        \node[point, secondary, xshift=0.7cm,
        label={[lightgray, inner sep=1]right:{$\End$}}] (\al-end) at (\al-end-row) {};

        \foreach \from/\to/\type in {%
          \al/\al-end-row, \al-end-row/\al-end%
        }{ \draw[role, lightgray, secondary] (\from) -- (\to); } %
      }
    \end{scope}    
  \end{pgfonlayer}

  \begin{scope}[xshift=3cm, yshift=2cm]
    \foreach \al/\x/\y/\lab/\wh/\extra in {%
      y0/0/0/{\Start}/right/, %
      y1/1/0/{B_1}/right/, %
      yN/3/0/{B_N}/right/, %
      yN1/4/0/{B_{N+1}}/right/rowind, %
      y21/5/0/{}/right/, %
      y2N/7/0/{}/right/, %
      y2N1/8/0/{}/right/rowind, %
      yl1/9.5/0/{B_{n{-}N}}/right/, %
      ylN/11.5/0/{B_{n{-}1}}/right/, %
      yl/12.5/0/{B_n}/right/rowind, %
      yend/13.5/0/\End/right/%
    }{ \node[point, \extra, label=above:{\footnotesize $\lab$}] (\al) at
      (\x,\y) {}; }

    \foreach \from/\to/\type in {%
      y0/y1/role, y1/yN/dotted, yN/yN1/role, yN1/y21/role, y21/y2N/dotted,
      y2N/y2N1/role, y2N1/yl1/dotted, yl1/ylN/dotted, ylN/yl/role, yl/yend/role%
    } {\draw[thick, \type] (\from) -- (\to);}

    \begin{pgfonlayer}{background}
      \foreach \first/\last in {%
        y1/yN1, y21/y2N1, yl1/yl%
      }{ \node[fit=(\first)(\last), rounded corners, fill=gray!20] {}; }
    \end{pgfonlayer}

    \node[xshift=-1cm] at (y0) {\Large$\q_n$};
  \end{scope}

  \begin{pgfonlayer}{background}

    \foreach \al in {y0,yend,start,row1,endl,endr} { %
      \node[fit=(\al), inner sep=3] (\al-a) {}; }

    \foreach \from/\to/\out/\in in {%
      y0-a/start-a/250/70, yend-a/endl-a/250/30
    } {\draw[homomorphism, black!50] (\from) to[out=\out, in=\in]
      node[above] {$h_l$} (\to);}

    \foreach \from/\to/\in in {%
      y0-a/row1-a/90, yend-a/endr-a/110
    } {\draw[homomorphism, black!75] (\from) to[out=290, in=\in]
      node[above] {$h_r$} (\to);}
  \end{pgfonlayer}
\end{tikzpicture}

}
  \caption{The structure of models $\I_l$ and $\I_r$ of $\K_2$, and
    homomorphisms $h_l \colon \q_n \to \I_l$ and $h_r \colon \q_n \to \I_r$.}
  \label{fig:k2}
\end{figure*}
The proof is by reduction of the undecidable \mbox{$N\times M$}-\emph{tiling problem}: given a finite set $\mathfrak T$ of \emph{tile types} $T$ with four colours $\textit{up}(T)$, $\textit{down}(T)$, $\textit{left}(T)$ and $\textit{right}(T)$, a tile type \mbox{$I \in \mathfrak T$}, and two colours $W$ (for wall) and $C$ (for ceiling), decide whether there exist $N,M \in \mathbb N$ such that the $N \times M$ grid can be tiled using $\mathfrak T$ in such a way that $(1,1)$ is covered by a tile of type $I$; every $(N,i)$, for $i \le M$, is covered by a tile of type $T$ with $\textit{right}(T) = W$; and every $(i,M)$, for $i \le N$, is covered by a tile of  type $T$ with $\textit{up}(T) = C$.

Given an instance of this problem, we first describe a KB $\K_2 = (\T_2, \{A(a)\})$ that uses (among others) 3 concept names $T_k$, $k=0,1,2$, for each tile type $T \in \mathfrak T$.  If a point $x$ in a model $\I$ of $\K_2$ is in $T_k$ and $\textit{right}(T) = \textit{left}(T')$, then $x$ has an $R$-successor in $T'_k$. Thus, branches of $\I$ define (possibly infinite) horizontal rows of tilings with $\mathfrak T$. If a branch contains a point $y \in T_k$ with $\textit{right}(T) = W$, then this $y$ can be the last point in the row, which is indicated by an $R$-successor $z \in \Row$ of $y$. In turn, $z$ has $R$-successors in all $T_{(k+1)\, \text{mod}\, 3}$ that can be possible beginnings of the next row of tiles. To coordinate the $\textit{up}$ and $\textit{down}$ colours between the rows---which will be done by the CQs separating $\K_1$ and $\K_2$--- we make every $x \in T_k$, starting from the second row, an instance of all $T'_{(k-1)\, \text{mod}\, 3}$ with $\textit{down}(T) = \textit{up}(T')$. The row started by $z \in \Row$ can be the last one in the tiling, in which case we require that each of its tiles $T$ has $\textit{up}(T) = C$. After the point in $\Row$ indicating the end of the final row, we add an $R$-successor in $\End$ for the end of tiling. The beginning of the first row is indicated by a $P$-successor in $\Start$ of the ABox element $a$, after which we add an $R$-successor in $I_0$ for the given initial tile type $I$; see the lowest branch in Fig.~\ref{fig:k2}. To generate a tree with all possible branches described above, we only require $\EL$ axioms of the form $E \ISA D$ and $E \ISA \exists S.D$.

The existence of a tiling of some $N \times M$ grid for the given instance can be checked by Boolean CQs $\q_n$ that require an $R$-path from $\Start$ to $\End$ going through $T_k$- or $\Row$-points:
\begin{equation*}
\exists \avec{x} \big( \Start(x_0) \land \bigwedge_{i=0}^n R(x_i, x_{i+1}) \land \bigwedge_{i=1}^n B_{i}(x_{i}) \land \End(x_{n+1}) \big)
\end{equation*}
with $B_{i} \in \{\Row\} \cup \{T_k \mid T \in \mathfrak{T}, k=0,1,2\}$; see Fig.~\ref{fig:k2}. The key trick is---using an axiom of the form $D \ISA E \sqcup E'$---to ensure that the $\Row$-point before the final row of the tiling has \emph{two alternative} continuations: one as described above, and the other one having just a single $R$-successor in $\End$; see Fig.~\ref{fig:k2} where $\lor$ indicates an \emph{or-node}. This or-node gives two models of $\K_2$ denoted $\I_l$ and $\I_r$ in the picture. If $\K_2 \models \q_n$, then $\q_n$ holds in both of them, and so there are homomorphisms $h_l \colon \q_n \to \I_l$ and $h_r \colon \q_n \to \I_r$. As  $h_l(x_{n-1})$ and $h_r(x_{n-1})$ are instances of $B_{n-1}$, we have $B_{n-1} = T_1^{N M-1}$ in the picture, and so $\textit{up}(T^{N M-1}) = \textit{down}(T^{N M})$. By repeating this argument until  $x_0$, we see that the colours between horizontal rows match and the rows are of the same length. (For this trick to work, we have to make the first $\Row$-point in every branch an instance of $\Start$.) In fact, we have:
\begin{lemma}
An instance of the $N \times M$-tiling problem has a positive answer iff there exists $\q_n$ such that $\K_2 \models \q_n$.
\end{lemma}

It is to be noted that to construct $\T_2$ with the properties described above one needs quite a few auxiliary concept names.

Next, we define $\K_1 = (\T_1,\{ A(a)\})$ to be the \EL KB with the following canonical model:
\ifextendedversion
\begin{center}
\else
\\\hspace*{0.2cm}
\fi
  \begin{tikzpicture}[xscale=1.3, yscale=1, %
    emptyind/.style={fill=gray!50, inner sep=1.7}]

    \foreach \al/\x/\y/\lab/\wh/\extra/\rot in {%
      a/0.5/0/A/above/constant/0, %
      x1/1.5/0/{\Start,\Sigma_0}/above//0, %
      x2/1.5/-0.5/{}/below/emptyind/0, %
      x3/2.2/-0.75/{\End,\Sigma_{0}~}/below//-14, %
      x4/2.9/-1/{~\End,\Sigma_{0}}/below//-14, %
      y1/3/0/{\Start,\Sigma_0}/above//0,%
      y2/3/-0.5/{}/above/emptyind/0,%
      y3/3.7/-0.75/{\End,\Sigma_{0}~}/below//-14,%
      y4/4.4/-1/{~\End,\Sigma_{0}}/below//-14,%
      z1/4.5/0/{\Start,\Sigma_0}/above//0,%
      z2/4.5/-0.5/{}/above/emptyind/0,%
      z3/5.2/-0.75/{\End,\Sigma_{0}~}/below//-14,%
      z4/5.9/-1/{~\End,\Sigma_{0}}/below//-14%
    }{ \node[point, \extra, label={[rotate=\rot, inner sep=2]\wh:{\scriptsize
          $\lab$}}] (\al) at (\x,\y) {}; }

    \node[label={[inner xsep=1]left:{\scriptsize$a$}}] at (a) {};%

    \foreach \from/\to/\type in {%
      x1/x2/role, x2/x3/role, x3/x4/role, %
      x1/y1/role, y1/y2/role, y2/y3/role, y3/y4/role, %
      y1/z1/role, z1/z2/role, z2/z3/role, z3/z4/role%
    } {\draw[role, \type] (\from) -- (\to);}

    \draw[role] (a) -- node[below] {\scriptsize$P$} (x1); %

    \foreach \from in {%
      x4, y4, z4
    } {\draw[dotted, thick] (\from) -- +(0.5,-0.2);}

    \draw[dotted, thick] (z1) -- +(1,0);
  \end{tikzpicture}
\ifextendedversion
\end{center}
\else
\\[-2pt]
\fi
where $\Sigma_0 = \{\textit{Row}\} \cup \{T_k \mid T \in \mathfrak T, \, k=0,1,2\}$. Note that the vertical $R$-successors of the $\Start$-points are not instances of any concept name, and so $\K_1$ does not satisfy any query $\q_n$. On the other hand, $\K_2 \models \q$ implies $\K_1 \models \q$, for every $\Sigma$-CQ $\q$ without a subquery of the form $\q_n$ and $\Sigma={\sf sig}(\K_{1})$.

This proves (\emph{i}) for $\Sigma$-CQ entailment. For $\Sigma$-\RCQ{} entailment, we slightly modify the construction, in particular, by adding $R(a,a)$ and $\textit{Row}(a)$ to the ABox $\{A(a)\}$, and a conjunct $R(y,x_0)$ with a free $y$ to $\q_n$. (The loop $R(a,a)$ plays roughly the same role as the path between two $\Start$-points in Fig.~\ref{fig:k2}.) To prove~(\emph{ii}), we take $\K_2' = \K_2 \cup \K_1$ and show that $\K_1$ $\Sigma$-CQ entails $\K_2$ iff $\K_1$ and $\K_2'$ are $\Sigma$-CQ inseparable.
Finally, we prove~(\emph{iii}) by replacing non-$\Sigma$ symbols in $\K_{2}$ with complex \ALC-concepts that cannot be used in CQs and extending the TBoxes appropriately; cf.~\cite[Lemma~21]{lutz2012non}.
\end{proof}

The TBoxes from the proof above can also be used to obtain 
\begin{theorem}\label{thm:TBoxundecidability}
\textup{(}i\textup{)} The problem whether a \hALC TBox $\Theta$-$\mathcal{Q}$ entails an $\ALC$ TBox is undecidable for $\mathcal{Q} \in \{\text{CQ},\text{\RCQ}\}$.

\textup{(}ii\textup{)} $\Theta$-$\mathcal{Q}$ inseparability between \hALC and \ALC TBoxes is undecidable for $\mathcal{Q} \in \{\text{CQ},\text{\RCQ}\}$.

\textup{(}iii\textup{)} For CQs, \textup{(}i\textup{)} and \textup{(}ii\textup{)} hold for full ABox signatures and for $\Theta = (\Sigma_1,\Sigma_2)$ with $\Sigma_{1}=\Sigma_{2}$.
\end{theorem}

Observe that our undecidability proof does not work for UCQs as the UCQ composed of the two disjunctive branches shown in  Fig.~\ref{fig:k2} (for non-trivial instances) distinguishes between the KBs independently of the existence of a tiling. We now show that, at least for \RUCQ{s}, entailment is decidable.


\newcommand{\fosh}{forest-shaped\xspace}

\section{UCQ-Entailment for \ALC-KBs}\label{rUCQ}

Theorem~\ref{thm:undecidability} might seem to suggest that any reasonable
notion of query inseparability is undecidable for \ALC KBs. Interestingly, this
is not the case: we show now that UCQ-entailment and rUCQ-entailment are
2\ExpTime-complete.  We discuss the UCQ case first. The upper bound proof
consists of two parts. First, given KBs $\K_1 = (\T_{1},\A_1)$ and
$\K_2=(\T_{2},\A_{2})$ and a signature $\Sigma$ we construct a tree automaton
$\mathfrak{A}$ such that $\L(\mathfrak{A})$ is non-empty iff

\medskip
$(\ast)$ there exists a forest-shaped model $\I_{1}$ of $\K_{1}$ of outdegree at most $|\T_{1}|$ such that no model of $\K_{2}$ is $\Sigma$-homorphically embeddable into $\I_{1}$. 

\medskip
\noindent 
Using complexity results for the emptiness problem for tree automata we obtain a 2\ExpTime prodecure that checks $(\ast)$. We then show in a second step that $(\ast)$ holds iff
$\K_{1}$ does not $\Sigma$-UCQ entail $\K_{2}$ and thus prove the 2\ExpTime upper bound \emph{and} a new homomorphism-based criterion for $\Sigma$-UCQ entailment. 
In the proof we use the model-theoretic criterion given in Theorem~\ref{crit:KB} (1) and also use the well known result from automata theory that non-empty languages accepted
by tree automata contain a regular tree to derive that there exists some model $\I_{1}$ satisfying $(\ast)$ iff there exists a \emph{regular} model $\I_{1}$ satisfying $(\ast)$. 
The lower bound is proved by a reduction of the word problem for exponentially bounded alternating Turing machines.

The proofs in the rUCQ case are similar. Interestingly, however, one can now prove directly a homomorphism based characrerization of rUCQ entailment without using results on regular tree languages.

We now discuss the proofs in more detail. Let $\K_1$, $\K_2$ be \ALC-KBs and
$\Sigma$ a signature. We use two-way alternating parity automata on infinite
trees (\TWAPAs) and consider the class $\Mod^{\it fo}_{\K_{1}}$, encoding
forest-shaped interpretations as labeled trees to make them accessible to
\TWAPAs.  A \emph{tree} is a non-empty (possibly infinite) set
$T \subseteq \Nbbm^*$ closed under prefixes with root $\varepsilon$.
We say that $T$ is \emph{$m$-ary} if, for every $x \in T$, the set $\{ i \mid x \cdot i \in T \}$ is of cardinality~$m$. Let $\Gamma$ be an alphabet with symbols from the set
%
%
$$
\{\mathit{root}, \mathit{empty}\}  \cup (\ind(\K_1) \times 2^{\mn{CN}(\Tmc_1)}) \cup
(\mathsf{RN}(\T_1) \times 2^{\mn{CN}(\Tmc_1)}),
$$
where $\mathsf{CN}(\T_i)$ (resp.\ $\mathsf{RN}(\T_i)$) denotes the set of
concept (resp.\ role) names in $\T_i$. A \emph{$\Gamma$-labeled tree} is a pair
$(T,L)$ with $T$ a tree and $L \colon T \rightarrow \Gamma$ a node labeling
function. We represent forest-shaped models of $\Tmc_1$ as $m$-ary
$\Gamma$-labeled trees, with $m = \text{max}(|\T_1|, |\ind(\K_1)|)$. The root
node labeled with $\mathit{root}$ is not used in the representation. Each ABox
individual is represented by a successor of the root labeled with a symbol from
$\ind(\K_1) \times 2^{\mn{CN}(\T_1)}$; non-ABox elements are represented by
nodes deeper in the tree labeled with a symbol from
$\mathsf{RN}(\T_1) \times 2^{\mn{CN}(\T_1)}$. The label $\mathit{empty}$ is
used for padding to make sure that every tree node has exactly $m$ successors.

Now we construct three \TWAPAs $\Amf_i$, for $i=0,1,2$. $\Amf_0$ ensures that
the tree is labeled in a meaningful way, e.g.\ that the \emph{root} label only
occurs at the root node; $\Amf_1$ accepts $\Gamma$-labeled trees that represent
a model of $\K_1$, and $\Amf_2$ accepts $\Gamma$-labeled trees $(T,L)$ which
represent an interpretation $\Imc_{(T,L)}$ such that some model of $\K_2$ is
$\Sigma$-homomorphically embeddable into $\Imc_{(T,L)}$. The most
interesting automaton is $\Amf_2$, which guesses a model of $\K_2$ along with a
homomorphism to $\Imc_{(T,L)}$; in fact, both can be read off from a successful
run of the automaton. The number of states of the $\Amf_i$ is
exponential in $|\Kmc_1 \cup \Kmc_2|$. It then remains to combine these
automata into a single \TWAPA \Amf such that $\L(\Amf)=\L(\Amf_0) \cap \L(\Amf_1)
\cap \overline{\L(\Amf_2)}$, which is possible with only polynomial blowup, and
to test (in time exponential in the number of states) whether
$\L(\Amf)=\emptyset$.

One can thus construct an automaton $\mathfrak{A}$ in exponential time such
that $\L(\Amf)$ contains exactly the (trees representing) models
$\I_1 \in\Mod^{\it fo}_{\K_{1}}$ of outdegree at most $|\T_{1}|$ such that no
model of $\K_2$ is $\Sigma$-homomorphically embeddable into $\I_1$. Then, if
$\L(\Amf)\neq\emptyset$, then by Rabin \cite{Rabin1972}, $\L(\Amf)$ contains (the
tree representing) a \emph{regular} model in the following sense:

\begin{definition}
  A ditree interpretation $\Imc$ is \emph{regular} if it has, up to
  isomorphisms, finitely many rooted subintertretations. A forest-shaped model
  $\Imc$ of a KB $\K$ is regular if the ditree interpretations $\Imc_{a}$,
  $a\in {\sf ind}(\K)$, are regular.
\end{definition}

Next, we show that regular models in $\L(\Amf)$ are witnesses to non-UCQ
entailment according to the characterization in Theorem~\ref{crit:KB}~(1).
\begin{lemma}\label{lem:1}
  Let $\K_{1}$ and $\K_{2}$ be KBs, $\Sigma$ a signature, and let $\Imc_{1}$ be
  a regular forest-shaped model of $\K_{1}$ of bounded outdegree. Assume that
  no model of $\K_{2}$ is $\Sigma$-homomorphically embeddable into $\Imc_{1}$.
  Then there exists $n>0$ such that no model of $\K_{2}$ is
  $n\Sigma$-homomorphically embeddable into $\Imc_{1}$.
\end{lemma}

By Lemma~\ref{lem:1} and Theorem~\ref{crit:KB}~(1) we then have that $\L(\Amf)$
is non-empty iff $\K_{1}$ does not $\Sigma$-UCQ entail $\K_{2}$. We have thus proved the following.
\begin{theorem}
The problem whether an \ALC KB $\Sigma$-entails an \ALC KB is decidable in 2\ExpTime.
\end{theorem}
We also obtain the following strengthening of the model-theoretic characterization for $\Sigma$-UCQ entailment.
\begin{theorem}\label{UCQ-full-hom}
  $\mathcal{K}_{1}$ $\Sigma$-UCQ entails $\K_{2}$ iff for all models $\Imc_{1}$
  of $\K_{1}$ there exists a model $\I_{2}$ of $\K_{2}$ that is
  $\Sigma$-homomorphically embeddable into $\Imc_{1}$.
\end{theorem}
\begin{proof}
  The direction from left to right follows from
  Theorem~\ref{crit:KB}~(1). Conversely, assume there exists a model $\Imc_{1}$
  of $\K_{1}$ for which there does not exist any model $\I_{2}$ of $\K_{2}$
  that is $\Sigma$-homomorphically embeddable into $\Imc_{1}$. Then there
  exists a forest-shaped model of $\K_{1}$ of outdegree at most $|\T_{1}|$ with
  this property. Then $\L(\mathfrak{A})$ is non-empty and so contains a regular
  tree. Then there exists a regular forest-shaped model of $\K_{1}$ of bounded
  outdegree for which there does not exist any model $\I_{2}$ of $\K_{2}$ that
  is $\Sigma$-homomorphically embeddable into $\Imc_{1}$. By Lemma~\ref{lem:1}
  and Theorem~\ref{crit:KB} (1) we obtain that $\mathcal{K}_{1}$ does not
  $\Sigma$-UCQ entail $\K_{2}$.
\end{proof}

The matching 2\ExpTime lower bound is proved in the appendix. Thus we obtain the following result.
\begin{theorem}
\label{thm:ucq}
The problem whether an \ALC KB $\Kmc_1$ $\Sigma$-UCQ entails an \ALC KB
$\Kmc_2$ is 2\ExpTime-complete.
\end{theorem}
As for rooted UCQs, we can strengthen the model-theoretic characterization by
replacing con-n$\Sigma$-homomorphic embeddability with con-$\Sigma$-homomorphic
embeddability, where $\Imc_{2}$ is \emph{con-$\Sigma$-homomorphically
  embeddable into $\Imc_{1}$} if the maximal $\Sigma$-connected
subinterpretation of $\Imc_{2}$ is $\Sigma$-homomorphically embeddable into
$\Imc_{1}$.
\begin{theorem}\label{rUCQ-full-hom}
  Let $\K_1$ and $\K_2$ be \ALC KBs, $\Sigma$ a signature, and let $\Mod_{\!1}$
  be complete for $\K_{1}$. Then $\K_1$ $\Sigma$-\RUCQ entails $\K_2$ iff for any
  $\I_1 \in \Mod_{\!1}$, there exists $\I_2 \models \K_2$ such that $\I_2$ is
  con-$\Sigma$-homomorphically embeddable into~$\I_1$.
\end{theorem}
\begin{proof}
  The proof is a straighforward modification of the proof of Lemma~\ref{lem:1}. A proof sketch is as follows: in 
	view of Theorem~\ref{crit:KB} (2), it suffices to prove
  $(\Rightarrow)$. Suppose $\I_{1} \in \Mod_{\!1}$. By Theorem~\ref{crit:KB} (2),
  for every $n\geq 0$, we have $\J\in \Mod^{\it fo}_{\K_{2}}$ and a
  $\Sigma$-homomorphism $h_{n} \colon \J_{|\leq n} \to \I_{1}$, where $\J_{|\leq
    n}$ is the subinterpretation of $\J$ whose elements are connected to ABox
  individuals by $\Sigma$-paths of length $\le n$. Clearly, for any $n\geq 0$,
  there are only finitely many non-isomorphic pairs $(\J_{|\leq n},h_{n})$. It
  can be shown that, thus, one can construct the required $\I_2\in\Mod^{\it
    fo}_{\K_{2}}$ and con-$\Sigma$-homomorphism $h$ as the limits of suitable
  chains $\J_{|\leq 0}\subseteq \J_{|\leq 1} \subseteq \cdots$ and $h_{0}
  \subseteq h_{1}\subseteq\cdots$, respectively.
\end{proof}
The above automata construction can be slightly modified to check the
condition of Theorem~\ref{rUCQ-full-hom}. The lower bound reduction works
already for \RUCQ{s}, therefore we obtain the following result:

\begin{theorem}
\label{thm:rucq}
The problem whether an \ALC KB $\Kmc_1$ $\Sigma$-\RUCQ entails an \ALC KB
$\Kmc_2$ is 2\ExpTime-complete.
\end{theorem}

\section{(r)CQ-Entailment for (Horn-)\ALC-TBoxes}

We show that CQ- and rCQ-entailment between \ALC TBoxes becomes
decidable when the second TBox is given in Horn-\ALC. In this case, entailments for CQs and UCQs and, respectively,
rCQs and rUCQs coincide.
%
%
We start with rCQs.

Our first observation is that if a $\Sigma_{1}$-ABox is a witness for
non-$\Theta$-rCQ entailment, then one can find a witness
$\Sigma_{1}$-ABox that is tree-shaped and of bounded outdegree. Here,
an ABox $\A$ is \emph{tree-shaped} if the graph with nodes
$\mn{ind}(\Amc)$ and edges $\{a,b\}$ for each $R(a,b)\in \A$ is a
tree, and $R(a,b)\in \A$ implies $S(a,b) \notin \Amc$ for all $S \neq
R$ and $S(b,a) \notin \Amc$ for all $S$.
\begin{theorem}
\label{thm:first}
Let $\T_{1}$ be an \ALC TBox, $\T_{2}$ a
\hALC TBox, and $\Theta=(\Sigma_{1},\Sigma_{2})$.
Then $\T_{1}$ $\Theta$-rCQ-entails $\T_{2}$ iff,
for all tree-shaped $\Sigma_{1}$-ABoxes $\A$ of outdegree bounded
by $|\T_{2}|$ and consistent with $\T_{1}$ and $\T_{2}$,
$\I_{\T_{2},\A}$ is con-$\Sigma_{2}$-homomorphically embeddable into any
model $\I_{1}$ of $(\T_{1},\A)$.
\end{theorem}
\begin{proof}
It is known that \hALC is unravelling tolerant, that is, $(\T,\A)\models C(a)$ for a \hALC TBox $\T$ and $\mathcal{EL}$-concept~$C$ iff $(\T,\A')\models C(a)$ for a finite sub-ABox $\A'$ of the tree-unravelling of $\A$ at $a$~\cite{lutz2012non}. Thus, any witness ABox for non-entailment w.r.t.~$\mathcal{EL}$-instance queries can be transformed into a tree-shaped witness ABox. The result follows by observing that if $\T_{1}$ does not $\Theta$-rCQ-entail $\T_{2}$, then this is witnessed by an $\mathcal{EL}$-instance query and by applying Theorem~\ref{rUCQ-full-hom} to the KBs. The bound on the outdegree is obtained by a careful analysis of derivations.
\end{proof}
%
%
\newcommand{\Sigmaabox}{\Sigma_1} \newcommand{\Sigmaquery}{\Sigma_2}
For the automaton construction, let $\Tmc_1$ be an \ALC TBox, $\Tmc_2$
a Horn-\ALC TBox, and $\Theta=(\Sigma_1,\Sigma_2)$ a pair of
signatures. Though Theorem~\ref{thm:first} provides a natural
characterization that is similar in spirit to
Theorem~\ref{rUCQ-full-hom}, we first need a further analysis of
con-$\Sigmaquery$-homomorphic embeddability in terms of simulations
whose advantage is that they are more compositional (they can be
partial and are closed under union).

Let $\Imc_1,\Imc_2$ be interpretations and $\Sigma$ a signature.  A
relation
$\mathcal{S}\subseteq \Delta^{\Imc_1}\times \Delta^{\Imc_2}$ is a
\emph{$\Sigma$-simulation} from $\Imc_{1}$ to $\Imc_{2}$ if (\emph{i}) $d\in A^{\Imc_{1}}$ and $(d,d')\in \mathcal{S}$ imply $d'\in A^{\Imc_{2}}$ for all $\Sigma$-concept
names $A$, and (\emph{ii}) if $(d,e)\in R^{\Imc_{1}}$ and $(d,d')\in \mathcal{S}$ then there is a $(d',e')\in R^{\Imc_{2}}$ with $(e,e')\in \mathcal{S}$ for all $\Sigma$-role names~$R$.  Let $d_i \in \Delta^{\Imc_i}$, $i \in
\{1,2\}$.  $(\Imc_{1},d_{1})$ is $\Sigma$-simulated by $(\Imc_{2},d_{2})$, in
symbols $(\Imc_{1},d_{1}) \leq_{\Sigma}(\Imc_{2},d_{2})$, if there exists a
$\Sigma$-simulation $\mathcal{S}$ with $(d_{1},d_{2})\in \mathcal{S}$.
\begin{lemma}
\label{lem:homtosim}
Let \Amc be a $\Sigmaabox$-ABox and $\Imc_1$ a model of $(\Tmc_1,\Amc)$. Then
$\Imc_{\Tmc_2,\Amc}$ is not con-$\Sigmaquery$-homomorphically embeddable into
$\Imc_1$ iff there is $a \in \mn{ind}(\Amc)$ such that one of the following
holds:
  \begin{enumerate}\itemsep 0cm

  \item[\rm (1)] there is a $\Sigmaquery$-concept name $A$ with
    $a \in A^{\Imc_{\Tmc_2,\Amc}} \setminus A^{\Imc_1}$;

  \item[\rm (2)] there is an $R$-successor $d$ of $a$ in
    $\Imc_{\Tmc_2,\Amc}$, for some $\Sigmaquery$-role name $R$, such
    that $d \notin \mn{ind}(\Amc)$ and, for all $R$-successors $e$
    of $a$ in $\Imc_1$, we have  $(\Imc_{\Tmc_2,\Amc},d)
    \not\leq_{\Sigmaquery}(\Imc_1,e)$.

  \end{enumerate}
\end{lemma}
We use a mix of two-way alternating B\"uchi automata on finite trees (\TWABAs)
and non-deterministic top-down automata on finite trees (NTAs).  A finite tree
$T$ is \emph{$m$-ary} if, for every $x \in T$, the set $\{ i \mid x \cdot i \in
T \}$ is of cardinality zero or exactly~$m$.  We use labeled trees to represent
a tree-shaped ABox \Amc and a model $\Imc_1$ such that, for some $a \in
\mn{ind}(\Amc)$, conditions~(1) and~(2) from Lemma~\ref{lem:homtosim} are satisfied,
and thus $\Imc_{\Tmc_2,\Amc}$ is not con-$\Sigmaquery$-homomorphically
embeddable into $\Imc_1$. To ensure that later, additional bookkeeping
information is needed. Node labels are taken from the alphabet
$$
\Gamma = \Gamma_0 \times 2^{\mn{cl}(\Tmc_1)} \times 2^{\mn{CN}(\Tmc_2)}
\times \{ 0,1 \} \times 2^{\mn{sub}(\Tmc_2)},
$$
where $\Gamma_0$ is the set of all subsets of $\Sigmaabox \cup \{ R^- \mid R
\in \Sigmaabox \}$ that contain at most one role (a role name $R$ or
its inverse $R^-$), $\mathsf{cl}(\T_i)$ is the set of subconcepts of
(concepts in) $\T_i$ closed under single negation, and $\mn{sub}(\Tmc_2)$ is the set of subconcepts of (concepts in)~$\Tmc_2$.
For a $\Gamma$-labeled tree $(T,L)$ and a node $x$
from $T$, we use $L_i(x)$ to denote the $(i+1)$st component of $L(x)$, where $i \in \{0,\dots,4\}$. Intuitively, the $L_0$-component represents the ABox
\Amc, the $L_1$-component the model $\Imc_1$, the $L_2$-component
represents $\Imc_{\Tmc_2,\Amc}$, and the $L_3$- and $L_4$-components help to guarantee conditions~(1) and~(2) from Lemma~\ref{lem:homtosim}.

To ensure that each component $i \in \{0,\dots,4\}$ indeed represents what it is
supposed to, we impose on it an \emph{$i$-properness} condition. For
example, a $\Gamma$-labeled $(T,L)$ tree is \emph{0-proper} if
(i)~$L_0(\varepsilon)$ contains no role and (ii)~for every non-root node $x$ of
$T$, $L_0(x)$ contains a role. A 0-proper $\Gamma$-labeled tree $(T, L)$
represents the following tree-shaped $\Sigmaabox$-ABox:
\[\begin{array}{r@{~}c@{~}l}
\Amc_{( T, L)} &= & \{ A(x) \mid A \in L_0(x) \} \cup{}\\ 
  &&  \{ R(x,y) \mid R \in L_0(y),~ y \text{ is a child of } x \} \cup{}\\  
  &&  \{ R(y,x) \mid  R^{-} \in L_0(y),~ y \text{ is a child of } x \}. 
\end{array}\]
Due to space limitations, we skip the remaining definitions of
properness and concentrate on explaining the most interesting components $L_3$ and $L_4$ of
$\Gamma$-labels. The $L_3$-component marks a single node
$x$ in the tree, which is the individual $a$ from Lemma~\ref{lem:homtosim} that
satisfies conditions~(1) and~(2). If (1) is satisfied, we do not need the
$L_4$-component. Otherwise, we store in that component at $x$ a set of concepts
$S=\{ \exists R . A, \forall R . B_1, \dots, \forall R . B_n \}$ such that $R
\in \Sigmaquery$ and all concepts from $S$ are true at $x$ in
$\Imc_{\Tmc_2,\Amc}$. This \emph{successor set} represents the
$R$-successor $d$ in condition~(2) of Lemma~\ref{lem:homtosim}. We then have to
make sure that, for any neighboring node $y$ of $x$ that represents an
$R$-successor of $x$ in $\Amc_{(T,L)}$, we have $(\Imc_{\Tmc_2,\Amc},d)
\not\leq_{\Sigmaquery}(\Imc_1,y)$. This can again happen via a concept name or
via a successor; we are done in the fomer case and use the $L_4$-component of
$y$ in the latter. It is important to note that we can never return to the same
node in this tracing process since we only follow roles in the forward direction
and the represented ABox is tree-shaped. This is crucial for achieving the \ExpTime overall complexity.

We show that $\Tmc_2$ is not $\Theta$-\RCQ-entailed by $\Tmc_1$ iff there is an
$m$-ary $\Gamma$-labeled tree that is $i$-proper for any $i \in
\{0,\dots,4\}$. It then remains to design a \TWABA \Amf that accepts exactly
those trees. We construct \Amf as the intersection of five automata $\Amf_i$, $i
< 5$, where each $\Amf_i$ ensures $i$-properness. Some of the automata are
\TWABA{s} with polynomially many states while others are NTAs with
exponentially many states. We mix automata models since some properness
conditions (2-properness) are much easier to describe with a \TWABA while for
others (4-properness), it does not seem to be possible to construct a \TWABA
with polynomially many states. In summary, we obtain the following result.
\begin{theorem}
\label{thm:exptbox}
It is \ExpTime-complete to decide whether an \ALC TBox $\Tmc_1$
$(\Sigma_1,\Sigma_2)$-\RCQ entails a Horn-\ALC TBox $\Tmc_2$.
\end{theorem}
Note that the \ExpTime lower bound holds already for entailment
of $\mathcal{EL}$ TBoxes and $\Sigma_{1}=\Sigma_{2}$ \cite{DBLP:journals/jsc/LutzW10}.
We now study the non-rooted case, starting with an analogue of
Theorem~\ref{thm:first}. As expected, moving to unrestricted queries
corresponds to moving to unrestricted homomorphisms. 
%
%
\begin{theorem}
\label{thm:second}
Let $\T_{1}$ and $\T_{2}$ be \hALC TBoxes and $\Theta=(\Sigma_{1},\Sigma_{2})$.  Then $\T_{1}$ $\Theta$-CQ entails $\T_{2}$ iff, for all tree-shaped $\Sigma_{1}$-ABoxes $\A$ of outdegree $\le |\T_{2}|$ and consistent with $\T_{1}$ and~$\T_{2}$, $\I_{\T_{2},\A}$ is $\Sigma_{2}$-homomorphically embeddable into $\I_{\T_{1},\A}$.
\end{theorem}
The automata construction described above can largely be reused for this
case. The main difference is that the two conditions in
Lemma~\ref{lem:homtosim} need to be extended with a third one: there is an
element $d$ in the subtree of $\Imc_{\Tmc_2,\Amc}$ rooted at $a$ that has an
$R$-successor $d_0$, $R \notin \Sigmaquery$, such that,
for all elements $e$ of $\Imc_1$, we have $(\Imc_{2},d_0)
\not\leq_{\Sigmaquery}(\Imc_1,e)$. To deal with this condition, it becomes
necessary to store multiple successor sets in the $L_4$-components instead
of only a single one, which increases the overall complexity to 2\ExpTime.
A matching lower bound can be proved by a (non-trivial) reduction of the
word problem for exponentially bounded alternating Turing machines.
\begin{theorem}
\label{thm:2exptbox}
$\Theta$-CQ entailment for \hALC TBoxes is 2\ExpTime-complete. The lower bound holds for $\Theta = (\Sigma,\Sigma)$.
%
\end{theorem}

\section{Future Work}
We have made first steps towards understanding query entailment and
inseparability for KBs and TBoxes in expressive DLs. Many problems remain to be
addressed. From a theoretical viewpoint, it would be of interest to solve the
open problems in Figures~\ref{table:kb} and \ref{table:tbox}, and also consider
other expressive DLs such as \textsl{DL-Lite}$_{\it bool}^\mathcal{H}$~\cite{ACKZ09} or
$\mathcal{ALCI}$. 
Also, our undecidability proof goes through for
\textsl{DL-Lite}$_{\it bool}^\mathcal{H}$, but the other cases remain open. From
a practical viewpoint, our model-theoretic criteria for query entailment are a
good starting point for developing algorithms for approximations of query
entailment based on simulations. Our undecidability and complexity results also
indicate that rUCQ-entailment is more amenable to practical algorithms than,
say, CQ-entailment and can be used as an approximation of the latter.

\smallskip
\noindent
{\bf Acknowledgments.}
This work has been supported by the EU IP project Optique, grant n.\ FP7-318338, DFG grant LU 1417/2-1, and EPSRC UK grants EP/M012646/1 and EP/M012670/1 (iTract).


\bibliographystyle{plainnat}
\bibliography{string-medium,main-bib}

\appendix
\clearpage

\newenvironment{theoremnum}[1]{\smallskip\noindent\textbf{Theorem~#1.}\hspace*{0.3em}\em}{\par\smallskip}

\newenvironment{lemmanum}[1]{\smallskip\noindent\textbf{Lemma~#1.}\hspace*{0.3em}\em}{\par\smallskip}


%
%
%
%
%
%
%
%


\section{Proof of Theorem~\ref{thm:undecidability}}

\subsection{Minimal models}

We consider \ELUf TBoxes, $\T$, that consist of concept inclusions of the form
\begin{itemize}
\item[--] $A \ISA C$,
\item[--] $A \ISA B \OR C$,
\item[--] $A \ISA \SOME{R}{C}$,
\end{itemize}
where $A, B, C$ are concept names and $R$ is a role name. We construct by induction a (possibly infinite) labelled forest $\mathfrak{O}$
with a labelling function $\ell$. For each $a \in \ind(\A)$, $a$ is the root
of a tree in $\mathfrak{O}$ with $A \in \ell(a)$ iff $A(a) \in \A$.
Suppose now that $\sigma$ is a node in $\mathfrak{O}$ and $A \in
\ell(\sigma)$. If $A \sqsubseteq C$ is an axiom of $\T$ and $C \notin
\ell(\sigma)$, then we add $C$ to $\ell(\sigma)$. If $A \sqsubseteq B \OR C$ is
an axiom of $\T$ and neither $B \in \ell(\sigma)$ nor $C \in \ell(\sigma)$,
then we add to $\ell(\sigma)$ either $B$ or $C$ (but not both); in this case, we
call $\sigma$ an \emph{or-node}. If $A \sqsubseteq \exists R.C$ is an axiom of
$\T$, but the constructed part of the tree does not contain a node $\sigma \cdot
w_{\exists R.C}$, then we add $\sigma \cdot w_{\exists R.C}$ as an  $R$-successor of
$\sigma$ and set $\ell(\sigma \cdot w_{\exists R.C}) = \{C\}$.

Given an \ELUf KB $\K = (\T,\A)$, we define a \emph{minimal model} $\M = (\Delta^\M, \cdot^\M)$ of $\K$ by
taking $\Delta^\M$ to be the set of nodes in $\mathfrak{O}$, $R^\M$ to be the $R$-relation in $\mathfrak{O}$ together with $(a,b)$ such that $R(a,b) \in
\A$, and set
$$
A^\M = \{\sigma \in \Delta^\M \mid A \in \ell(\sigma)\},
$$
for every concept name $A$.

%
\begin{lemma}\label{min-elu-complete}
  Let $\K$ be an \ELUf KB $\K$ and let $\Mod_\K$ be the set of its minimal models. Then $\Mod_\K$ is complete for $\K$.
\end{lemma}
\begin{proof}
It suffices to show that (\emph{i}) every minimal model is a model of $\K$, and (\emph{ii}) for every model $\I$ of $\K$, there is a minimal model $\M$ that is homomorphically embeddable into $\I$.
The former follows from the construction.

  (\emph{ii}) Let $\I$ be a model of $\K$. We construct by induction a set $\Delta$ and
  a labelling function $\ell$ defining a minimal model $\M$ and a function $h$
  such that $h$ is a homomorphism from $\M$ to $\I$.
  First we set $a \in \Delta$ and $A \in \ell(a)$, for each $A(a) \in
  \A$. Suppose that $A \in \ell(a)$ for some $a$. If $A \ISA C$ is an axiom in
  $\T$ and $C \notin \ell(A)$, we add $C$ to $\ell(a)$. Suppose now that $A \ISA
  B \OR C$ is an axiom in $\T$, and $B \notin \ell(A)$, $C \notin
  \ell(a)$. Since $\I$ is a model of $\K$, it must be the case that $B \in
  \type^\I(a)$ or $C \in \type^\I(a)$. In the former case, we add $B$ to
  $\ell(a)$, in the latter case, we add $C$ to $\ell(a)$. We now set $h(a) = a$,
  for each $a \in \ind(\A)$. Clearly, $A \in \type^\I(h(a))$, for each $A \in
  \ell(a)$.

  Suppose that $\sigma \in \dom[\M]$ such that $h(\sigma)$ is set, and $A \in
  \ell(\sigma)$.
  Suppose further that $A \ISA \SOME{R}{C} \in \T$ and $\sigma \cdot
  w_{\SOME{R}{C}}$ is not in $\Delta$. Since $\I$ is a model of $\K$ and by
  inductive assumption $A \in \type^\I(h(\sigma))$, there exists $d \in
  \dom[\I]$ such that $(h(\sigma),d) \in R^\I$ and $d \in C^\I$. So we add
  $\sigma \cdot w_{\SOME{R}{C}}$ to $\Delta$ as successor of $\sigma$, define
  $\ell(\sigma \cdot w_{\SOME{R}{C}})$ similarly to the base case starting from
  $\{C\}$, and set $h(\sigma \cdot w_{\SOME{R}{C}}) = d$. Clearly, for each
  $\sigma \in \Delta$, for each $A \in \ell(\sigma)$ we have that $A \in
  \type^\I(h(\sigma))$.

  Now the minimal model $\M$ is defined as $(\Delta, \cdot^\M)$, where
  $\cdot^\M$ is defined as in the construction of minimal model. By
the  construction of $\Delta$ and the fact that $\M$ is minimal, we obtain that
  $h$ is indeed a homomorphism from $\M$ to $\I$.
  %
\end{proof}

\subsection{Proof of Theorem~\ref{thm:undecidability}~(\emph{i}) and~(\emph{ii}) for CQs}

A \emph{tile type} $T = (\textit{up}(T), \textit{down}(T),
\textit{left}(T), \textit{right}(T))$ consists of four colours. The following $N\times
M$-\emph{tiling problem} is known to be undecidable: given a finite set
$\mathfrak T$ of tile types, a tile type $I \in \mathfrak T$ and two colours ${\it wall}$ and ${\it ceiling}$, decide whether there exist $N,M
\in \mathbb N$ such that the $N \times M$ grid can be tiled using $\mathfrak T$
in such a way that $(1,1)$ is covered with a tile of type $I$, every $(N,i)$,
for $i \le M$, is covered with a tile of some type $T$ with $\textit{right}(T) =
\textit{wall}$, and every $(i,M)$, for $i \le N$, is covered with a tile of some
type $T$ with $\textit{up}(T) = \textit{ceiling}$.

We require role names $P$ and $R$, and the following concept names:
\begin{itemize}
\item[--] $T^{\first}, T_k,T_k^{\halt}, \widehat{T}_k$ for $T \in \mathfrak{T}$, $k=0,1,2$;
\item[--] $\Row, \Row_k, \Row_k^{\halt}$, for $k=0,1,2$;
\item[--] $A$, $\Start$ and $\End$.
\end{itemize}
Let $\K_2 = (\T_2,\{A(a)\})$, where $\T_2$ contains the following axioms, for $k = 0,1,2$:
\begin{align}
\label{initial}
& A \sqsubseteq \exists P.(\Start \sqcap \exists R.I^{\first}),\\
\label{horizontal-0}
& T^{\first} \sqsubseteq \exists R.S^{\first},\, \text{ if $\textit{right}(T) =
  \textit{left}(S)$,\, $T,S \in \mathfrak T$}, \\
\label{end-of-row-0}
& T^{\first} \sqsubseteq \exists R.(\Start \sqcap \Row_1), T \in \mathfrak T,
\textit{right}(T) = \textit{wall},\\
\label{first-row-0}
& T^{\first} \sqsubseteq \widehat{T}_0,\quad \text{for $T \in \mathfrak T$},\\
%
\label{7}
& \Row_k \sqsubseteq \exists R. T_k, \quad \text{for $T \in \mathfrak T$},\\
\label{horizontal-k}
& T_k \sqsubseteq \exists R.S_k, \quad \text{ if $\textit{right}(T) =
  \textit{left}(S)$ and $T,S \in \mathfrak T$}, \\
\label{end-of-row-k}
& T_k \sqsubseteq \exists R.\Row_{(k+1) \,\text{mod}\, 3}, \quad \text{if }
\textit{right}(T) = \textit{wall},\\
\label{start-last-row}
& T_k \sqsubseteq \exists R.\Row^{\halt}_{(k+1)\, \text{mod}\, 3},\quad
\text{if } \textit{right}(T) = \textit{wall},\\
\label{row-k}
& \Row_k \sqsubseteq \Row, \\
\label{tile-hat}
& T_k \sqsubseteq \widehat{T}_k,\quad \text{for $T \in \mathfrak T$},\\
\label{vertical-k}
& T_k \sqsubseteq \widehat{S}_{(k-1)\, \text{mod}\, 3}, \quad \text{if $\textit{down}(T) = \textit{up}(S)$, $T,S \in \mathfrak T$},\\
%
\label{disjunction}
& \Row_k^{\halt} \sqsubseteq \exists R.{\End} \sqcup
\bigsqcap_{\textit{up}(T) = \textit{ceiling}} \exists R.T^{\halt}_k, \\
\label{horizontal-halt}
& T_k^{\halt} \sqsubseteq \exists R. S^{\halt}_k, \quad \text{if
  $\textit{right}(T) = \textit{left}(S)$ and $\textit{up}(S) =
  \textit{ceiling}$},\\
\label{end}
& T_k^{\halt} \sqsubseteq \exists R.(\Row \sqcap \exists R.{\End}), \quad
\text{if $\textit{right}(T) = \textit{wall}$},\\
\label{row-halt}
& \Row^{\halt}_k \sqsubseteq \Row,\\
\label{vertical-halt}
& T_k^{\halt} \sqsubseteq \widehat{S}_{(k-1)\, \text{mod}\, 3}, ~~ \text{if
  $\textit{down}(T) = \textit{up}(S)$, }  T,S \in \mathfrak{T}.
\end{align}
The axioms \eqref{initial}-\eqref{first-row-0} produce the following tree
rooted at an $A$-point:
\begin{center}
  \begin{tikzpicture}[xscale=1.2, yscale=0.8]
    \foreach \al/\x/\y/\type/\lab/\wh/\extra in {%
      a/0/0/point/A/right/constant, %
      start/0/-1/point/\Start/right/, %
      n0/0/-2/point/I^{\first}/80/,%
      n1/-1/-3/point/T^{\first}/left/,%
      n2/0.1/-3/point/T^{\first}/right/,%
      st1/2/-3/subtree/{\qquad\Start,\Row_1}/90/,%
      n3/-2/-4/point/T^{\first}/left/,%
      n4/-0.9/-4/point/T^{\first}/right/,%
      st2/0.3/-4/subtree/{\qquad\qquad~\Start,\Row_1}/90/%
    }{ \node[\type, \extra, label=\wh:{\footnotesize $\lab$}] (\al) at (\x,\y)
      {}; }

    \foreach \from/\to/\type in {%
      a/start/dashed, start/n0/role, %
      n0/n1/role, n0/n2/role, n0/st1.north/subtreecolor, %
      n1/n3/role, n1/n4/role, n1/st2.north/subtreecolor%
    } {\draw[role, \type] (\from) -- (\to);}

    \foreach \from in {%
      n2, n3, n4%
    } {\draw[thick, dotted] (\from) -- +(0,-0.5);}

    \foreach \from/\to in {%
      n1/n2, n3/n4%
    } {\draw[thick, dotted] ($0.7*(\from)+0.3*(\to)+(0,0.1)$) --
      ($0.3*(\from)+0.7*(\to)+(0,0.1)$);}

    \foreach \wh in {%
      st1, st2%
    }{ \node[anchor=west] at (\wh.east) {\small $\tau_1$}; }
  \end{tikzpicture}
\end{center}
The axioms \eqref{7}-\eqref{vertical-k} produce trees $\tau_k$ rooted at
$\Row_k$-points:
\begin{center}
  \begin{tikzpicture}[xscale=1.1, yscale=0.8]
    \foreach \al/\x/\y/\type/\lab/\wh/\extra in {%
      r0/0/0/point/{\Row_k}/above/rowind, %
      r1/-1.2/-1/point/T_k/left/, %
      r2/1.2/-1/point/T_k/right/, %
      r3/-2.2/-2/point/T_k/left/, %
      r4/-1.1/-2/point/T_k/right/, %
      nr2/0/-2/subtree//85/,%
      end2/1.8/-2/subtree//90/%
    }{ \node[\type, \extra, label=\wh:{\footnotesize $\lab$}] (\al) at (\x,\y) {}; }

    \foreach \from/\to/\type in {%
      r0/r1/role, r0/r2/role, %
      r1/r3/role, r1/r4/role, %
      r1/nr2.north/subtreecolor, r1/end2.north/subtreecolor%
    } {\draw[thick, role, \type] (\from) -- (\to);}

    \foreach \from in {%
      r2, r3, r4%
    } {\draw[thick, dotted] (\from) -- +(0,-0.5);}

    \foreach \from/\to in {%
      r1/r2, r3/r4%
    } {\draw[thick, dotted] ($0.7*(\from)+0.3*(\to)$) --
      ($0.3*(\from)+0.7*(\to)$);}

    \foreach \wh/\lab in {%
      nr2/, end2/\halt%
    }{ \node[anchor=west] at (\wh.80) {\small $\tau_{_{(k+1)\, \text{mod}\, 3}}^{\lab}$}; }

    \node at (-2,0.5) {$\tau_k:$};
  \end{tikzpicture}
\end{center}
Finally, the axioms \eqref{disjunction}-\eqref{vertical-halt} produce trees
$\tau_k^{\text{halt}}$ rooted at $\Row_k^{\halt}$-points:
\begin{center}
  \begin{tikzpicture}[xscale=1.1, yscale=0.8, %
    point/.style={thick,circle,draw=black,minimum size=1mm,inner sep=0pt}%
    ]

    \foreach \al/\x/\y/\lab/\wh/\extra in {%
      v0/-0.25/0/\Row_k^{\halt}/above/rowind, %
      u1/-2/-1/{\End}/left/,%
      v1/-0.5/-1/{T_k^{\halt}}/left/,%
      v2/1/-1/T_k^{\halt}/right/,%
      v3/-0.75/-2/T_k^{\halt}/left/,%
      v4/0/-2/T_k^{\halt}/right/,%
      row1/1.5/-2/\Row/right/rowind,%
      end1/1.5/-3/\End/right/subtreecolor%
    }{ \node[point, \extra, label=\wh:{\footnotesize $\lab$}] (\al) at (\x,\y)
      {}; }

    \foreach \from/\to/\type in {%
      v0/u1/role, v0/v1/role, v0/v2/role, v1/v3/role, v1/v4/role, v1/row1/subtreecolor,
      row1/end1/subtreecolor%
    } {\draw[thick, role, \type] (\from) -- node[pos=0.4, outer sep=0, inner
      sep=0] (\from-\to) {} (\to);}

    \foreach \from in {%
      v2, v3, v4%
    } {\draw[thick, dotted] (\from) -- +(0,-0.5);}

    \foreach \from/\to in {%
      v1/v2, v3/v4%
    } {\draw[thick, dotted] ($0.7*(\from)+0.3*(\to)$) --
      ($0.3*(\from)+0.7*(\to)$);}

    \draw (v0-u1) to[bend right] node[pos=0.7, yshift=0.12cm] {\scriptsize $\lor$} (v0-v1);

    \node at (-3,0) {$\tau_k^{\text{halt}}:$};
  \end{tikzpicture}
\end{center}

Denote by $\q_n$ any Boolean CQ of the form
\begin{equation*}
  \exists \vec{x} \big(  \Start(x_0) \land \bigwedge_{i=0}^n R(x_i, x_{i+1})
  \land \bigwedge_{i=1}^n B_{i}(x_{i}) \land \End(x_{n+1}) \big)
\end{equation*}
where $B_{i} \in \{\Row\} \cup \{\widehat{T}_k \mid T
\in \mathfrak{T}, k=0,1,2\}$.

\begin{lemma}\label{qn-tiling}
There exists a CQ $\q_n$ such that $\prod \Mod_{\K_2} \models \q_n$ iff there exist $N,M \in \mathbb N$ for which $\mathfrak T$ tiles the $N \times M$ grid as described above.
\end{lemma}
\begin{proof}
$(\Leftarrow)$ Suppose $\mathfrak T$ tiles the $N \times
M$ grid so that a tile of type $T^{ij} \in \mathfrak T$ covers $(i,j)$. Let
$$
\textit{block}_j = (\widehat T^{1,j}_k, \dots, \widehat T^{N,j}_k, \Row),
$$
for $j=1,\dots,M-1$ and $k =(j-1) \mod 3$. Let $\q_n$ be the CQ in which the $B_i$ follow the pattern
$$
\textit{block}_1, \ \textit{block}_2,  \dots, \ \textit{block}_{M-1}
$$
(thus, $n = (N+1) \times (M-1)$).  In view of Lemma~\ref{min-elu-complete}, we only need to prove $\M \models \q_n$ for each minimal model $\M \in \Mod_{\K_2}$. Take such an $\M$. We have to show that there is an $R$-path $x_0,\dots,x_{n+1}$ in $\M$ such that $x_i \in B_i^\M$ and $x_{n+1} \in \End^\M$.

First, we construct an auxiliary $R$-path $y_0,\dots,y_n$.
We take $y_0 \in \Row^\M$ and $y_1 \in I_0^\M$ by~\eqref{initial} ($I_0 = T^{1,1}$). Then we take $y_2 \in (T^{1,1})^\M, \dots, y_{N+1} \in (T^{N,1})^\M$ by~\eqref{horizontal-0}. We now have $\textit{right}(T^{N,1}) = \textit{wall}$. By~\eqref{end-of-row-0}, we obtain $y_{N+2} \in \Row_1$. By~\eqref{row-k}, $y_{N+2} \in \Row_1^\M \subseteq \Row^\M$. We proceed in this way, starting with \eqref{7}, till the moment we construct $y_{n-1} \in T^{N,M-1}$, for which we use~\eqref{start-last-row} and \eqref{row-halt} to obtain $y_n \in \Row^{\halt}_k \subseteq \Row^\M$, for some $k$. Note that $T^\M \subseteq \widehat T^\M$ by~\eqref{tile-hat}.

By~\eqref{disjunction}, two cases are possible now.

\emph{Case 1}: there is $y$ such that $(y_n,y) \in R^\M$ and $y
\in \End^\M$. Then we take $x_0 = y_0, \dots, x_n = y_n, x_{n+1} = y$.

\emph{Case 2}: there is an object $z_1$ such that $(y_n,z_1) \in R^\M$ and $z_1
\in (T^{\halt}_k)^\M$, where $T = T^{1,M}$ for which $\textit{up}(T) =
\textit{ceiling}$. We then use~\eqref{horizontal-halt} and find objects
$z_2,\dots,z_N,u,v$ such that $z_i \in (T^{\halt}_k)^\M$, where $T = T^{i,M}$,
$u \in \Row^\M$ and $v \in \End^\M$. We take $x_0 = y_{N+1},\dots, x_{n-N-1} =
y_n$, $x_{n-N} = z_1,\dots, x_{n-1}= z_N$, and $x_n = u, x_{n+1} = v$. Note
that, by~\eqref{vertical-k} and~\eqref{vertical-halt}, we have $(T^{i,j})^\M
\subseteq (\widehat T^{i,j-1})^\M$.

\begin{figure}
  \centering
  \begin{tikzpicture}[yscale=1, %
    point/.style={thick,circle,draw=black,fill=white, minimum size=1mm,inner
      sep=0pt}%
    ]

    \begin{scope}[yshift=-0.75cm]
      \foreach \al/\x/\y/\lab/\wh/\extra in {%
        y0/0/0/{\Start}/right/, %
        y1/0/-0.5/B_1/right/, %
        yN/0/-1.5/{B_{N+1}}/right/rowind, %
        yN1/0/-2//right/, %
        ylN2/0/-4/{B_{n-N-1}}/right/rowind, %
        ylN1/0/-4.5//right/, %
        yl/0/-5.5/{B_n}/right/rowind, %
        yend/0/-6/\End/right/%
      }{ \node[point, \extra, label=\wh:{\scriptsize $\lab$}] (\al) at (\x,\y)
        {}; }

      \foreach \from/\to/\type in {%
        y0/y1/role, y1/yN/dotted, yN/yN1/role, yN1/ylN2/dotted, ylN2/ylN1/role,
        ylN1/yl/dotted, 
        yl/yend/role%
      } {\draw[thick, \type] (\from) -- (\to);}

      \begin{pgfonlayer}{background}
        \foreach \first/\last in {%
          y1/yN, yN1/{0,-2.9}, ylN1/yl%
        }{ \node[fit=(\first)(\last), rounded corners, fill=gray!30] {}; }
      \end{pgfonlayer}
    \end{scope}

    \begin{scope}[xshift=-2.3cm]
      \foreach \al/\x/\y/\hei/\wid in {%
        t1/0/0/1cm/1.8cm, %
        t2/0.4/-1.5/1cm/1.8cm, %
        t3/0.1/-4/1.5cm/2.2cm%
      }{ \node[subtree, minimum height=\hei, minimum width=\wid] (tree-\al) at
        (\x,\y) {};%

        \node[point] (\al) at (\x,\y) {};%
      }

      \foreach \al/\x/\y/\lab/\wh/\extra in {%
        a/0/1//right/constant,%
        start1/0/0.5//right/,%
        start2/0.4/-1//right/rowind,%
        start3/0.1/-3.5//right/rowind,%
        startm/-0.2/-5.5//right/rowind,%
        end1/-0.8/-6/\End/above left/,%
        s1/0/-4.6/Q_1/right/,%
        s2/-0.1/-5.1/S_1/right/%
      }{ \node[point, \extra, label={[inner sep=1]\wh:{\scriptsize $\lab$}}] (\al) at
        (\x,\y) {}; }

      \foreach \from/\to/\type in {%
        a/start1/dashed, %
        start1/t1/role, start2/t2/role, start3/t3/role, %
        s1/s2/role, s2/startm/role, startm/end1/role, %
        tree-t2.260/start3/dotted, t3/s1/dotted%
      } {\draw[thick, \type] (\from) -- (\to);}

      \foreach \al/\lab/\wh in {%
        start1/y_0/left, t1/y_1/left, start2/y_{N+1}/above left,
        t2/y_{N+2}/left, start3/y_{n-N-1}/left, startm/y_n/above left,
        t1/~I_0/right%
      }{ \node[label={[inner sep=2]\wh:{\scriptsize $\lab$}}, inner sep=0] at
        (\al) {}; }

      \foreach \al/\lab/\wh in {%
        startm/\sigma/above right, end1/\sigma y_{\exists R. \End}/272%
      }{\node[label={[inner sep=2,red]\wh:{\scriptsize $\lab$}}, inner sep=0] at (\al)
        {}; }


      \node at (-0.7,1) {$\M_\ell$};
      
    \begin{pgfonlayer}{background}
      \foreach \from/\to/\in in {%
        y0/start1/10, y1/t1/25, yN/start2/-15, ylN2/start3/0, yl/startm/-15,
        yend/end1/0%
      } {\draw[homomorphism, black!20] (\from) to[out=150, in=\in] (\to);}
    \end{pgfonlayer}
    \end{scope}

    \begin{scope}[xshift=2.7cm, yshift=-0cm]
      \foreach \al/\x/\y/\hei/\wid in {%
        t1/0/0/1cm/1.8cm, %
        t2/0.4/-1.5/1cm/1.8cm, %
        t3/0.1/-4/1.5cm/2.2cm, %
        tm/0.3/-6/1.5cm/2.6cm%
      }{ \node[subtree, minimum height=\hei, minimum width=\wid] (tree-\al) at
        (\x,\y) {};%
        \node[point] (\al) at (\x,\y) {};%
      }

      \foreach \al/\x/\y/\lab/\wh/\extra in {%
        a/0/1//right/constant,%
        start1/0/0.5//right/,%
        start2/0.4/-1//right/rowind,%
        start3/0.1/-3.5//right/rowind,%
        startm/-0.2/-5.5//right/rowind,%
        s1/0/-4.6/Q_1/right/,%
        s2/-0.1/-5.1/S_1/right/,%
        vN1/0.1/-6.6/U_2^{\halt}/right/,%
        vN/0.3/-7.1/T_2^{\halt}/right/,%
        row/0.5/-7.5//above right/rowind,%
        end2/0.5/-8/\End/right/%
      }{ \node[point, \extra, label={[inner sep=1]\wh:{\scriptsize $\lab$}}] (\al) at
        (\x,\y) {}; }

      \foreach \from/\to/\type in {%
        a/start1/dashed, %
        start1/t1/role, start2/t2/role, start3/t3/role, startm/tm/role, %
        s1/s2/role, s2/startm/role, vN1/vN/role, vN/row/role, row/end2/role,
        tree-t2.260/start3/dotted, t3/s1/dotted, tm/vN1/dotted%
      } {\draw[thick, \type] (\from) -- (\to);}

      \foreach \al/\lab/\wh in {%
        start1/y_0/left, t1/y_1/left, start2/y_{N+1}/above left,
        t2/y_{N+2}/left, start3/y_{n-N-1}/left, startm/y_n/above left,%
        t1/~I_0/right, tm/z_1/left, vN/z_{N}/left
      }{ \node[label={[inner sep=2]\wh:{\scriptsize $\lab$}}, inner sep=0] at
        (\al) {}; }

      \foreach \al/\lab/\wh in {%
        startm/\sigma/above right%
      }{\node[label={[inner sep=2,red]\wh:{\scriptsize $\lab$}}, inner sep=0] at (\al)
        {}; }

      \begin{pgfonlayer}{background}
        \foreach \from/\to/\in in {%
          y0/start2/200, y1/t2/200, yN/tree-t2.south/200, ylN2/startm/200,
          yl/row/170, yend/end2/180%
        } {\draw[homomorphism, black!40] (\from) to[out=-30, in=\in] (\to);}
      \end{pgfonlayer}

      \begin{pgfonlayer}{background}
        \foreach \start/\finish/\x/\i in {%
          start1/start2/1.3/Row 1, %
          start2/tree-t2.south/1.8/2,%
          start3/startm/1.3/$M-1$, %
          startm/row/1.8/$M$%
        } { \draw[gray, thick, <->] ({(\x,0)}|-\start) -- %
          node[right] {\small \i} ({(\x,0)}|-\finish); }
      \end{pgfonlayer}

      \node at (0.7,1) {$\M_r$};
    \end{scope}

\end{tikzpicture}

\caption{The two homomorphisms to two minimal models}
\label{two-homomorphisms}
\end{figure}


\smallskip%
$(\Rightarrow)$ Let $\q_n$ be such that $\prod \Mod_{\K_2} \models \q_n$, and so $\M \models \q_n$ for each $\M \in \Mod_{\K_2}$. Consider all the pairwise distinct pairs $(\M,h)$ such that $\M \in \Mod_{\K_2}$ and $h$ is a homomorphism from $\q_n$ to $\M$.
Note that $h(\q_n)$ contains an or-node $\sigma_h$ (which is an instance of $\Row^{\halt}_k$, for some $k$).
We call $(\M,h)$ and $h$ \emph{left} if $h(x_{n+1}) = \sigma_h \cdot w_{\exists R.\End}$, and \emph{right} otherwise.
It is not hard to see that there exist a left $(\M_\ell,h_\ell)$ and a right $(\M_r,h_r)$ with $\sigma_{h_\ell} = \sigma_{h_r}$ (if this is not the case, we can construct $\M \in \Mod_{\K_2}$ such that $\M\not\models\q_n$).

Take $(\M_\ell,h_\ell)$ and $(\M_r,h_r)$ such that $\sigma_{h_\ell} = \sigma_{h_r} =
\sigma$ and use them to construct the required tiling. Let $\sigma = a
w_0\cdots w_n$. We have $h_\ell(x_{n+1}) = \sigma \cdot w_{\exists R.\End}$ and
$h_\ell(x_{n}) = \sigma$. Let $h_r(x_{n+1}) = \sigma v_1\cdots v_{m+2}$, which is
an instance of $\End$. Then $h_r(x_n) = \sigma v_1\cdots v_{m+1}$, which is an
instance of $\Row$.

Suppose $v_{m} = w_{\exists R. T^{\halt}_2}$ (other $k$'s are treated
analogously). By~\eqref{end}, $\textit{right}(T) = \textit{wall}$; by
\eqref{horizontal-halt}, $\textit{up}(T) = \textit{ceiling}$.
Suppose $w_{n-1} = w_{\exists R. S_k}$. Then it must be that
$k=1$. By~\eqref{start-last-row}, $\textit{right}(S) = \textit{wall}$. Consider
the atom $B_{n-1}(x_{n-1})$ from $\q_n$. Then both $a w_0 \cdots w_{n-1}$ and
$\sigma v_1 \cdots v_{m}$ are instances of $B_{n-1}$. By \eqref{tile-hat} and
\eqref{vertical-halt}, $B_{n-1} = \widehat S_1$ and $\textit{down}(T) =
\textit{up}(S)$.

Suppose $v_{m-1} = w_{\exists R. U^{\halt}_2}$. By \eqref{horizontal-halt}, $\textit{right}(U) = \textit{left}(T)$ and $\textit{up}(U) = \textit{ceiling}$.
Suppose $w_{n-2} = w_{\exists R. Q_1}$. By \eqref{horizontal-k}, $\textit{right}(Q) = \textit{left}(S)$.
Consider the atom $B_{n-2}(x_{n-2})$ from $\q_n$. Then both $a w_0\cdots w_{n-2}$
and $\sigma v_1\cdots v_{m-1}$ are instances of $B_{n-2}$. By \eqref{tile-hat}
and \eqref{vertical-halt}, $B_{n-2} = \widehat Q_1$ and $\textit{down}(U) =
\textit{up}(Q)$.

We proceed in the same way until we reach $\sigma$ and $a w_0\cdots w_{n-N-1}$,
for $N = m$, both of which are instances of $B_{n-N-1} = \Row$. Thus have tiled
the two last rows of the grid. We proceed further and tile the whole $N \times
M$ grid, where $M = n/(N+1) + 1$.
\end{proof}

Note that $\K_2$ encodes tilings with at least 3 rows, hence, $M \geq 3$.

\medskip

We now define a KB $\K_1 = (\T_1,\{A(a)\})$.
Let
$
\Sigma_0 = \{\Row\} \cup \{\hat{T}_k \mid T \in \mathfrak{T}, k = 0,1,2\},
$
and let $\T_1$ contain the following axioms:
\begin{align}
  A \sqsubseteq &~ \exists P. D,\\
   D\sqsubseteq &~ \exists R.D ~\AND~ \exists R. \exists R. E ~\AND
  \bigsqcap_{X \in \Sigma_0} X \AND \Start,\\
  E \sqsubseteq &~ \exists R. E ~\AND~ \bigsqcap_{X \in \Sigma_0} X
  ~\AND~ \End.
\end{align}

As $\K_1$ is an $\EL$-KB, it has a canonical model $\M_{\K_1}$:
\begin{center}
  \begin{tikzpicture}[xscale=1.1, yscale=0.8]

    \foreach \al/\x/\y/\lab/\wh/\extra in {%
      a/0/1//right/constant, %
      x0/0/0/{\Start,\Sigma_0}/above left/, %
      x1/1.2/-1//right/, %
      x2/1.2/-2/{\End,\Sigma_{0}}/left/, %
      x3/1.2/-3/{\End,\Sigma_{0}}/left/, %
      y1/-1/-1/{\Start,\Sigma_0\qquad}/above/,%
      y2/-0.2/-2//right/,%
      y3/-0.2/-3/{\End,\Sigma_{0}}/left/,%
      y4/-0.2/-4/{\End,\Sigma_{0}}/left/,%
      z1/-2/-2/{\Start,\Sigma_0\qquad}/above/,%
      z2/-1.6/-3//right/,%
      z3/-1.6/-4/{\End,\Sigma_{0}}/left/,%
      z4/-1.6/-5/{\End,\Sigma_{0}}/left/,%
      w1/-3/-3/{\Start,\Sigma_0\qquad}/above/%
    }{ \node[point, \extra, label=\wh:{\scriptsize $\lab$}] (\al) at (\x,\y) {};
    }

    \foreach \from/\to/\type in {%
      a/x0/dashed, x0/x1/role, x1/x2/role, x2/x3/role,%
      x0/y1/role, y1/y2/role, y2/y3/role, y3/y4/role, %
      y1/z1/role, z1/z2/role, z2/z3/role, z3/z4/role, %
      z1/w1/role%
    } {\draw[role, \type] (\from) -- (\to);}

   \foreach \al/\lab/\wh in {%
     x2/\sigma_{\End}/right, x0/\sigma_{\Start}/right%
    } {\node[label={[red, inner sep=1]\wh:{\footnotesize $\lab$}}] at (\al) {};}

   \foreach \from in {%
     x3, y4, z4, w1%
    } {\draw[dotted, thick] (\from) -- +(0,-0.5);}

    \foreach \from/\i/\shift in {%
      x1/1/{0.2cm,0.3cm}, y2/2/{0.2cm,0.3cm}, z2/3/{0.2cm,0.3cm},
      w1/\omega/{-0.4cm,-0.1cm}%
    } {\node[shift=(\shift), text=blue] at (\from) {$\pi_\i$};}

  \end{tikzpicture}
\end{center}

Let $\Sigma$ be the signature of $\K_1$.
\begin{lemma}\label{qn-rejected-by-k1}
$\prod \Mod_{\K_2}$ is $n\Sigma$-homomorphically embeddable into $\M_{\K_1}$ for any $n$ iff there does not exist a CQ $\q_n$ such that $\prod \Mod_{\K_2} \models \q_n$.
\end{lemma}
\begin{proof}
$(\Rightarrow)$ Suppose $\prod \Mod_{\K_2} \models \q_n$ for some $n$. Since $\prod \Mod_{\K_2}$ is $n\Sigma$-homomorphically embeddable into $\M_{\K_1}$, we then have $\M_1 \models \q_n$, which is clearly impossible because of the $B_i$ and $\End$ in $\q_n$.

$(\Leftarrow)$ Suppose $\prod \Mod_{\K_2} \not\models \q_n$ for all CQs of the form $\q_n$. Take any subinterpretation of $\prod \Mod_{\K_2}$ whose domain contains $m$ elements. We can regard this subinterpretation as a Boolean $\Sigma$-CQ, and so denote it by $\q$. Without loss of generality we can assume that $\q$ is connected; clearly, $\q$ is tree-shaped. We know that there is no $\Sigma$-homomorphism from $\q_n$ into $\q$ for any $n$; in particular, $\q$ does not have a subquery of the form $\q_n$. We have to show that $\M_{\K_1} \models \q$.

Suppose $\q$ contains $A$ or $P$, then they appear at the root of $\q$ or, respectively, in the fist edge of $\q$. By the structure of $\K_2$, it follows then $\q$ does not contain $\End$ and, therefore,  can be mapped into $\pi_\omega$. In what follows, we assume that $\q$ does not contain $A$ and $P$.

If $\q$ does not contain $\Start$ atoms, or $\q$ does not contain $\End$ atoms,
then clearly, $\M_{\K_1} \models \q$.  In the former case, $\q$ can  be mapped
to $\pi_1$ by sending the root of $\q$ to $\sigma_{\End}$. In the latter
case, $\q$ can be mapped to $\pi_\omega$ by sending the root of $\q$ to
$\sigma_{\Start}$.

Assume that $\q$ contains both $\Start$ and $\End$ atoms.  If there exists a(n
$R$-)path from a $\Start$ node to an $\End$ node in $\q$, then by the structure
of $\K_2$, the $\Start$ node must be the root of $\q$. Since $\q$ does not
contain a subquery of the form $\q_n$, this $R$-path should contain variables
with the empty $\Sigma$-concept label, in which case $\q$ can be mapped into
some $\pi_i$, $1\leq i<\omega$, by mapping the root of $\q$ to $\sigma_{\Start}$.

Now, assume that in $\q$ there does not exist a path from a $\Start$ node to an
$\End$ node. Hence, the $\Start$ node is not the root of $\q$.
Let $\M$ be a minimal model of $\K_2$. Then the root $y_0$ of $\q$ should be
mapped to an element of the form $\delta \cdot w_{\SOME{R}{T^{\first}}}$ in
$\Delta^\M$, since there is a path from the root of $\q$ to a $\Start$ node.
By the structure of $\K_2$, the general form of $\q$ should be as follows:
\begin{align*}
  Q_{T_0} \AND{} & \SOME{R}{(Q_{\Start} \AND Q_{\text{noEnd}})} \AND{} \\
  & \SOME{R}{(Q_{T_0} \AND \SOME{R}{(Q_{\Start} \AND Q_{\text{noEnd}})}
    \AND{}\\
    & \qquad\SOME{R}{(Q_{T_0} \AND \SOME{R}{(Q_{\Start} \AND
        Q_{\text{noEnd}})}} \AND{}\\
    &\hspace{4cm}\cdots \AND \SOME{R}{Q_{\End}}))}
\end{align*}
where $Q_{\End}$ is an \EL concept constructed using $R$ and concepts in
$\Sigma_0\cup \{\End\}$, $Q_{\text{noEnd}}$ is an \EL concept constructed using
$R$ and concepts in $\Sigma_0$, $Q_{\Start}$ is either an empty query or a
$\Start$ atom, and $Q_{T_0}$ is either an empty query or a $\widehat{T}_0$
atom.
%
We prove that each path in $\q$ ending with an $\End$ node must have at least
one intermediate node with the empty $\Sigma$-concept label.

For simplicity assume that $\q$ consists of two subtrees $\q_{\End}$ and
$\q_{\Start}$, where $\q_{\End}$ is a path ending with an
$\End$ node, and $\q_{\Start}$ is a tree rooted in a $\Start$ node.
By contradiction, assume that each intermediate node in $\q_{\End}$ is labeled
with either some $\widehat{T}_k$ or $\Row$.
Since $\K_2\models \q_{\End}$ it follows that there is some $n$ such that
the distance between two neighbour $\Row$ nodes in $\q_{\End}$ is $n$.
Let $\M_\ell$ and $\M_r$ be minimal models that satisfy \eqref{disjunction} by
picking the first and the second disjunct, respectively, and identical,
otherwise. Assume that $\M_\ell$ satisfies $\q_{\End}$ by mapping $y_0$ to
$\sigma_l$ of the form $\delta \cdot w_{\SOME{R}{T^{\first}}}$ and $\M_r$
satisfies $\q_{\End}$ by mapping $y_0$ to $\sigma_r$ of the form $\sigma_l
\cdots w_{\SOME{R}{T^{\first}}}$. Then the distance between $\sigma_l$ and
$\sigma_r$ is $n$.
Let the distance from $y_0$ to the first $\Row$ node $y_m$ be $m$. Then $m$
should be less than or equal $n-1$. Therefore, $y_m$ should be mapped to
a predecessor $\sigma'$ of $\sigma_r$  in $\M_\ell$. However, such a mapping is
not a homomorphism as the $\Sigma$-label of $\sigma'$ does not contain $\Row$
(only, a concept of the form $\widehat{T}_0$).
Contradiction with the assumption that $\K_2 \models \q$ and that the label of
$y_l$ is non-empty.

Finally, we conclude that $\q$ can be mapped to $\M_1$ as follows: $y_0$ to
$\sigma_{\Start}$, $\q_{\Start}$ into $\pi_\omega$, and $\q_{\End}$ into
$\pi_i$, where the distance from $y_0$ to the first gap is $i$, for $1 \leq i
< |\q|$.
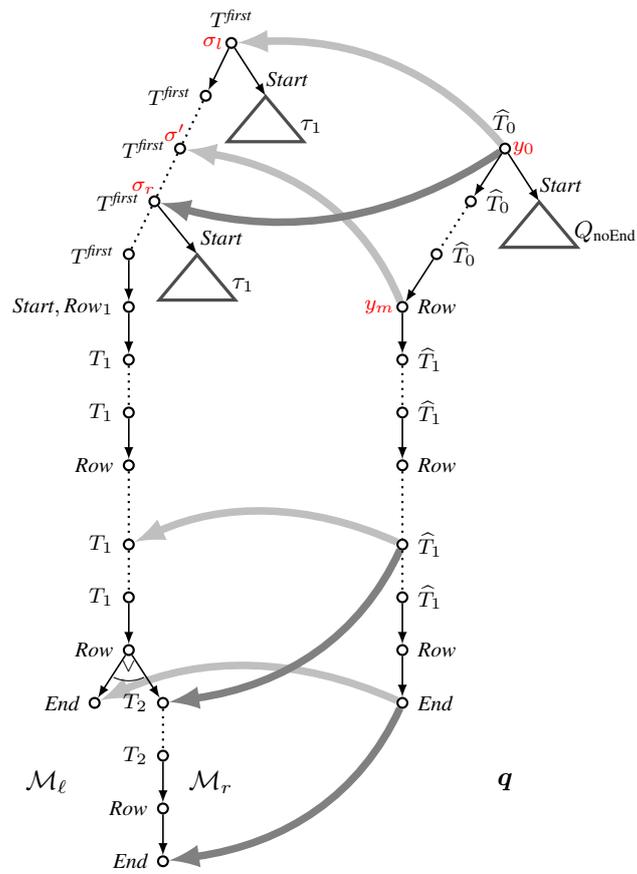
\begin{figure}
  \centering
  \scalebox{1}{\begin{tikzpicture}[xscale=0.9, yscale=0.7]
    \foreach \al/\x/\y/\type/\lab/\wh/\extra in {%
      n0/0/1/point/T^{\first}/90/,%
      n1/-0.375/0/point/T^{\first}/left/,%
      st0/0.5/0/subtree/{\qquad\Start}/90/,%
      n1a/-0.75/-1/point/T^{\first}/left/,%
      n2/-1.125/-2/point/T^{\first}/left/,%
      st2/-0.5/-3/subtree/{\qquad~\Start}/90/,%
      n3/-1.5/-3/point/T^{\first}/left/,%
      x0/-1.5/-4/point/{\Start,\Row_1}/left/,%
      x1/-1.5/-5/point/{T_1}/left/,%
      x2/-1.5/-6/point/{T_1}/left/,%
      x10/-1.5/-7/point/{\Row}/left/,%
      x11/-1.5/-8.5/point/{T_1}/left/,%
      x12/-1.5/-9.5/point/{T_1}/left/,%
      x20/-1.5/-10.5/point/{\Row}/left/,%
      x/-2/-11.5/point/{\End}/left/,%
      x21/-1/-11.5/point/{T_2}/left/,%
      x22/-1/-12.5/point/{T_2}/left/,%
      x30/-1/-13.5/point/{\Row}/left/,%
      xx/-1/-14.5/point/{\End}/left/%
    }{ \node[\type, \extra, label=\wh:{\footnotesize $\lab$}] (\al) at (\x,\y)
      {}; }

    \foreach \from/\to/\type in {%
      n0/n1/role, n0/st0.north/role, n1/n1a/dotted, n1a/n2/dotted,
      n2/n3/dotted, n2/st2.north/role, n3/x0/role, %
      x0/x1/role, x1/x2/dotted, x2/x10/role,%
      x10/x11/dotted, x11/x12/dotted, x12/x20/role, x20/x/role, %
      x20/x21/role, x21/x22/dotted, x22/x30/role, x30/xx/role%
    } {\draw[thick, \type] (\from) -- (\to);}

    \foreach \to/\al in {x/left,x21/right}{%
      \draw[draw=none] (x20) -- node[pos=0.5, inner sep=0] (\al) {} (\to); }%
    \draw (left) to[bend right] node[above] {\scriptsize $\lor$} (right);

    \foreach \wh in {%
      st0, st2%
    }{ \node[anchor=west] at (\wh.east) {\small $\tau_1$}; }

    \foreach \al/\lab/\wh in {%
      n0/\sigma_l/left, n1a/\sigma'~~/above, n2/\sigma_r~~~~/above%
    }{\node[label={[inner sep=2,red]\wh:{\footnotesize $\lab$}}, inner sep=0]
      at (\al) {}; }

    \node at (-2.7,-13) {$\M_\ell$};
    \node at (-0.3,-13) {$\M_r$};

    \begin{scope}[xshift=4cm, yshift=0cm]
      \foreach \al/\x/\y/\type/\lab/\wh/\extra in {%
        y0/-0/-1/point/\widehat{T}_0/90/,%
        y1/-0.5/-2/point/\widehat{T}_0/right/,%
        yst0/0.5/-2/subtree/{\qquad\Start}/90/,%
        y3/-1/-3/point/\widehat{T}_0/right/,%
        z0/-1.5/-4/point/\Row/right/,%
        z1/-1.5/-5/point/\widehat{T}_1/right/,%
        z2/-1.5/-6/point/{\widehat{T}_1}/right/,%
        z10/-1.5/-7/point/{\Row}/right/,%
        z11/-1.5/-8.5/point/{\widehat{T}_1}/right/,%
        z12/-1.5/-9.5/point/{\widehat{T}_1}/right/,%
        z20/-1.5/-10.5/point/{\Row}/right/,%
        end/-1.5/-11.5/point/{\End}/right/%
      }{ \node[\type, \extra, label=\wh:{\footnotesize $\lab$}] (\al) at
        (\x,\y) {}; }

      \foreach \from/\to/\type in {%
        y0/y1/role, y0/yst0.north/role, y1/y3/dotted, y3/z0/role, %
        z0/z1/role, z1/z2/dotted, z2/z10/role,%
        z10/z11/dotted, z11/z12/dotted, z12/z20/role, z20/end/role%
      } {\draw[thick, \type] (\from) -- (\to);}

      \foreach \wh in {%
        yst0%
      }{ \node[anchor=west] at (\wh.east) {\small $Q_{\text{noEnd}}$}; }

      \foreach \al/\lab/\wh in {%
        y0/y_0/right, z0/y_m/left%
      }{\node[label={[inner sep=2,red]\wh:{\footnotesize $\lab$}}, inner sep=0]
        at (\al) {}; }

      \node at (0,-13) {$\q$};
    \end{scope}

    \begin{pgfonlayer}{background}
      \foreach \from/\to in {%
        y0/n0,  z0/n1a, z11/x11, end/x%
      } {\draw[homomorphism, gray!50] (\from) to[bend right] (\to);}

      \foreach \from/\to in {%
        y0/n2, z11/x21, end/xx%
      } {\draw[homomorphism, gray] (\from) to[bend left] (\to);}
    \end{pgfonlayer}
  \end{tikzpicture}}
  \label{fig:empty-labels}
  \caption{A query that contains both $\Start$ and $\End$ atoms must have
    variables with empty concept labels.}
\end{figure}
\end{proof}

As an immediate consequence of the obtained results we have:

\smallskip
\noindent
{\bf Theorem~\ref{thm:undecidability}}~(\emph{i})
{\em The problem whether a \hALC KB $\Sigma$-CQ entails an $\ALC$ KB is undecidable.}

\medskip
\smallskip
\noindent
{\bf Theorem~\ref{thm:undecidability}}~(\emph{ii})
{\em $\Sigma$-CQ inseparability between \hALC and \ALC KBs is undecidable.}
\begin{proof}
Let $\K_2' = \K_2 \cup \K_1$. Then the following set $\Mod_{\K_2'}$ is complete for $\K_2'$:
$$
\Mod_{\K_2'} = \{ \M \uplus \M_{\K_1} \mid \M \in \Mod_{\K_2}\},
$$
where $\M \uplus \M_{\K_1}$ is the interpretation that results from merging the roots $a$ of $\M$ and $\M_{\K_1}$.
As before, we set $\Sigma = \sig(\K_1)$.
It suffices to show that $\K_1$ $\Sigma$-CQ entails $\K_2$ iff $\K_1$ and $\K_2'$ are $\Sigma$-CQ inseparable.

($\Leftarrow$) follows from $\K_2 \models \q(\avec{a}) \ \Rightarrow \ \K'_2 \models \q(\avec{a})$.

($\Rightarrow$) It follows from the definition that $\K'_2$ $\Sigma$-CQ entails $\K_1$. So we have to show that $\K_1$ $\Sigma$-CQ entails $\K'_2$. Suppose this is not the case and there is a $\Sigma$-CQ $\q$ such that $\K'_2 \models \q$ and $\K_1 \not\models \q$. We can assume  $\q$ to be a \emph{smallest connected} CQ with this property; in particular, no proper sub-CQ of $\q$ separates $\K_1$ and $\K'_2$.


Now, we cannot have $\K_2 \models \q$ because this would contradict the fact that $\K_1$ $\Sigma$-CQ entails $\K_2$. Then $\K_2 \not\models \q$, and so there is $\M \in \Mod_{\K_2}$ such that  $\M \not \models \q$. On the other hand, we have $\M \uplus \M_{\K_1} \models \q$. Take a homomorphism $h\colon \q \to \M \uplus \M_{\K_1}$. As $\q$ is connected, $\M \not \models \q$ and $\M_{\K_1} \not\models \q$, there is a variable $x$ in $\q$ such that $h(x) = a$. For every  variable $x$ with $h(x) = a$, we remove $\exists x$ from the prefix of $\q$ if any.
Denote by $\q'$ the maximal sub-CQ of $\q$ such that $h(\q') \subseteq \M$ (more precisely, $S(\avec{y})$ is in $\q'$ iff $h(\avec{y}) \subseteq \Delta^\M$). Clearly, $\q' \subsetneqq \q$ and $\K'_2 \models \q'$. Denote by $\q''$ the complement of $\q'$ to $\q$.  Now, we either have $\K_1 \models \q'$ or $\K_1 \not\models \q'$. The latter case contradicts the choice of $\q$ because $\q'$ is its proper sub-CQ. Thus, $\K_1 \models \q'$, and so there is a homomorphism $h' \colon \q' \to \M_{\K_1}$ with $h'(x) = a$ for every free variable $x$. Define a map $g \colon \q \to \M_{\K_1}$ by taking $g(y) = h'(y)$ if $y$ is in $\q'$ and $g(y) = h(y)$ otherwise. The map $g$ is a homomorphism because all the variables that occur in both $\q'$ and $\q''$ are free and must be mapped by $g$ to $a$. Therefore, $\M_{\K_1} \models \q$, which is a contradiction.
\end{proof}


\subsection{Proof of Theorem~\ref{thm:undecidability}~(\emph{i}) and~(\emph{ii}) for rCQs}

Let
\begin{align}\label{1}
\A = \{ R(a,a), \Row(a), A(a) \} \cup \{\widehat T_0(a) \mid T \in \mathfrak{T}\}.
\end{align}
$\T_2$ contains the following axioms, where $k = 0,1,2$:
\begin{align}\label{2}
& A \sqsubseteq \exists R. (\Row \sqcap \exists R. I_0),
\\
\label{r-7}
& \Row_k \sqsubseteq \exists R. T_k, \quad \text{for $T \in \mathfrak T$},\\
\label{r-horizontal-k}
& T_k \sqsubseteq \exists R.S_k, \quad \text{ if $\textit{right}(T) =
  \textit{left}(S)$ and $T,S \in \mathfrak T$}, \\
\label{r-end-of-row-k}
& T_k \sqsubseteq \exists R.\Row_{(k+1) \,\text{mod}\, 3}, \quad \text{if }
\textit{right}(T) = \textit{wall},\\
\label{r-start-last-row}
& T_k \sqsubseteq \exists R.\Row^{\halt}_{(k+1)\, \text{mod}\, 3},\quad
\text{if } \textit{right}(T) = \textit{wall},\\
\label{r-row-k}
& \Row_k \sqsubseteq \Row, \\
\label{r-tile-hat}
& T_k \sqsubseteq \widehat{T}_k,\quad \text{for $T \in \mathfrak T$},\\
\label{r-vertical-k}
& T_k \sqsubseteq \widehat{S}_{(k-1)\, \text{mod}\, 3}, \quad \text{if
  $\textit{down}(T) = \textit{up}(S)$, $T,S \in \mathfrak T$},
\\
\label{r-disjunction}
& \Row_k^{\halt} \sqsubseteq \exists R.{\End} \sqcup
\bigsqcap_{\textit{up}(T) = \textit{ceiling}} \exists R.T^{\halt}_k, \\
\label{r-horizontal-halt}
& T_k^{\halt} \sqsubseteq \exists R. S^{\halt}_k, \quad \text{if
  $\textit{right}(T) = \textit{left}(S)$ and $\textit{up}(S) =
  \textit{ceiling}$},\\
\label{r-end}
& T_k^{\halt} \sqsubseteq \exists R.(\Row \sqcap \exists R.{\End}), \quad
\text{if $\textit{right}(T) = \textit{wall}$},\\
\label{r-row-halt}
& \Row^{\halt}_k \sqsubseteq \Row,\\
\label{r-vertical-halt}
& T_k^{\halt} \sqsubseteq \widehat{S}_{(k-1)\, \text{mod}\, 3}, ~~ \text{if
  $\textit{down}(T) = \textit{up}(S)$, }  T,S \in \mathfrak{T}.
\end{align}

Let $\K_2 = (\T_2, \A)$. Consider a CQ $\q_n(X)$ of the form
\begin{equation*}
  \exists \vec{x} \big( R(X, x_0) \land \bigwedge_{i=0}^{n} \bigl( R(x_i,
  x_{i+1}) \land B_{i}(x_{i})\bigr) \land \End(x_{l+1}) \big)
\end{equation*}
where $B_{i} \in \{\Row\} \cup \{\widehat{T}_k \mid T
\in \mathfrak{T}, k=0,1,2\}$.

\begin{lemma}\label{qn-tiling-loop}
There exists a CQ $\q_n(X)$ such that $\prod \Mod_{\K_2} \models \q_n(a)$ iff there exist $N,M \in \mathbb N$ for which $\mathfrak T$ tiles the $N \times M$ grid as described above.
\end{lemma}
\begin{proof} $(\Leftarrow)$ Suppose $\mathfrak T$ tiles
  the $N \times M$ grid under which a tile of type $T^{ij} \in \mathfrak T$
  covers $(i,j)$. Let
$$
\textit{block}_j = (\widehat T^{1,j}_k, \dots, \widehat T^{N,j}_k, \Row),
$$
for $j=1,\dots,M-1$ and $k =(j-1) \mod 3$. Let $\q_n$ be the CQ in which the $B_i$ follow the pattern
$$
\Row, \ \textit{block}_1, \ \textit{block}_1, \ \textit{block}_2,  \dots, \ \textit{block}_{M-1}
$$
(thus, $n = (N+1) \times M +1$).  In view of Proposition 5 we only need to
prove $\M \models \q_n$ for each minimal model $\M \in \Mod_{\K_2}$. Take such
an $\M$. We have to show that there is an $R$-path $a,x_0,\dots,x_{n+1}$ in $\M$
such that $x_i \in B_i^\M$ and $x_{n+1} \in \End^\M$.

First, we construct an auxiliary $R$-path $y_0,\dots,y_{n-N-1}$.
We take $y_0 \in \Row^\M$ and $y_1 \in I_0^\M$ by~\eqref{2} ($I_0 = T^{1,1}$). Then we take $y_2 \in (T^{2,1})^\M, \dots, y_{N} \in (T^{N,1})^\M$ by~\eqref{r-horizontal-k}. We now have $\textit{right}(T^{N,1}) = \textit{wall}$. By~\eqref{r-end-of-row-k}, we obtain $y_{N+1} \in \Row_1$. By~\eqref{r-row-k}, $y_{N+1} \in \Row_1^\M \subseteq \Row^\M$. We proceed in this way, starting with \eqref{r-7}, till the moment we construct $y_{n-1} \in T^{N,M-1}$, for which we use~\eqref{r-start-last-row} and \eqref{r-row-halt} to obtain $y_n \in \Row^{\halt}_k \subseteq \Row^\M$, for some $k$. Note that $T^\M \subseteq \widehat T^\M$ by~\eqref{r-tile-hat}.

By~\eqref{r-disjunction}, two cases are possible now.

\emph{Case 1}: there is an object $y$ such that $(y_n,y) \in R^\M$ and $y \in \End^\M$. Then we take $x_0 = \dots = x_{N} = a$, $x_{N+1} = y_0, \dots, x_{n} = y_{n-N-1}, x_{n+1} = y$.

\emph{Case 2}: there is an object $z_1$ such that $(y_n,z_1) \in R^\M$ and $z_1 \in (T^{\it halt}_k)^\M$, where $T = T^{1,M}$ for which $\textit{up}(T) = \textit{ceiling}$. We then use~\eqref{r-horizontal-halt} and find objects $z_2,\dots,z_N,u,v$ such that $z_i \in (T^{\it halt}_k)^\M$, where $T = T^{i,M}$, $u \in \Row^\M$ and $v \in \End^\M$. We take $x_0 = y_0,\dots, x_{n-N-1} = y_{n-N-1}, x_{n-N} = z_1,\dots, x_{n-1}= z_N, x_n = u, x_{n+1} = v$. Note that, by~\eqref{r-vertical-k} and~\eqref{r-vertical-halt}, we have $(T^{i,j})^\M \subseteq (\widehat T^{i,j-1})^\M$.

$(\Rightarrow)$ Let $\q_n(X)$ be such that $\prod \Mod_{\K_2} \models \q_n(a)$,
by Proposition 5 it follows $\M \models \q_n$ for each $\M \in
\Mod_{\K_2}$. Consider all the pairwise distinct pairs $(\M,h)$ such that $\M
\in \Mod_{\K_2}$ and $h$ a homomorphism from $\q$ to $\M$.
Note that $h(\q)$ contains an or-node $\sigma_h$ (which is an instance of $\Row^{\halt}_k$, for some $k$). We call $(\M,h)$ and $h$ \emph{left} if $h(x_{n+1}) = \sigma_h \cdot w_{\exists R.\End}$, and \emph{right} otherwise.
It is not hard to see that there exist a left $(\M_\ell,h_\ell)$ and a right $(\M_r,h_r)$ with $\sigma_{h_\ell} = \sigma_{h_r}$ (if this is not the case, we can construct $\M\in\Mod_{\K_2}$ such that $\M\not\models\q$).

Take $(\M_\ell,h_\ell)$ and $(\M_r,h_r)$ such that $\sigma_{h_\ell} = \sigma_{h_r} = \sigma$ and use them to construct the required tiling. Let $\sigma = a w_0\cdots w_{n'}$. We have $h_\ell(x_{n}) = \sigma$, $h_\ell(x_{n+1}) = \sigma \cdot w_{\exists R.\End}$. Let $h_r(x_{n+1}) = \sigma v_1\dots v_{m+2}$, which is an instance of $\End$. Then $h_r(x_n) = \sigma v_1\dots v_{m+1}$, which is an instance of $\Row$.  Suppose $v_{m} = w_{\exists R. T^{\halt}_2}$ (other $k$'s are treated analogously). By \eqref{r-end}, $\textit{right}(T) = \textit{wall}$; by~\eqref{r-horizontal-halt}, $\textit{up}(T) = \textit{ceiling}$.
Suppose $w_{n'-1} = w_{\exists R. S_k}$. Now, we know that $k=1$. By \eqref{r-start-last-row}, $\textit{right}(S) = \textit{wall}$. Consider the atom $B_{n-1}(x_{n-1})$ from $\q$. Then both $a w_0\cdots w_{n'-1}$ and $\sigma v_1\cdots v_{m}$ are instances of $B_{n-1}$. By \eqref{r-tile-hat} and \eqref{r-vertical-halt}, $B_{n-1} = \widehat S_1$ and $\textit{down}(T) = \textit{up}(S)$.

Suppose $v_{m-1} = w_{\exists R. U^{\halt}_2}$. By \eqref{r-horizontal-halt}, $\textit{right}(U) = \textit{left}(T)$ and $\textit{up}(U) = \textit{ceiling}$.
Suppose $w_{n'-2} = w_{\exists R. Q_1}$. By \eqref{r-horizontal-k}, $\textit{right}(Q) = \textit{left}(S)$.
Consider the atom $B_{n-2}(x_{n-2})$ from $\q$. Then both $a w_0\cdots w_{n'-2}$ and $\sigma \cdots v_{m-1}$ are instances of $B_{n-2}$. By \eqref{r-tile-hat} and \eqref{r-vertical-halt}, $B_{n-2} = \widehat Q_1$ and $\textit{down}(U) = \textit{up}(Q)$.

We proceed in the same way until we reach $\sigma$ and $a w_0\cdots w_{n'-N-1}$, for $N=m$, both of which are instances of $B_{n-N-1} = \Row$. Thus we have tiled the last two rows of the grid. Let us proceed in that fashion until we have reached some variable $x_t$, for $t \geq 0$, of $\q$ that is mapped by $h_\ell$ to $a w_0 w_1$ (see Fig.~\ref{two-homomorphisms-rooted}). Note that this situation is guaranteed to occur. Indeed, $h_\ell(a) = a$, $h_\ell(x_0) \in \{ a, a w_0\}$, $h_\ell(x_1) \in \{ a, a w_0, a w_0 w_1\}$ etc. Clearly, assuming $h_\ell(x_i) \in \{ a, a w_0\}$ for all $0 \leq i \leq n+1$ produces a contradiction.

\begin{figure}
  \centering
  \scalebox{1}{\begin{tikzpicture}[yscale=1, %
    point/.style={thick,circle,draw=black,fill=white, minimum size=1mm,inner
      sep=0pt}%
    ]

    \begin{scope}[yshift=-0.75cm]
      \foreach \al/\x/\y/\lab/\wh/\extra in {%
        y0/0/3/{a}/above/, %
        yy0a/0/2//right/, %
        yy0/0/1.5/(x_0)/left/rowind, %
        yy1/0/1//right/, %
        yyNa/0/0.4//right/, %
        yyN/0/0/(x_{N+1})/left/rowind, %
        y1/0/-0.5/{~~~\textcolor{red}{x_t}}/right/, %
        yN/0/-1.5/{}/right/rowind, %
        yN1/0/-2//right/, %
        yN2/0/-3//right/rowind, %
        ylN2/0/-4/{(x_{n-N-1})}/left/rowind, %
        ylN1/0/-4.5//right/, %
        yl/0/-5.5/{(x_n)}/left/rowind, %
        yend/0/-6/\End/below/%
      }{ \node[point, \extra, label={[inner sep=1]\wh:{\scriptsize $\lab$}}] (\al) at (\x,\y)
        {}; }

      \foreach \from/\to/\type in {%
        y0/yy0a/dotted, yy0a/yy0/role, yy0/yy1/role, yy1/yyNa/dotted,
        yyNa/yyN/role, yyN/y1/role, y1/yN/dotted, yN/yN1/role, yN1/yN2/dotted,
        yN2/ylN2/dotted, ylN2/ylN1/role, ylN1/yl/dotted, 
        yl/yend/role%
      } {\draw[thick, \type] (\from) -- (\to);}

      \begin{pgfonlayer}{background}
        \foreach \first/\last/\ind in {%
          yy1/yyN/1, y1/yN/1, yN1/yN2/2, ylN1/yl/{M-1}%
        }{ \node[fit=(\first)(\last), rounded corners, fill=gray!30,
          label=right:{\scriptsize \textit{block}$_{\ind}$}] {}; }
      \end{pgfonlayer}
    \end{scope}

    \begin{scope}[xshift=-2.3cm]
      \foreach \al/\x/\y/\hei/\wid in {%
        t1/0/0/1cm/1.8cm, %
        t2/0.4/-1.5/1cm/1.8cm, %
        t3/0.1/-4/1.5cm/2.2cm%
      }{ \node[subtree, minimum height=\hei, minimum width=\wid] (tree-\al) at
        (\x,\y) {};%

        \node[point] (\al) at (\x,\y) {};%
      }

      \foreach \al/\x/\y/\lab/\wh/\extra in {%
        a/0/1/{D,\Row,\{\widehat{T}_0\}}/left/constant,%
        start1/0/0.5/\Row/right/rowind,%
        start2/0.4/-1//right/rowind,%
        start3/0.1/-3.5//right/rowind,%
        startm/-0.2/-5.5//right/rowind,%
        end1/-0.8/-6/\End/above left/,%
        s1/0/-4.6/Q_1/right/,%
        s2/-0.1/-5.1/S_1/right/%
      }{ \node[point, \extra, label={[inner sep=1]\wh:{\scriptsize $\lab$}}] (\al) at
        (\x,\y) {}; }

      \foreach \from/\to/\type in {%
        a/start1/role, %
        start1/t1/role, start2/t2/role, start3/t3/role, %
        s1/s2/role, s2/startm/role, startm/end1/role, %
        tree-t2.260/start3/dotted, t3/s1/dotted%
      } {\draw[thick, \type] (\from) -- (\to);}

    \draw[role] (a) to[out=130, in=50, looseness=25] (a);

      \foreach \al/\lab/\wh in {%
        start1/y_0/left, t1/y_1/left, start2/y_{N+1}/above left,
        t2/y_{N+2}/left, start3/y_{n-2N-2}/left, startm/y_{n-N-1}/above left,
        t1/~I_0/right%
      }{ \node[label={[inner sep=2]\wh:{\scriptsize $\lab$}}, inner sep=0] at
        (\al) {}; }

      \foreach \al/\lab/\wh in {%
        startm/\sigma/above right, end1/\sigma y_{\exists R. \End}/272%
      }{\node[label={[inner sep=2,red]\wh:{\scriptsize $\lab$}}, inner sep=0] at (\al)
        {}; }


      \node at (-0.7,2) {$\M_\ell$};

    \begin{pgfonlayer}{background}
      \foreach \from/\to/\in in {%
        y0/a/50, yyNa/a/10,%
        yyN/start1/0, y1/t1/0, yN/start2/-30, ylN2/start3/0, yl/startm/-15,
        yend/end1/0%
      } {\draw[homomorphism, black!30] (\from) to[out=160, in=\in] (\to);}
    \end{pgfonlayer}
    \end{scope}

    \begin{scope}[xshift=2.7cm, yshift=-0cm]
      \foreach \al/\x/\y/\hei/\wid in {%
        t1/0/0/1cm/1.8cm, %
        t2/0.4/-1.5/1cm/1.8cm, %
        t3/0.1/-4/1.5cm/2.2cm, %
        tm/0.3/-6/1.5cm/2.6cm%
      }{ \node[subtree, minimum height=\hei, minimum width=\wid] (tree-\al) at
        (\x,\y) {};%
        \node[point] (\al) at (\x,\y) {};%
      }

      \foreach \al/\x/\y/\lab/\wh/\extra in {%
        a/0/1//right/constant,%
        start1/0/0.5/\Row/right/rowind,%
        start2/0.4/-1//right/rowind,%
        start3/0.1/-3.5//right/rowind,%
        startm/-0.2/-5.5//right/rowind,%
        s1/0/-4.6/Q_1/right/,%
        s2/-0.1/-5.1/S_1/right/,%
        vN1/0.1/-6.6/U_2^{\halt}/right/,%
        vN/0.3/-7.1/T_2^{\halt}/right/,%
        row/0.5/-7.5//above right/rowind,%
        end2/0.5/-8/\End/right/%
      }{ \node[point, \extra, label={[inner sep=1]\wh:{\scriptsize $\lab$}}] (\al) at
        (\x,\y) {}; }

      \foreach \from/\to/\type in {%
        a/start1/role, %
        start1/t1/role, start2/t2/role, start3/t3/role, startm/tm/role, %
        s1/s2/role, s2/startm/role, vN1/vN/role, vN/row/role, row/end2/role,
        tree-t2.260/start3/dotted, t3/s1/dotted, tm/vN1/dotted%
      } {\draw[thick, \type] (\from) -- (\to);}

    \draw[role] (a) to[out=130, in=50, looseness=25] (a);

      \foreach \al/\lab/\wh in {%
        start1/y_0/left, t1/y_1/left, start2/y_{N+1}/above left,
        t2/y_{N+2}/left, start3/y_{n-2N-2}/left, startm/y_{n-N-1}/above left,%
        t1/~I_0/right, tm/z_1/left, vN/z_{N}/left
      }{ \node[label={[inner sep=2]\wh:{\scriptsize $\lab$}}, inner sep=0] at
        (\al) {}; }

      \foreach \al/\lab/\wh in {%
        startm/\sigma/above right%
      }{\node[label={[inner sep=2,red]\wh:{\scriptsize $\lab$}}, inner sep=0] at (\al)
        {}; }

      \begin{pgfonlayer}{background}
        \foreach \from/\to/\in in {%
          y0/a/170, yy0a/a/200, yy0/start1/200, yy1/t1/200, %
          yyN/start2/200, y1/t2/200, yN/tree-t2.south/200, ylN2/startm/200,
          yl/row/170, yend/end2/180%
        } {\draw[homomorphism, black!60] (\from) to[out=-20, in=\in] (\to);}
      \end{pgfonlayer}

      \begin{pgfonlayer}{background}
        \foreach \start/\finish/\x/\i in {%
          start1/start2/1.3/Row 1, %
          start2/tree-t2.south/1.8/2,%
          start3/startm/1.3/$M-1$, %
          startm/row/1.8/$M$%
        } { \draw[gray, thick, <->] ({(\x,0)}|-\start) -- %
          node[right] {\small \i} ({(\x,0)}|-\finish); }
      \end{pgfonlayer}

      \node at (0.7,2) {$\M_r$};
    \end{scope}
\end{tikzpicture}}

\caption{The two homomorphisms to two minimal models}
\label{two-homomorphisms-rooted}
\end{figure}
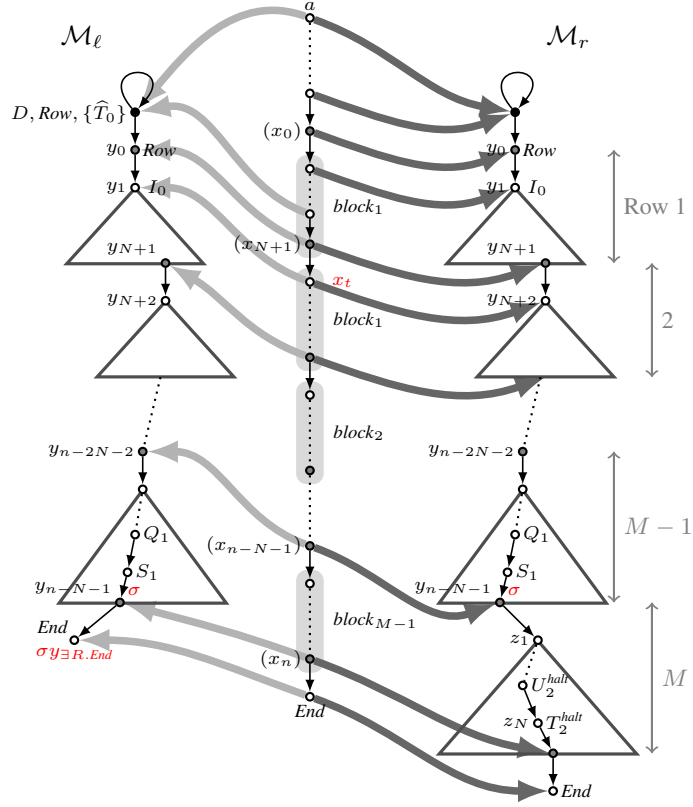



Let $h_r(x_t) = a w_0 \cdots w_s$ for some $s > 1$ and note that $s =
N+2$. By~\eqref{2}, it follows that $a w_0 w_1$ is an instance of $I_0$
therefore $B_t = \hat{I}_0$ and, by~\eqref{r-vertical-k}, we also get that $a
w_0\cdots w_s$ is an instance of $V_1$ for some tile $V$ such that
$\textit{down}(V) = \textit{up}(I)$. Thus, we have the tiling as required since
the vertical and horizontal compatibility of the tiles is ensured by the
construction above and by the fact that the tile $I$ occurs in it as the
initial tile.
\end{proof}

Let $\Sigma_0 = \{\Row\} \cup \{\hat{T}_k \mid T \in \mathfrak{T}, k = 0,
1,2\}$. Set $\K_1 = (\T_1,\A)$
and $\T_1$ to contain the following axioms:
\begin{align}
   A\sqsubseteq &~ \exists R.D ~\AND~ \exists R. \exists R. E ~\AND
  \bigsqcap_{X \in \Sigma_0} X,\\
  E \sqsubseteq &~ \exists R. E ~\AND~ \bigsqcap_{X \in \Sigma_0} X
  ~\AND~ \End.
\end{align}

The canonical model $\M_{\K_1}$ of $\K_1$ is as follows:
\begin{center}
  \begin{tikzpicture}[xscale=1.1, yscale=0.8]

    \foreach \al/\x/\y/\lab/\wh/\extra in {%
      a/0/0/{\Row,\widehat{T}_0}/right/constant, %
      x1/1.2/-1//right/, %
      x2/1.2/-2/{\End,\Sigma_{0}}/left/, %
      x3/1.2/-3/{\End,\Sigma_{0}}/left/, %
      y1/-1/-1/{\Sigma_0\qquad}/above/,%
      y2/-0.3/-2//right/,%
      y3/-0.3/-3/{\End,\Sigma_{0}}/left/,%
      y4/-0.3/-4/{\End,\Sigma_{0}}/left/,%
      z1/-2/-2/{\Sigma_0\qquad}/above/%
    }{ \node[point, \extra, label=\wh:{\scriptsize $\lab$}] (\al) at (\x,\y) {};
    }

    \foreach \from/\to/\type in {%
      a/x1/role, x1/x2/role, x2/x3/role,%
      a/y1/role, y1/y2/role, y2/y3/role, y3/y4/role, %
      y1/z1/role
    } {\draw[role, \type] (\from) -- (\to);}

    \draw[role] (a) to[out=150, in=40, looseness=20] (a);

   \foreach \from in {%
     x3, y4, z1
    } {\draw[dotted, thick] (\from) -- +(0,-0.5);}

    \foreach \from/\i/\shift in {%
      x1/1/{0.2cm,0.3cm}, y2/2/{0.2cm,0.3cm}, z1/\omega/{-0.4cm,-0.1cm}%
    } {\node[shift=(\shift), text=blue] at (\from) {$\pi_\i$};}

  \end{tikzpicture}
\end{center}

As before, let $\Sigma = \sig(\K_1)$.

\begin{lemma}\label{qn-rejected-by-k1-r}
$\prod \Mod_{\K_2}$ is $n\Sigma$-homomorphically embeddable into $\M_{\K_1}$ for any $n$ iff there does not exist a CQ $\q_n$ such that $\prod \Mod_{\K_2} \models \q_n$.
\end{lemma}
\begin{proof}
$(\Rightarrow)$ Suppose $\prod \Mod_{\K_2} \models \q_n(a)$ for some $n$. Since $\prod \Mod_{\K_2}$ is $n\Sigma$-homomorphically embeddable into $\M_{\K_1}$, we then have $\M_{\K_1} \models \q_n(a)$, which is clearly impossible because of the $B_i$ and $\End$ in $\q_n$.

$(\Leftarrow)$ Suppose $\prod \Mod_{\K_2} \not\models \q_n(a)$ for all $n$. Take any subinterpretation of $\prod \Mod_{\K_2}$ whose domain contains $m$ elements. We can regard this subinterpretation as a Boolean $\Sigma$-CQ, and so denote it by $\q$. Without loss of generality we can assume that $\q$ is connected; clearly, $\q$ is either:
 \begin{description}
 \item[(i)] tree shaped with a root different from $a$,
 \item[(ii)] tree shaped rooted in $a$ and containing a loop $R(a, a)$
\end{description}
We know that there is no $\Sigma$-homomorphism from $\q_n$ into $\q$ for any $n$; in particular, $\q$ does not have a subquery of the form $\q_n$. We have to show that $\M_{\K_1} \models \q$.

If \textbf{(i)} holds we map $\q$ to the branch $\pi_1$ in the obvious way. Suppose, \textbf{(ii)} holds. We will show how to map $\q$ starting from $a$. We call a variable $x$ in $\q$ a \emph{gap} if there exists no $A \in \Sigma$ such that $A(x)$ is in $\q$. By the condition of the lemma we know that every path $\rho$ in $\q$ either:
\begin{description}
 \item[(a)] does not contain $\End(x)$, or
 \item[(b)] contains $\End(x)$ and contains a gap $y$ that occurs between the root $a$ and $x$
\end{description}
  For the paths $\rho$ of type \textbf{(b)} let $t_\rho$ be the minimal distance from the root $a$ to a gap of the path $\rho$. Denote by $\mathcal{R}$ the set of all path $\rho$ of $\q$. If all $\rho \in \mathcal{R}$ are of type \textbf{(a)} we map $\q$ on the path $\pi_\omega$. Otherwise, let $t_0$ be the minimal number of all the $t_\rho$ (that are defined) and $\mathcal{R}_{t_0}$ the set of paths $\rho$ such that $t_\rho=t_0$. We map all the path of $\mathcal{R}_{t_0}$ to the path $\pi_{t_0}$ of $\M_{\K_1}$. For the rest $\mathcal{R} \setminus \mathcal{R}_{t_0}$ we find again the minimal number $t_1$ of all the $t_\rho$ for $\rho \in \mathcal{R} \setminus \mathcal{R}_{t_0}$ and denote by $\mathcal{R}_{t_1}$ the set of paths $\rho$ such that $t_\rho=t_1$. Clearly, can map all the paths in $\mathcal{R}_{t_1}$ to $\pi_{t_1}$. We continue in that way for sufficiently many steps to map all the paths of $\mathcal{R}$.
\end{proof}

We now obtain Theorem~\ref{thm:undecidability}~(\emph{i}) and~(\emph{ii}) for rCQs in the same way as in the previous section.

\subsection{Proof of Theorem~\ref{thm:undecidability}~(\emph{iii})}

To prove undecidability results if separating CQs can have arbitrary symbols we
modify the KBs introduced above. We follow \cite{lutz2012non} and replace the non-$\Sigma$-symbols by complex
\ALC-concepts that, in contrast to concept names, cannot occur in CQs. In
detail, consider a set $\Sigma_{\sf hide}$ of concept names
and take a fresh concept name $Z_{B}$ and fresh role names $r_{B}$ and $s_{B}$
for every $B\in \Sigma_{\sf hide}$. Now let for each $B\in \Sigma_{\sf hide}$
$$
H_{B}=  \forall r_{B}.\exists s_{B}.\neg Z_{B}
$$
and set
$$
\T_{\Sigma_{\sf hide}} = \{ \top \sqsubseteq \exists r_{B}.\top, \top \sqsubseteq \exists s_{B}.Z_{B}\mid B \in \Sigma_{\sf hide}\}
$$
Note that $\T_{\Sigma_{\sf hide}}$ is an $\mathcal{EL}$ TBox that generates trees with edges $r_{B}$ and $s_{B}$
such that the $s_{Z}$-successors satisfy $Z_{B}$. One can satisfy $H_{B}$ in a certain node by introducing in addition
to the $s_{B}$-successors satisfying $Z_{B}$ other $s_{B}$-successors not satisfying $Z_{B}$. Those additional
$s_{B}$-successors will not influence the answers to CQs. We now summarize the main properties of $\T_{\Sigma_{\sf hide}}$ in a formal way.
For an ABox $\A$ and any set $p(\Sigma_{\sf hide}) = \{J_{B} \mid B\in \Sigma_{\sf hide}\}$ with $J_{B}\subseteq \ind(\A)$ for all $B\in \Sigma_{\sf hide}$
construct a model $\I$
as follows: $\Delta^{\I}$ is the set of words $w=a v_{1}\cdots v_{n}$
such that $a\in \ind(\A)$ and $v_{i}\in \{r_{B},s_{B},\bar{s}_{B}\mid B \in \Sigma_{\sf hide}\}$
where $v_{i}\not=\bar{s}_{B}$ if (i) $i>2$ or (ii) $i=2$ and ($a\not\in J_{B}$ or $v_{1}\not=r_{B}$).
For all concept names $A$ not of the form $Z_{B}$ set
$$
A^{\I} = \{ a \mid A(a) \in \A\}
$$
Let for $B\in \Sigma_{\sf hide}$:
$$
Z_{B}^{\I} = \{ a \mid Z_{B}(a)\in \A\} \cup \{w \mid {\sf tail}(w) = s_{B}\}
$$
where ${\sf tail}(w)$ is the last symbol in $w$. For all role names $R$ not of the form $r_{B}$ or $s_{B}$ set
$$
R^{\I} = \{(a,b) \mid R(a,b)\in \A\}
$$
Finally, let for $B\in \Sigma_{\sf hide}$:
\begin{eqnarray*}
r_{B}^{\I} &= &\{ (a,b) \mid r_{B}(a,b)\in \A\}\cup \{(w,wr_{B}) \mid wr_{B}\in \Delta^{\I}\}\\				
s_{B}^{\I} &= & \{ (a,b) \mid s_{B}(a,b)\in \A\}\cup \{(w,ws_{B}) \mid wr_{B}\in \Delta^{\I}\} \cup \{(w,w\bar{s}_{B}) \mid w\bar{s}_{B}\in \Delta^{\I}\}.
\end{eqnarray*}					
The following result summarizes the main properties of $\I$ \cite{lutz2012non}.
\begin{lemma}\label{lem:hide}
The following holds for every $\A$ and $p(\Sigma_{\sf hide})$:
\begin{itemize}
\item $\I$ is a model of $\T_{\Sigma_{\sf hide}}$ and $\A$;
\item $J_{B}= (H_{B})^{\I}$ for all $B\in \Sigma_{\sf hide}$;
\item for every CQ $q(\vec{x})$ and $\vec{a}$ in $\ind(\A)$:
\qquad $
\T_{\Sigma_{\sf hide}},\A \models q(\vec{a}) \quad \Leftrightarrow \quad \I \models q(\vec{a})
$
\end{itemize}
\end{lemma}
A \emph{hiding scheme} $\mathcal{H}$ consists of three sets of concept names, $\Sigma_{\sf in}$, $\Sigma_{\sf out}$, and $\Sigma_{\sf hide}$.
Let $C^{\Sigma_{\sf hide}}$ be the result of replacing in a concept $C$ every $B\in \Sigma_{\sf hide}$ by $H_{B}$.
For a given TBox $\T$ we denote by $\T^{\mathcal{H}}$ the TBox containing $\T_{\Sigma_{\sf hide}}$ and the following CIs:
\begin{itemize}
\item $A \sqsubseteq H_{A}$, for $A\in \Sigma_{\sf in}$;
\item $C^{\Sigma_{\sf hide}} \sqsubseteq D^{\Sigma_{{\sf hide}}}$, for all $C \sqsubseteq D\in \T$;
\item $H_{A} \sqsubseteq A$, for all $A\in \Sigma_{\sf out}$.
\end{itemize}
A TBox $\T$ \emph{admits trivial models} if the singleton interpretation
in which all concept and role names are interpreted by the empty set is a model of $\T$.
We consider TBoxes that admit trivial models since for such TBoxes the nodes generated by $\T_{\Sigma_{\sf hide}}$
trivially satisfy $\T$. Oberve that the TBoxes constructed in the undecidability proofs
above all admit trivial models.
\begin{theorem}\label{hide1}
The problem whether a \hALC KB full signature-CQ entails an $\ALC$ KB is undecidable.
\end{theorem}
\begin{proof}
We consider the KBs $\K_{1}=(\T_{1},\A)$ and $\K_{2}=(\T_{2},\A)$ and $\Sigma={\sf sig}(\K_{1})$
constructed in the proof of Theorem~\ref{thm:undecidability} (\emph{i}) for $\Sigma$-CQ-entailment.

Define a hiding scheme $\mathcal{H}$ by setting
\begin{itemize}
\item $\Sigma_{\sf in} = {\sf sig}(\A) = \{A\}$;
\item $\Sigma_{\sf out}$ is the set of concept names in $\Sigma$;
\item $\Sigma_{\sf hide} = {\sf sig}(\K_{1}\cup \K_{2})$.
\end{itemize}
Define new KBs as follows:
$
\K_{1}'= (\T_{1}\cup \T_{\Sigma_{\sf hide}},\A)$, $\K_{2}'=(\T_{2}^{\mathcal{H}},\A).
$
Using the facts that
\begin{itemize}
\item ${\sf sig}(\A) \subseteq \Sigma$;
\item all role names in $\K_{2}$ are contained in $\Sigma$;
\item $\T_{1}$ and $\T_{2}$ admit trivial models
\end{itemize}
it is straightforward to check that $\K_{1}$ $\Sigma$-CQ-entails $\K_{2}$ iff
$\K_{1}'$ full signature CQ-entails $\K_{2}'$
\end{proof}
\begin{theorem}\label{hide2}
The problem whether a \hALC KB is full signature-CQ inseparable from an $\ALC$ KB is undecidable.
\end{theorem}
\begin{proof}
We consider the KBs $\K_{1}=(\T_{1},\A)$ and $\K_{2}'=\K_{1}\cup \K_{2}$ and the signature $\Sigma={\sf sig}(\K_{1})$
constructed in the proof of Theorem~\ref{thm:undecidability} (\emph{ii}) for $\Sigma$-CQ-inseparability. Assume
$\K_{1}=(\T_{1},\A)$ and $\K_{2}=(\T_{2},\A)$.

Consider the same hiding scheme $\mathcal{H}$ as in the proof of Theorem~\ref{hide1}:
\begin{itemize}
\item $\Sigma_{\sf in} = {\sf sig}(\A) = \{A\}$;
\item $\Sigma_{\sf out}$ is the set of concept names in $\Sigma$;
\item $\Sigma_{\sf hide} = {\sf sig}(\K_{1}\cup \K_{2})$.
\end{itemize}
Define new KBs $\K_{1}^{\ast}$ and $\K_{2}^{\ast}$ as follows:
$
\K_{1}^{\ast}= (\T_{1}\cup \T_{\Sigma_{\sf hide}},\A)$, $\K_{2}^{\ast}=(\T_{1} \cup \T_{2}^{\mathcal{H}},\A).
$
Using the facts that
\begin{itemize}
\item ${\sf sig}(\K_{1}) \subseteq \Sigma$;
\item all role names in $\K_{1}\cup \K_{2}$ are contained in $\Sigma$;
\item $\T_{1}$ and $\T_{2}$ admit trivial models
\end{itemize}
it is straightforward to check that $\K_{1}$ and $\K_{2}'$ are $\Sigma$-CQ-inseparable iff
$\K_{1}^{\ast}$ and $\K_{2}^{\ast}$ are full signature CQ-inseparable.
\end{proof}
\begin{theorem}\label{hide3}
The problem whether a \hALC KB full signature-rCQ entails an $\ALC$ KB is undecidable.
\end{theorem}
\begin{proof}
We consider the KBs $\K_{1}=(\T_{1},\A)$ and $\K_{2}=(\T_{2},\A)$ and $\Sigma={\sf sig}(\K_{1})$
constructed in the proof of Theorem~\ref{thm:undecidability} (\emph{i}) for $\Sigma$-rCQ-entailment.

Define a hiding scheme $\mathcal{H}$ by setting
\begin{itemize}
\item $\Sigma_{\sf in} = {\sf sig}(\A) = \{R,\Row,A\} \cup \{\widehat{T}_{0} \mid T\in \mathfrak{T}\}$;
\item $\Sigma_{\sf out}$ is the set of concept names in $\Sigma$;
\item $\Sigma_{\sf hide} = {\sf sig}(\K_{1}\cup \K_{2})$.
\end{itemize}
Define new KBs as follows:
$
\K_{1}'= (\T_{1}\cup \T_{\Sigma_{\sf hide}},\A)$, $\K_{2}'=(\T_{2}^{\mathcal{H}},\A).
$
Using the facts that
\begin{itemize}
\item ${\sf sig}(\A) \subseteq \Sigma$;
\item all role names in $\K_{2}$ are contained in $\Sigma$;
\item $\T_{1}$ and $\T_{2}$ admit trivial models
\end{itemize}
it is straightforward to check that $\K_{1}$ $\Sigma$-rCQ-entails $\K_{2}$ iff
$\K_{1}'$ full signature rCQ-entails $\K_{2}'$
\end{proof}
The proof of the following results is now
similar to the proof of Theorem~\ref{hide2} using the KBs constructed in the proof of Theorem~\ref{thm:undecidability} (\emph{ii}) for $\Sigma$-rCQ-inseparability.
\begin{theorem}
The problem whether a \hALC KB is 
full signature-rCQ inseparable from an $\ALC$ KB is undecidable.
\end{theorem}

\section{Proof of Theorem~\ref{thm:TBoxundecidability}}

\subsection{Proof of Theorem~\ref{thm:TBoxundecidability}~(\emph{i}) and~(\emph{ii}) for CQs}

We formulate the result again.
\begin{theorem} Let $\Theta=(\Sigma_{1},\Sigma_{2})$.

(\emph{i}) The problem of whether a \hALC TBox $\Theta$-CQ-entails
an \ALC TBox is undecidable.

(\emph{ii}) $\Theta$-CQ inseparability between \hALC TBoxes and \ALC TBoxes is undecidable. 

(\emph{iii}) $\Theta$-CQ inseparability between \hALC TBoxes and \ALC TBoxes is undecidable for $\Sigma_{1}=\Sigma_{2}$.
\end{theorem}
\begin{proof}
We prove (\emph{i}). The proof of (\emph{ii}) is similar.
Let $\K_{1}=(\T_{1},\A)$ and $\K_{2}=(\T_{2},\A)$ be the KBs and $\Sigma$ be the signature from the proof of Theorem~\ref{thm:undecidability} (\emph{i}) for 
$\Sigma$-CQ-entailment. Recall that $\A = \{A(a)\}$. Let $\Sigma_{1}=\{A\}$, $\Sigma_{2}=\Sigma$, and $\Theta=(\Sigma_{1},\Sigma_{2})$.
We claim that $\T_{1}$ $\Theta$-CQ-entails $\T_{2}$ iff $\K_{1}$ $\Sigma$-CQ-entails $\K_{2}$. Clearly, if $\K_{1}$ does not $\Sigma$-CQ-entail $\K_{2}$, then we have found a 
$\Sigma_{1}$-ABox $\A$ that witnesses that $\T_{1}$ does not $\Theta$ CQ-entail $\T_{2}$. Conversely, observe that all $\Sigma_{1}$-ABoxes $\A'$ are sets of assertions of the form $A(b)$ and so if any such $\A'$ provides a counterexample for $\Theta$-CQ-entailment between $\T_{1}$ and $\T_{2}$, then $\A$ does.

We now prove (\emph{iii}). Consider $\K_{1}$ and $\K_{2}'= (\T_{1}\cup \T_{2},\A)$ from the proof of Theorem~\ref{thm:undecidability} (\emph{ii}) for 
$\Sigma$-CQ-inseparability. Now let 
$$
\Sigma=\{A,R,\Row,\End,\Start\} \cup \{\hat{T}_k \mid T \in \mathfrak{T}, k = 0,1,2\}
$$
Then one can show that $\T_{1}$ and $\T_{1}\cup\T_{2}$ are $(\Sigma,\Sigma)$-CQ-inseparable iff
$\K_{1}$ and $\K_{2}'$ are $\Sigma$-CQ-inseparable. The latter is undecidable.
\end{proof}

\subsection{Proof of Theorem~\ref{thm:TBoxundecidability} for full ABox signature and CQs}
We now aim to extend the result above to the full ABox signature case and inseparability.

\begin{theorem}
(\emph{i}) The problem of whether a \hALC TBox full ABox signature $\Sigma$-CQ-entails
an \ALC TBox is undecidable.

(\emph{ii}) Full ABox signature $\Sigma$-CQ inseparability between \hALC TBoxes and \ALC TBoxes is undecidable.
\end{theorem}
\begin{proof}
We consider the inseparability case.
Let $\K_{1}=(\T_{1},\A)$ and $\K_{2}'= (\T_{1}\cup \T_{2},\A)$ be the KBs and $\Sigma={\sf sig}(\K_{1})$ be the signature 
from the proof of Theorem~\ref{thm:undecidability} (\emph{ii})
for $\Sigma$-CQ-inseparability between KBs. We set 
$$
\Sigma_{0}=\{R,\Row,\End,\Start\} \cup \{\hat{T}_k \mid T \in \mathfrak{T}, k = 0,1,2\}
$$
Observe that for any signature $\Gamma$ between $\Sigma_{0}$ and $\Sigma_{0}\cup {\sf sig}(\K_{1})$, the KBs $\K_{1}$ and $\K_{2}'$
are $\Gamma$-CQ-inseparable iff they are $\Sigma_{0}$-CQ-inseparable.
We construct TBoxes $\T_{1}^{\ast}$ and $\T_{2}^{\ast}$ from the TBoxes $\T_{1}$ and $\T_{2}$
such that full ABox signature $\Sigma_{0}$-CQ-inseparability between $\T_{1}^{\ast}$ and $\T_{2}^{\ast}$ is undecidable. 
To this end define a hiding scheme $\mathcal{H}$ by setting
\begin{itemize}
\item $\Sigma_{\sf in}= \{A\}$;
\item $\Sigma_{\sf out}$ is the set of concept names in $\Sigma_{0}$;
\item $\Sigma_{\sf hide} = {\sf sig}(\K_{2})$.
\end{itemize}
Define TBoxes $\T_{1}^{\ast}$ and $\T_{2}^{\ast}$ by setting
$$
\T_{1}^{\ast}= \T_{1} \cup \T_{\Sigma_{\sf hide}}, \quad \T_{2}^{\ast}= \T_{1} \cup \T_{2}^{\mathcal{H}}
$$
Now one can prove that $\K_{1}$ and $\K_{2}'$ are $\Sigma$-CQ-inseparable iff
$\T_{1}^{\ast}$ and $\T_{2}^{\ast}$ are full ABox signature $\Sigma_{0}$-CQ-inseparable.
The direction from right to left is trivial as we can take the ABox $\A$ as a witness separating $\T_{1}^{\ast}$ and $\T_{2}^{\ast}$
if $\K_{1}$ and $\K_{2}$ are $\Sigma_{0}$-CQ-separable. For the converse direction assume that an ABox $\A'$
$\Sigma_{0}$-CQ-separates $\T_{1}^{\ast}$ and $\T_{2}^{\ast}$. As $P\not\in \Sigma_{0}$ one can then prove that  
there exists $A(b)\in \A'$ such that $\{A(b)\}$ is an ABox that separates $\T_{1}^{\ast}$ and $\T_{2}^{\ast}$. But
then $\K_{1}$ and $\K_{2}'$ are $\Sigma_{0}$-CQ-separable as well.
\end{proof}

\subsection{Proof of Theorem~\ref{thm:TBoxundecidability}~(\emph{i}) and~(\emph{ii}) for rCQs}
We state the result again.
\begin{theorem}
(i) The problem of whether a \hALC TBox $\Theta$-rCQ-entails
an \ALC TBox is undecidable.

(ii) $\Theta$-rCQ inseparability between \hALC TBoxes and \ALC TBoxes is undecidable.
\end{theorem}
\begin{proof}
We consider the inseparability case.
Let $\K_{1}=(\T_{1},\A)$ and $\K_{2}'= (\T_{1}\cup \T_{2},\A)$ be the KBs  
from the proof of Theorem~\ref{thm:undecidability} (\emph{ii})
for $\Sigma$-rCQ-inseparability between KBs. 
Let $\Sigma_{1}={\sf sig}(\A)$ and 
$$
\Sigma_{2}=\{R,\Row,\End\} \cup \{\hat{T}_k \mid T \in \mathfrak{T}, k = 0,1,2\}
$$
and let $\Theta=(\Sigma_{1},\Sigma_{2})$. Define $\T_{2}'=\T_{1}\cup \T_{2}$.
It is sufficient to show that $\K_{1}$ and $\K_{2}'$ are $\Sigma_{2}$-rCQ-inseparable iff
$\T_{1}$ and $\T_{2}'$ are $\Theta$-rCQ-inseparable. The direction from right to
left is trivial as we can use the ABox $\A$ as a witness ABox for $\Theta$-rCQ-separability
between $\T_{1}$ and $\T_{2}'$ if $\K_{1}$ and $\K_{2}'$ are $\Sigma_{2}$-rCQ-separable.
Conversely, assume there is a $\Sigma_{1}$-ABox $\A'$ which $\Theta$-rCQ-separates $\T_{1}$
and $\T_{2}'$. Using the fact that $\End\not\in \Sigma_{1}$ and that
$$
A \sqsubseteq \exists R. (\Row \sqcap \exists R. I_0)
$$
and
$$
A\sqsubseteq \exists R.D ~\AND~ \exists R. \exists R. E ~\AND
  \bigsqcap_{X \in \Sigma_0} X
$$
are the only concept inclusions in $\T_{1}\cup \T_{2}$ that generate new $R$-successors from ABox individuals
one can now readily show that $\K_{1}$ and $\K_{2}'$ are $\Sigma_{2}$-rCQ--separable.
\end{proof}

%
%




%
%
%
%
\section{Proof of Theorem~\ref{thm:ucq}}
We aim to prove that it is 2\ExpTime-complete to decide whether an \ALC KB $\Kmc_1$
$\Sigma$-UCQ entails an \ALC KB $\Kmc_2$.

\subsection{Tree Automata Preliminaries}

We introduce two-way alternating parity automata on infinite trees (\TWAPAs).
Let \Nbbm\xspace denote the \emph{positive} integers. A \emph{tree} is
a non-empty (and potentially infinite) set $T \subseteq \Nbbm^*$
closed under prefixes.
The node $\varepsilon$ is the \emph{root} of $T$. As a convention, we
take $x \cdot 0 = x$ and $ (x \cdot i) \cdot -1 = x$. Note that
$\varepsilon \cdot -1$ is undefined.  We say that $T$ is
\emph{$m$-ary} if for every $x \in T$, the set $\{ i \mid x \cdot i
\in T \}$ is of cardinality exactly $m$.  W.l.o.g., we assume that all
nodes in an $m$-ary tree are from $\{1,\dots,m\}^*$.

We use $[m]$ to denote the set $\{-1,0,\dots,m\}$ and for any set $X$,
let $\Bmc^+(X)$ denote the set of all positive Boolean formulas over
$X$, i.e., formulas built using conjunction and disjunction over the
elements of $X$ used as propositional variables, and where the special
formulas $\mn{true}$ and $\mn{false}$ are allowed as well. For an
alphabet $\Gamma$, a \emph{$\Gamma$-labeled tree} is a pair $(T,L)$
with $T$ a tree and $L:T \rightarrow \Gamma$ a node labeling function.
\begin{definition}[\TWAPA]
  A \emph{two-way alternating automaton (\TWAPA) on infinite $m$-ary trees} is
  a tuple $\Amf=(Q,\Gamma,\delta,q_0, c)$ where $Q$ is a finite set of
  \emph{states}, $\Gamma$ is a finite alphabet,
  $\delta: Q \times \Gamma \rightarrow \Bmc^+(\mn{tran}(\Amf))$ is the
  \emph{transition function} with $\mn{tran}(\Amf) = [m] \times Q$ the set of
  \emph{transitions} of \Amf, $q_0 \in Q$ is the \emph{initial state}, and
  $c: Q \to \Nbbm$ is a function assigning natural numbers to the states.
\end{definition}
Intuitively, a transition $(i,q)$ with $i>0$ means that a copy of the
automaton in state $q$ is sent to the $i$-th successor of the
current node. Similarly, $(0,q)$
means that the automaton stays at the current node and switches to
state $q$, and $(-1,q)$ indicates moving to the
predecessor of the current node.
\begin{definition}[Run, Acceptance]
  A \emph{run} of a \TWAPA $\Amf = (Q,\Gamma,\delta,q_0,c)$ on an infinite
  $\Gamma$-labeled tree $(T,L)$ is a $T \times Q$-labeled tree
  $(T_r,r)$ such that the following conditions are satisfied:
  \begin{enumerate}

  \item $r(\varepsilon) = ( \varepsilon, q_0)$

  \item if $y \in T_r$, $r(y)=(x,q)$, and $\delta(q,L(x))=\vp$, then
    there is a (possibly empty) set $Q = \{ (c_1,q_1),\dots,(c_n,q_n)
    \} \subseteq \mn{tran}(\Amf)$ such that $Q$ satisfies $\vp$ and
    for $1 \leq i \leq n$, $x \cdot c_i$ is defined and  a node
    in $T$, and there is a $y \cdot i \in T_r$ such that $r(y \cdot
    i)=(x \cdot c_i,q_i)$.

  \end{enumerate}
  We say that $(T_r,r)$ is \emph{accepting} if in all infinite paths $\pi =
  y_1y_2\cdots$ of $T_r$, we have $\textsf{min}\{c(q) \mid r(y_i)=q\text{ for
    infinitely many }y_i\in\pi\}$ is even.
  An infinite $\Gamma$-labeled tree $(T,L)$ is \emph{accepted} by \Amf if there
  is an accepting run of \Amf on $(T,L)$. We use $L(\Amf)$ to denote the set of
  all infinite $\Gamma$-labeled tree accepted by \Amf.
\end{definition}
We will use the following results from automata theory:
\begin{theorem}\label{thm:autostuff1}~\\[-4mm]
  \begin{enumerate}

  \item Given a \TWAPA, we can construct in polynomial time a \TWAPA
    that accepts the complement language;

  \item Given a constant number of \TWAPAs, we can construct in polytime a
    \TWAPA that accepts the intersection language;

  \item Emptiness of \TWAPAs can be checked in single exponential time in the
    number of states.

  \item Given a \TWAPA \Amf, if $\L(\Amf)\neq\emptyset$, then $\L(\Amf)$
    contains a regular tree \cite{Rabin1972}.
  \end{enumerate}
\end{theorem}

\subsection{Regular Interpretations and Homomorphisms}

\begin{lemmanum}{\ref{lem:1}}
Let $\K_{1}$ and $\K_{2}$ be KBs, $\Sigma$ a signature, and let $\Imc_{1}$ be
a regular forest-shaped model of $\K_{1}$ of bounded outdegree. Assume no model
of $\K_{2}$ is $\Sigma$-homomorphically embeddable into $\Imc_{1}$.
Then there exists $n>0$ such that no model of $\K_{2}$ is $n\Sigma$-homomorphically
embeddable into $\Imc_{1}$.
\end{lemmanum}
\begin{proof}
Assume to the contrary that

  \medskip
  \noindent
  ($\ast$) for any $n>0$
  there exists a model $\J \in \Mod^{\it fo}_{\K_{2}}$
  that is $n\Sigma$-homomorphically embeddable into $\I_1$.

  \medskip
  \noindent
  Denote by $\J_{|\leq n}$ the subinterpretation of $\J$ whose elements are
  connected to ABox individuals by paths using role names of length $\le n$. A
  $(\Sigma,n)$-homomorphism $h$ from $\J$ is a $\Sigma$-homomorphism with
  domain $\J_{|\leq n}$.  Let $\Xi_{n}$ be the class of $(\J,h)$ with $\J \in
  \Mod^{\it fo}_{\K_{2}}$ and $h$ a $(\Sigma,n)$-homomorphism from $\J$ to
  $\I_{1}$. By ($\ast$) all $\Xi_{n}$ are non-empty. We may assume that for
  $(\I,h),(\J,f)\in \Xi:=\bigcup_{n\geq 0}\Xi_{n}$ we have $\I_{|\leq
    n}=\J_{|\leq n}$ if $\I_{|\leq n}$ and $\J_{|\leq n}$ are isomorphic. We
  are going to define classes $\Theta_{n} \subseteq \bigcup_{m\geq n}\Xi_{m}$
  such that the following conditions hold:
  \begin{itemize}
  \item[(a)] $\Theta_{n} \cap \Xi_{m} \not=\emptyset$ for all $m\geq n$;
  \item[(b)] $\I_{|\leq n} = \J_{| \leq n}$ and $h_{|\leq n} = f_{|\leq n}$ for
    all $(\I,h),(\J,f)\in \Theta_{n}$ ($h_{|\leq n}$ denotes the restriction of
    $h$ to $\I_{|\leq n}$).
  \end{itemize}
  Let $\Theta_{0}$ be the set of all pairs $(\J,h)$ such that $(\J,h)\in
  \Xi_{0}$ and $h$ is a $(\Sigma,n)$-homomorphism from $\J$ into $\I_{1}$ for
  some $n\geq 0$. Our assumptions directly imply that $\Theta_{0}$ has the
  properties (a) and (b) above since $h(a)=a$ holds for every homomorphism $h$
  and all ABox individuals $a$ in $\K_{2}$. Suppose now that $\Theta_{n}$ has
  been defined and satisfies (a) and (b).  Let
  \begin{eqnarray*}
    \Delta_{\J}^{0}  & = & \{ y \in \Delta^{\J_{|\leq n+1}}\setminus\Delta^{\J_{|\leq n}} \mid (x,y)\in R^{\J} \mbox{ for some $R\in \Sigma$ and $x\in \Delta^{\J_{|\leq n}}$ }\}\\	   
    \Delta_{\J}^{1}  & = & \Delta^{\J_{|\leq n+1}}\setminus (\Delta^{\J_{|\leq n}} \cup \Delta_{\J}^{0})
  \end{eqnarray*}
  Define an equivalence relation $\sim$ on $\Theta_{n} \cap (\bigcup_{m\geq
    n+1}\Xi_{m})$ by setting $(\I,h) \sim (\J,f)$ if
  \begin{itemize}
  \item $\I_{|\leq n+1} = \J_{| \leq n+1}$;
  \item $h(x) = f(x)$, for all $x\in \Delta_{\J}^{0}$;
  \item $h(x) = f(x)$, for all $x\in \Delta_{\J}^{1}$ such that $h(x)\in {\sf ind}(\K_{1})$ or $f(x)\in {\sf ind}(\K_{1})$;
  \item $h(x)$ and $f(x)$ are roots of isomorphic subinterpretations of $\I_{1}$, for all $x\in \Delta_{\J}^{1}$ such that $h(x)\not\in {\sf ind}(\K_{1})$ and
    $f(x)\not\in {\sf ind}(\K_{1})$.
  \end{itemize}
  By the bounded outdegree and regularity of $\I_{1}$, the properties (a) and
  (b) of $\Theta_{n}$, and the bounded outdegree of all $\J$ such that
  $(\J,h)\in \Xi_{n}$, the number of equivalence classes is finite. Hence there
  exists an equivalence class $\Theta$ satisfying (a).  Clearly we can modify
  $\Theta$ in such a way that also $h(x)=f(x)$ for all $x\in \Delta_{\J}^{1}$
  such that $h(x)\not\in {\sf ind}(\K_{1})$ and $f(x)\not\in {\sf ind}(\K_{1})$
  while preserving all the remaining properties of $\Theta$. The resulting set
  is as required for $\Theta_{n+1}$.
	
  We now define an interpretation $\J$ with a $\Sigma$-homomorphism $h$ as
  follows:
  \begin{eqnarray*}
    \J & = &\bigcup_{n<\omega} \{\J_{|\leq n} \mid \exists h\;(\J,h)\in \Theta_{n}\}\\
    h  &= &\bigcup_{n<\omega} \{h_{|\leq n} \mid \exists \J\;(\J,h)\in \Theta_{n}\}
  \end{eqnarray*}
  It is straightforward to show that $\J$ is a model of $\K_{2}$ and $h$ is a
  $\Sigma$-homomorphism from $\J$ into $\I_{1}$, as required.
\end{proof}

\subsection{$\Gamma$-labeled Trees}

Fix \ALC KBs $\K_1=(\T_1,\A_1)$ and $\K_2=(\T_2,\A_2)$, and a signature
$\Sigma$.  We aim to check if there is a model
$\Imc_1 \in \Mod^{\it fo}_{\K_{1}}$ into which no model $\Imc_2$ of $\Kmc_2$ is
$\Sigma$-homomorphically embeddable.
In the following, we construct a \TWAPA \Amf that accepts (suitable
representations of) the desired models $\Imc_1$, and for deciding their
existence, it then remains to check emptiness.

We start with encoding \fosh interpretations as labeled trees. For
$i \in \{1,2\}$, we use $\mathsf{CN}(\T_i)$ and $\mathsf{RN}(\T_i)$ to denote
the set of concept names and role names in~$\T_i$, respectively.  Node labels
are taken from the alphabet
$$
\begin{array}{rcl}
\Gamma &=& \{\mathit{root}, \mathit{empty}\} \cup
(\ind(\A_1) \times 2^{\mn{CN}(\T_1)}) \cup
(\mathsf{RN}(\T_1) \times 2^{\mn{CN}(\T_1)}).
\end{array}
$$
The $\Mod^{\it fo}_{\K_1}$ models will be represented as $m$-ary
$\Gamma$-labeled trees, with $m = \text{max}(|\T_1|, |\ind(\K_1)|)$. The root
node is not used in the representation and receives label $\mathit{root}$. Each
ABox individual is represented by a successor of the root labeled with a symbol
from $\ind(\A_1) \times 2^{\mn{CN}(\T_1)}$; anonymous elements are represented
by nodes deeper in the tree labeled with a symbol from $\mathsf{RN}(\T_1)
\times 2^{\mn{CN}(\T_1)}$. The node label $\mathit{empty}$ is used for padding
to achieve that every tree node has exactly $m$ successors. We call a
$\Gamma$-labeled tree \emph{proper} if it satisfies the following conditions:
\begin{itemize}

\item the root is labeled with $\mathit{root}$;

\item for every $a \in \ind(\A_1)$, there is exactly one successor of
  the root that is labeled with a symbol from $\{a\} \times
  2^{\mn{CN}(\T_1)}$; all remaining successors of the root are labeled
  with $\mathit{empty}$;

\item all other nodes are labeled with a symbol from $\mathsf{RN}(\T_1) \times
  2^{\mn{CN}(\T_1)}$ or with $\mathit{empty}$;

\item if a node is labeled with $\mathit{empty}$, then so are all its successors.

\end{itemize}
A proper
$\Gamma$-labeled tree $(T,L)$ represents the following interpretation~$\I_{(T,L)}$:
$$
\begin{array}{@{}r@{}c@{}l}
  \Delta^{\I_{(T,L)}} &=& \mn{ind}(\Amc_1) \cup \{ x \in T \mid |x|>1\}\\[1mm]
  \Int[\I_{(T,L)}]{A} &=& \{ a \mid \exists x \in T: L(x)=(a,\type) \text{ with
  } A \in \type \}  \cup
  \{ x \in T \mid L(x)=(R,\type) \text{ with } A \in \type\}
  \\[1mm]
  \Int[\I_{(T,L)}]{R} &=& \{(a,b) \mid R(a,b) \in \A_1\} \cup{}\\[1mm]
  &&\{(a,ij) \mid ij \in T, L(i)=(a,\type_1), 
  L(ij)=(R,\type_2) \} \cup{} \\[1mm]
  &&\{(x,xi) \mid xi \in T,  L(x)=(S,\type_1), 
  L(xi)=(R,\type_2) \}.
\end{array}
$$
Note that $\I_{(T,L)}$ satisfies all required conditions to qualify as
a \fosh model of $\Tmc_1$ (except that it need not satisfy $\Tmc_1$),
and that its outdegree is bounded by $|\Tmc_1|$. Conversely, every
\fosh model of $\Tmc_1$ with outdegree bounded by $|\Tmc_1|$ can be
represented as a proper $m$-ary $\Gamma$-labeled tree.

\subsection{The automata construction}

The desired \TWAPA \Amf is assembled from the following three
automata:
\begin{itemize}

\item a \TWATA $\Amf_0$ that accepts an $m$-ary $\Gamma$-labeled tree iff it is proper;

\item a \TWATA $\Amf_1$ that accepts a proper $m$-ary $\Gamma$-labeled tree
  $(T,L)$ iff $\Imc_{(T,L)}$ is a model of $\Tmc_1$;

\item a \TWAPA $\Amf_2$ that accepts a proper $m$-ary $\Gamma$-labeled tree
  $(T,L)$ iff there is a model $\Imc_2$ of $\Kmc_2$ that is
  $\Sigma$-homomorphically embeddable into $\Imc_{(T,L)}$.

\end{itemize}
Then, $\L(\Amf_0) \cap \L(\Amf_1) \cap \overline{\L(\Amf_2)} = \emptyset$ iff
for each model $\Imc_1 \in \Mod^{\it fo}_{\K_{1}}$ there exists a model
$\Imc_2$ of $\Kmc_2$ that is $\Sigma$-homomorphically embeddable into $\Imc_1$.
We thus define \Amf to be the intersection of $\Amf_0$, $\Amf_1$, and the
complement of $\Amf_2$.

\smallskip

The construction of $\Amf_0$ is trivial, details are omitted.

\smallskip

The construction of $\Amf_1$ is quite standard \cite{CaEO07}.  Let $C_{\T_1}$
be the negation normal form (NNF) of the concept $$\midsqcap_{C \sqsubseteq D
  \in \Tmc_1} (\neg C \OR D)$$ and let $\mathsf{cl}(C_{\T_1})$ denote the set of
subconcepts of $C_{\T_1}$, closed under single negation.
Now, the
\TWATA $\Amf_{1} = \tup{Q, \Gamma, \delta, q_0}$ is defined by
setting
$$
\begin{array}{r@{}c@{}l}
  Q &=& \{q_0, q_1, q_\emptyset\} \cup
  \{q^{a,C}, q^{C}, q^{R}, q^{\neg R} \mid a \in \mn{ind}(\Amc_1),
  C \in \mathsf{cl}(C_{\T_1}), R \in \mathsf{RN}(\T_1)\}
\end{array}
$$
and defining the transition function $\delta$ as follows:
%
$$
\begin{array}{r@{~~}c@{~~}l}
  \delta(q_0,\mathit{root}) &=& %
  \displaystyle \bigwedge_{i=1}^m (i, q_1)
  \\
  \delta(q_1,\ell) &=&
  \displaystyle
  ((0, q_\emptyset) \lor (0,q^{C_{\T_1}})) \land \bigwedge_{i=1}^m (i, q_1)
  \\
  \delta(q^{\SOME{R}{C}}, (a,U)) &=&
  \displaystyle
  \bigvee_{i=1}^m ((i, q^{R}) \wedge (i, q^{C})) \lor
  \bigvee_{R(a,b) \in \A_1}(-1,q^{b,C})
  \\[5mm]
  \delta(q^{\forall R.C}, (a,U)) &=&
  \displaystyle
  \bigwedge_{i=1}^m ((i,q_{\emptyset}) \lor (i, q^{\NOT R}) \lor (i, q^{C})) \land
  \bigwedge_{R(a,b) \in \A_1}(-1,q^{b,C})
  \\[5mm]
  \delta(q^{a,C},\mathit{root}) &=&
  \displaystyle
  \bigvee_{i=1}^m  (i, q^{a,C})
  \\
  \delta(q^{a,C}, (a,U)) &=& (0,q^C)
  \\[2mm]
  \delta(q^{\SOME{R}{C}}, (S,U)) &=&
  \displaystyle
  \bigvee_{i=1}^m ((i, q^{R}) \wedge (i, q^{C}))
  \\[4mm]
  \delta(q^{\forall R.C}, (S,U)) &=&
  \displaystyle
  \bigwedge_{i=1}^m ((i,q_{\emptyset}) \lor (i, q^{\NOT R}) \lor (i, q^{C}))
  \\[4mm]
  \delta(q^{C \AND C'}, (x,U)) &=&
  (0, q^{C}) \land (0, q^{C'}) \\[2mm]
  \delta(q^{C \OR C'}, (x,U)) &=&
  (0, q^{C}) \lor (0, q^{C'}) \\[2mm]
  \delta(q^{A}, (x,U)) &=& \mathsf{true}, \text{ if }A\in U
  \\
  \delta(q^{\NOT A}, (x,U)) &=& \mathsf{true}, \text{ if }A\notin U
  \\
  \delta(q^{R}, (R,U)) &=& \mathsf{true}
  \\
  \delta(q^{\NOT R}, (S,U)) &=& \mathsf{true},  \text{ if }R\neq S
  \\
  \delta(q_\emptyset,\mathit{empty}) &=& \mn{true} \\
  \delta(q,\ell) &=& \mathsf{false} %
  \quad \text{ for all other }q \in Q,~\ell\in \Gamma.
\end{array}
$$
Where $x$ in the labels $(x,U)$ stands for an individual $a$ or for a role name
$S$, and $\ell$ in the second transition is any label from $\Gamma$. It is
standard to show that $\Amf_1$ accepts the desired tree language.


\smallskip

For constructing $\Amf_2$, we first introduce some preliminaries.
We use $\mathsf{cl}(\T_2)$ to denote the set of
subconcepts of (concepts in) $\T_2$, closed under single negation.
For each interpretation $\I = (\Delta^\I, \cdot^\I)$ and $d \in
\Delta^{\I}$, the \emph{$\T_2$-type of $d$ in $\I$}, denoted
$\type_{\T_2}^{\I}(d)$, is defined as
$
\type_{\T_2}^{\I}(d)= \{ C\in \mathsf{cl}(\T_2) \mid d\in C^{\I}\}.
$
A subset $\type \subseteq \mathsf{cl}(\T_2)$ is a \emph{$\T_2$-type}
if $\type = \type_{\T_2}^{\I}(d)$, for some model $\I$ of $\T_2$ and
$d\in \Delta^{\I}$.  With $\types(\T_2)$, we denote the set of all
$\T_2$-types.  Let $\type,\type' \in \types(\T_2)$.
For $\SOME{R}{C} \in \type$, we say that $\type'$ is an
\emph{$\SOME{R}{C}$-witness for $\type$} if $C \in \type'$ and $\midsqcap \type
\AND \SOME{R}{(\midsqcap \type')}$ is satisfiable w.r.t.\ $\T_2$. Denote by
$\succ_{\SOME{R}{C}}(\type)$ the set of all $\SOME{R}{C}$-witnesses for $\type$.
A \emph{completion of} $\K_2$ is a function $\tau \colon \ind(\A_2) \to
\types(\T_2)$ such that, for any $a \in \ind(\A_2)$, the KB
$$
\big(\T_2 \cup \bigcup_{a \in \ind(\A_2), C \in \tau(a)}A_a \ISA C, ~\A \cup \bigcup_{a \in \ind(\A_2)} A_a(a)\big)
$$
is consistent,
where $A_a$ is a fresh concept name for each $a \in\ind(A_2)$.
Denote by $\compl(\K_2)$ the set of all completions of~$\K_2$; it can be
computed in exponential time in $|\K_2|$.

We now construct the \TWAPA $\Amf_2$. It is easy to see that if there
is an assertion $R(a,b) \in \Amc_2 \setminus \Amc_1$ with $R \in
\Sigma$, then no model of $\Kmc_2$ is $\Sigma$-homomorphically
embeddable into a \fosh model of $\Kmc_1$. In this case, we choose
$\Amf_2$ so that it accepts the empty language.

Now assume that there is no such assertion. It is also easy to see
that any model $\Imc_2$ of $\Kmc_2$ such that some $a \in \ind(\K_2)
\setminus \ind(\K_1)$ occurs in the $\Imc_2$-extension of some
$\Sigma$-symbol is not $\Sigma$-homomorphically embeddable into
a \fosh model of $\Kmc_1$.  For this reason, we should only consider
completions of $\Kmc_2$ such that for all $a \in \ind(\K_2) \setminus
\ind(\K_1)$, $\tau(a)$ contains no $\Sigma$-concept names and no
existential restrictions $\exists R . C$ with $R \in \Sigma$. We use
$\compl_{\mn{ok}}(\Kmc_2)$ to denote the set of all such completions.
%
%
Now the \TWAPA
$\Amf_{2}=\tup{Q, \Gamma, \delta, q_0, c}$ is defined by setting
$$
\begin{array}{r@{~}c@{~}l}
Q &=& \{q_0\} \cup
\{q^{a,\type}, q^{R,\type}, f^{\type} \mid a \in \ind(\Amc_1),
\type \in \types(\T_2), R \in \mathsf{RN}(\T_2) \cap \Sigma\}
\end{array}
$$
%
%
and defining the transition function $\delta$ as follows:
$$
\begin{array}{r@{\,}c@{\,}l@{\,}l}
  %
  \delta(q_0,\mathit{root}) &=&
  \displaystyle
  \bigvee_{\tau \in \compl_{\mn{ok}}(\K_2)} %
  \bigwedge_{a \in \ind(\A_2) \cap \ind(\A_1)} %
  \bigvee_{i=1}^{m} (i, q^{a,\tau(a)}) \\[6mm]
  \delta(q^{a,\type}, (a,U)) &=&
  \displaystyle
  \bigwedge_{\substack{\SOME{R}{C} \in \type\\R \in \Sigma}} %
  \bigvee_{\boldsymbol{s}\in \succ_{\SOME{R}{C}}(\type)}
  \left(\bigvee_{i=1}^m
    (i, q^{R,\boldsymbol{s}}) \lor%
    \bigvee_{R(a,b) \in \Amc_1} (-1, q^{b,\boldsymbol{s}})
  \right)  \wedge
  &\displaystyle\bigwedge_{\substack{\SOME{R}{C} \in \type\\R \notin \Sigma}} %
  \bigvee_{\boldsymbol{s}\in \succ_{\SOME{R}{C}}(\type)} f^{\boldsymbol{s}}
  \\[8mm]
  \delta(q^{S,\type}, (S,U)) &=& 
  \displaystyle
\bigwedge_{\substack{\SOME{R}{C} \in \type\\R \in \Sigma}} %
  \bigvee_{\boldsymbol{s}\in \succ_{\SOME{R}{C}}(\type)}%
  \bigvee_{i=1}^m (i, q^{R,\boldsymbol{s}}) %
 &\displaystyle\bigwedge_{\substack{\SOME{R}{C} \in \type\\R \notin \Sigma}} %
  \bigvee_{\boldsymbol{s}\in \succ_{\SOME{R}{C}}(\type)}%
  f^{\boldsymbol{s}} %
\end{array}
$$
where the latter two transitions are subject to the conditions that
every $\Sigma$-concept name in $\type$ is also in $U$;
also put
\begin{align*}
  \delta(f^{\type}, (x,U)) &= (0,q^{x,\type}) \lor \bigvee_{i=1}^m (i,f^{\type}) \lor (-1,f^{\type})\\%
  \delta(f^{\type}, \mathit{root}) &= \bigvee_{i=1}^m (i,f^{\type}) \\%
  \label{hom-anon}
  \delta(q^{a,\type}, \mathit{root}) &= \bigvee_{i=1}^m (i,q^{a,\type}) \\
  \delta(q,\ell) &= \mathsf{false}
  \quad \text{ for all other }q \in Q \text{ and } \ell\in \Gamma,
\end{align*}
where $x$ stands for an individual $a$ or for a role name $S$. We observe that
the states $f^{\type}$ are used for finding non-deterministically the
homomorphic image of $\Sigma$-disconnected successors in the tree. Finally, we
set
$c(q) = 2$ for $q \in \{q_0,q^{a,\type},q^{R,\type}\}$ and $c(f^{\type})=1$.
\begin{lemma}
  $(T,L) \in L(\Amf_2)$ iff there is a model $\Imc_2$ of $\K_2$ such
  that $\Imc_2$ is $\Sigma$-homomorphically embeddable into
  $\I_{(T,L)}$.
\end{lemma}
\begin{proof}
  ($\Rightarrow$) Given an accepting run $(T_r,r)$ for $(T,L)$, we can
  construct a \fosh model $\Imc_2$ of $\K_2$ and a
  $\Sigma$-homomorphism $h$ from $\Imc_2$ to $\I_{(T,L)}$.
  Intuitively, each node $y \in T_r$ with $r(y) = (x,q^{a,\type})$
  imposes that $a$ has type $\type$ in $\Imc_2$, and each node $y \in
  T_r$ with $r(y) = (x,q^{R,\type})$ imposes that $\Imc_2$ contains an
  element $y$ that belongs to a tree-shaped part of $\Imc_2$, is
  connected to its predecessor via $R$, and has type $\type$.
  The homomorphism $h$ is defined by choosing the identity on
  individual names,  and setting $h(y)=a$ when $r(y) =
    (x,q^{a,\type})$ and $h(y)=x$ when $r(y) = (x,q^{R,\type})$.

  ($\Leftarrow$) Assume that there is a model $\Imc_2$ of $\K_2$ such
  that $\Imc_2$ is $\Sigma$-homomorphically embeddable into
  $\I_{(T,L)}$. By the proof of Theorem~\ref{crit:KB}, we can assume
  $\Imc_2 \in \Mod^{\it fo}_{\K_{2}}$. It is now straightforward to construct an
  accepting run for $(T,L)$ by using $\Imc_2$ as a guide.
\end{proof}

It is easy to verify that the constructed automaton \Amf has only single
exponentially many states. Thus checking its emptiness can be done in
2\ExpTime.

\subsection{Strengthening of the model-theoretic characterization}

By Rabin \cite{Rabin1972}, whenever $\L(\Amf)$ is non-empty, we actually have a
regular forest-shaped model $\Imc_{1}$ of $\K_{1}$ of bounded outdegree such
that no model of $\K_{2}$ is $\Sigma$-homomorphically embeddable into
$\Imc_{1}$. By Lemma~\ref{lem:1} and Theorem~\ref{crit:KB}~(1) we then have
that $\L(\mathfrak{A})$ is non-empty iff $\K_{1}$ does not $\Sigma$-UCQ entail
$\K_{2}$.  We thus obtain the following strengthening of the model-theoretic
characterization:

\begin{theoremnum}{\ref{UCQ-full-hom}}
  $\mathcal{K}_{1}$ $\Sigma$-UCQ entails $\K_{2}$ iff for all models $\Imc_{1}$
  of $\K_{1}$ there exists a model $\I_{2}$ of $\K_{2}$ that is
  $\Sigma$-homomorphically embeddable into $\Imc_{1}$.
\end{theoremnum}
\begin{proof}
  Assume there exists a model $\Imc_{1}$ of $\K_{1}$ for which this is not the
  case. Then there exists a forest-shaped model of $\K_{1}$ of bounded
  outdegree for which this is not the case. Then $\L(\mathfrak{A})$ is
  non-empty. Then $\L(\mathfrak{A})$ contains a regular tree. Then there exists
  a regular forest-shaped model of $\K_{1}$ of bounded outdegree for which this
  is not the case. By Lemma~\ref{lem:1} and Theorem~\ref{crit:KB} (1) we obtain
  that $\mathcal{K}_{1}$ does not $\Sigma$-UCQ entail $\K_{2}$. The converse
  direction is clear.
\end{proof}

It thus remains to invoke Point~3 of Theorem~\ref{thm:autostuff1} to obtain the
upper bound in Theorem~\ref{thm:ucq}.

\subsection{2\ExpTime lower bound}

We reduce the word problem of exponentially space bounded alternating Turing
machines (ATMs), see \cite{ChandraKS81}. Let $M = (\Gamma,Q,q_0,q_{a}, q_{r},\delta)$ be an ATM with a tape alphabet
  $\Gamma$, a set of states $Q$ partitioned into existential $Q_\exists$ and
  universal $Q_\forall$ states, an initial state $q_0 \in Q_\exists$, an
  accepting state $q_{a}\in Q$, a rejecting state $q_{r}\in Q$ (all three are distinct), and a transition function:
  \begin{equation*}
    \delta\colon (Q\setminus\{q_{a}, q_{r}\})\times \Gamma \times \{0,1\} \to
    Q \times \Gamma \times\{-1,+1\},
  \end{equation*}
  which, for a state $q$ and symbol $a$, gives two instructions, $\delta(q,\sigma,0)$
  and $\delta(q,\sigma,1)$ (we will also denote them by $\delta_0(q,\sigma)$ and $\delta_1(q,\sigma)$, respectively). We assume that existential and universal states strictly
  alternate: any transition from an existential state leads to a universal
  state, and vice versa. Moreover, we assume that any run of $M$ on every input stops either in $q_{a}$ or $q_{r}$ \nb{V: new}.
  
Let $w$ be an input to $M$.  We aim to construct \ALC TBoxes $\Tmc_1$ and
$\Tmc_2$ and a signature $\Sigma$ such that the following are equivalent:
\begin{enumerate}

\item There is a model $\I_1$ of $\K_1 = (\Tmc_1,\{A(a)\})$ such that no model of
  $(\Tmc_2,\{A(a)\})$ is $\Sigma$-homomorphically embeddable into $\Imc_1$;

\item $M$ accepts $w$.

\end{enumerate}

The models of $\K_1$ encode all possible sequences of configurations of $M$
starting from the initial one.  Hence, most of the models do not correspond to
correct runs of $M$. The branches of the models stop at the accepting and
rejecting states.
On the other hand, the models of $\K_2$ encode all possible copying defects,
after the first step of the machine, or after the second step, and so on, or
detect valid (hence without copying defects) but rejecting runs.
Then, if there exists a finite model $\I_1$ of $\K_1$ such that no model of
$\K_2$ is $\Sigma$-homomorphically embeddable into $\I_1$, then $\I_1$
represents a valid accepting run of $M$.

The signature $\Sigma$ contains the following symbols:
\begin{enumerate}

\item concept names $A_0,\dots,A_{n-1}$ and
  $\overline{A}_0,\dots,\overline{A}_{n-1}$ that serve as bits in the binary
  representation of a number between 0 and $2^{n}-1$, identifying the position
  of tape cells inside configuration sequences ($A_0$, $\overline{A}_0$
  represent the lowest bit);

\item the concept names $A_\sigma$, for each $\sigma \in \Gamma$;

\item the concept names $A_{q,\sigma}$, for each $\sigma \in \Gamma$ and $q \in
  Q$;

\item concept names $X_0,X_1$ to distinguish the two successor configurations;

\item the role names $R$, $S$; $R$ is used to connect the successor configurations, whereas $S$ is used to connect a root of each configuration with symbols that occur in the cells of it. 
\end{enumerate}
Moreover, we use the following auxilary symbols not in $\Sigma$:
\begin{enumerate}
\item $B_i, B_\sigma, B_{q,\sigma}$; $G_i, G_\sigma, G_{q,\sigma}$; $C_\sigma, C_{q,\sigma}$, for $q, \sigma$ as above, $0 \leq i \leq n-1$,

\item $L_i^\ell$, $D_\textit{rej}^\ell$, $D_\textit{trans}^\ell$, $\textit{Counter}^\ell_m$, for $\ell = 0,1$, $m = -1, +1$,

\item $K$, $\textit{Stop}$, $K_0$, $Y$, $D$, $\bar{D}$,  $D_\textit{trans}$, $D_\textit{copy}$, $D_\textit{conf}$, $E$, $E_B$, $E_G$. 
\end{enumerate}
$\T_1$ contains the axioms:
\[
\begin{array}{r@{~}l}
  A \sqsubseteq & \exists R. (X_0 \sqcap K) \sqcap \exists R. (X_1 \sqcap K)\\
  (X_0 \sqcup X_1) \sqcap \neg \textit{Stop} \sqsubseteq& \exists R. (X_0 \sqcap K) \sqcap \exists R. (X_1 \sqcap K)\\[2mm]
  K \sqsubseteq& \exists S. (L_0^0 \sqcap \overline{A}_0) \sqcap \exists S. (L_0^1 \sqcap A_0)\\
  L_i^\ell \sqsubseteq& \exists S. (L_{i+1}^0 \sqcap \overline{A}_{i+1}) \sqcap \exists S. (L_{i+1}^1 \sqcap A_{i+1}) \text{ for }0\leq i\leq n-2, \ell=0,1\\
  L_{n-1}^k \sqsubseteq& \bigsqcup_{\sigma \in \Gamma} (A_\sigma \sqcup \bigsqcup_{q \in Q} A_{q,\sigma})\\[2mm]

  A_{\sigma_1} \sqcap A_{\sigma_2} \sqsubseteq& \bot \text{ for }\sigma_1\neq\sigma_2\\
  A_{\sigma_1} \sqcap A_{q_2,\sigma_2} \sqsubseteq& \bot\\
  A_{q_1,\sigma_1} \sqcap A_{q_2,\sigma_2} \sqsubseteq& \bot \text{ for }(q_1,\sigma_1)\neq(q_2,\sigma_2)\\[2mm]
  A_i \sqsubseteq& \forall S. A_i\\
  \overline{A}_i \sqsubseteq& \forall S. \overline{A}_i\\[2mm]
  \exists S^n. A_{q_a,\sigma} \sqsubseteq& \textit{Stop}\\
  \exists S^n. A_{q_r,\sigma} \sqsubseteq& \textit{Stop}\\
\end{array}
\]
where $\exists S^n.A$ is an abbreviation for the concept $\exists S. \exists S \dots \exists S.A$ ($S$ occurs $n$ times). The models of $\K_1$ look as follows:

\begin{center}
  \begin{tikzpicture}[xscale=1.4, yscale=1.2]
    \node[point, constant, label=left:{$a$}, label=below:{\small$K_0$}] (a) at
    (0,0) {};

    \foreach \al/\x/\y/\conc/\wh in {%
      x0/1/0.5/X_0/below,%
      x1/1/-0.5/X_1/above,%
      x00/2/0.9/X_0/above,%
      x01/2/0.3/X_1/above,%
      x10/2/-0.3/X_0/below,%
      x11/2/-0.9/X_1/below%
    }{ \node[point, label=\wh:{\scriptsize$\conc$}] (\al) at (\x,\y) {}; }

    \foreach \al/\down in {%
      a/1, x0/-1, x1/1%
    }{ \draw[gray] (\al) -- ++(0.4,\down*-0.7) -- ++(-0.8,0) -- (\al); }

    \foreach \from/\to/\wh in {%
      a/x0/above, a/x1/below, x0/x00/above, x0/x01/below, x1/x10/above,
      x1/x11/below%
    }{ \draw[role] (\from) -- node[\wh] {\scriptsize $R$} (\to); }

    \foreach \from in {%
      x00, x01, x10, x11%
    }{ \draw[dashed] (\from) -- ++(0.7,0); }
  \end{tikzpicture}
\end{center}
where the gray triangles are the trees encoding configurations rooted at $K$
except for the initial configuration. These trees are binary trees of depth
$n$, where each leaf represents a tape cell. The initial configuration is
encoded at $a$. For $w = \sigma_1\dots \sigma_m$, $\T_1$ contains the axioms:

\[
\begin{array}{r@{~}l}
  A \sqsubseteq &  \exists S. (L_0^0 \sqcap \overline{A}_0 \sqcap K_0) \sqcap \exists S. (L_0^1 \sqcap A_0 \sqcap K_0)\\
  K_0 \sqsubseteq& \forall S. K_0\\
  K_0 \sqcap (\textsf{val}_A=0) \sqsubseteq& A_{q_0,\sigma_1}\\
  K_0 \sqcap (\textsf{val}_A=i) \sqsubseteq& A_{\sigma_{i+1}} \text{ for }1 \leq i \leq m-1\\
  K_0 \sqcap (\textsf{val}_A\geq m) \sqsubseteq& A_{\square}\\
\end{array}
\]
where $(\textsf{val}_A=j)$ denotes the conjunction over $A_i,\overline{A}_i$
expressing the fact that the value of the $A$-counter is $j$, for $j \leq 2^n-1$.
Let 
\begin{align*}
\textsf{pos}^B&=(\overline{B}_0 \sqcup B_0) \sqcap \cdots \sqcap (\overline{B}_{n-1} \sqcup
B_{n-1}),\\
\textsf{state}_\forall^B& = \bigsqcup_{q\in Q_{\forall}, \sigma \in \Gamma} B_{q,\sigma},\\
\textsf{state}_\exists^B& = \bigsqcup_{q\in Q_{\exists}, \sigma \in \Gamma} B_{q,\sigma},\\
\textsf{symbol}^B&=\bigsqcup_{\sigma \in \Gamma} B_{\sigma}
\end{align*}
and analogously define $\textsf{pos}^G$, $\textsf{state}_\forall^G$, $\textsf{state}_\exists^G$, $\textsf{symbol}^G$.

$\T_2$ contains the following axioms:

\[
\begin{array}{r@{~}l}
  A \sqsubseteq &  \exists R. (X_0 \sqcap Y) \sqcup \exists R. (X_1 \sqcap Y) \sqcup D_\textit{rej}^0,\\
  Y \sqcap \overline{D} \sqsubseteq &  \exists R. (X_0 \sqcap Y) \sqcup \exists R. (X_1 \sqcap Y),\\
  Y \sqsubseteq& D \sqcup \overline{D}, \qquad D \sqcap \overline{D} \sqsubseteq \bot,\\
  D\sqsubseteq& D_{\textit{trans}} \sqcup D_{\textit{copy}} \sqcup D_{\textit{conf}},\\
\end{array}
\]
where $\ell = 0,1$, $\sigma \in \Gamma$.

$D_{\textit{trans}}$ encodes defects in executing transitions. It guesses the
(correct) position of the head, the symbol under it and the state by means of
the concepts $\textsf{pos}^B$ and $\textsf{state}_\forall^B$ or
$\textsf{state}_\exists^B$. This information is stored in the symbols
transparent to $\Sigma$ ($B_x$ and $\overline{B}_x$). Later we ensure that
symbols $B_x$ and $\overline{B}_x$ are propagated via the $S$-successors.
\[
\begin{array}{r@{~}l}
  D_{\textit{trans}} \sqsubseteq& \textsf{pos}^B \sqcap
                                  \exists S^n.E
                                  \sqcap \big((D_{\textit{trans}}^0 \sqcap D_{\textit{trans}}^1 \sqcap \textsf{state}_\exists^B)
                                  \sqcup  ((D_{\textit{trans}}^0 \sqcup D_{\textit{trans}}^1) \sqcap \textsf{state}_\forall^B)\big)\\

  D_{\textit{trans}}^\ell \sqsubseteq& \exists R. (X_\ell \sqcap \exists S^n.E).
  \end{array}
\]
We assume that here and everywhere below $\ell = 0,1$. For existential states both $X_0$ and $X_1$ successors must be ``defected'',
while for universal states at least of them.
The defected value at the successor configuration is stored in symbols $C_x^\ell$,
while the relative position of the defect is stored in $\textit{Counter}_m^\ell$
for $m \in \{-1,0,+1\}$. For $\delta_\ell(q,\sigma)=(q',\sigma',m)$,
$m \in \{-1,+1\}$,
\[
\begin{array}{r@{~}l}
 B_{q,\sigma} \sqcap D_{\textit{trans}}^\ell \sqsubseteq & (\textit{Counter}_{0}^\ell \sqcap \bigsqcup_{\sigma'' \in \Gamma \setminus \{ \sigma'\}} C_{\sigma''}^\ell) \sqcup
                                        (\textit{Counter}_m^\ell \sqcap \bigsqcup_{\sigma'' \in\Gamma} (C_{\sigma''}^\ell \sqcup
                                        \bigsqcup_{p \in Q \setminus \{q'\}}C_{p,\sigma''}^\ell)).
\end{array}
\]
The position of the defect is passed along the $R$-successor as follows:
\[
\begin{array}{r@{~}l}
  \textit{Counter}_{+1}^\ell \sqcap \overline{B}_k \sqcap B_{k-1} \sqcap \cdots \sqcap B_0
  \sqsubseteq& \forall R. (\neg X_\ell \sqcup (B_k \sqcap \overline{B}_{k-1} \sqcap \cdots \sqcap \overline{B}_0)) \text{ for } n > k \geq 0 \\
  \textit{Counter}_{+1}^\ell \sqcap \overline{B}_j \sqcap \overline{B}_{k}
  \sqsubseteq& \forall R. (\neg X_\ell \sqcup \overline{B}_{j}) \text{ for } n > j > k\\
  \textit{Counter}_{+1}^\ell \sqcap B_j \sqcap \overline{B}_{k}
  \sqsubseteq& \forall R. (\neg X_\ell \sqcup B_{j}) \text{ for } n > j > k\\[2mm]

  \textit{Counter}_{-1}^\ell \sqcap B_k \sqcap \overline{B}_{k-1} \sqcap \cdots \sqcap \overline{B}_0
  \sqsubseteq& \forall R. (\neg X_\ell \sqcup (\overline{B}_k \sqcap B_{k-1} \sqcap \cdots \sqcap B_0)) \text{ for } n > k \geq 0 \\
  \textit{Counter}_{-1}^\ell \sqcap \overline{B}_j \sqcap B_{k}
  \sqsubseteq& \forall R. (\neg X_\ell \sqcup\overline{B}_j) \text{ for } n > j > k\\
  \textit{Counter}_{-1}^\ell \sqcap B_j \sqcap B_{k}
  \sqsubseteq& \forall R. (\neg X_\ell \sqcup B_j) \text{ for } n > j > k\\[2mm]

 \textit{Counter}_{0}^\ell \sqcap  B \sqsubseteq& \forall R. (\neg X_\ell \sqcup B) \text{ for }B \in \{B_i, \overline{B}_i \mid 0 \leq i \leq n-1\}
\end{array}
\]
The defect is copied via $R$ as follows:
\[
\begin{array}{r@{~}l}
 C_x^\ell \sqsubseteq& \forall R. (\neg X_\ell \sqcup B_x), \quad x \in \{ (q,\sigma), \sigma \mid q \in Q, \sigma \in \Gamma \}
\end{array}
\]
All symbols $B_x$ and $\overline{B}_x$ are propagated down the $S$-successors,
and at the concept $E$ they are copied into the symbols $A_x$ and $\overline{A}_x$:
\[
\begin{array}{r@{~}l}
  B_x \sqsubseteq \forall S. B_x, \ & \bar{B}_i \sqsubseteq \forall S. \bar{B}_i,\\
  E \sqcap B_x \sqsubseteq A_x, \ & x \in \{0 \dots n-1\} \cup \{ (q,\sigma), \sigma \mid q \in Q, \sigma \in \Gamma \}\\
  E \sqcap \overline{B}_i \sqsubseteq \overline{A}_i, \ & 0 \leq i \leq n-1.\\[2mm]
\end{array}
\]
A model of a transition defect can be depicted as follows, for $n=3$ and
$\delta_1(q_1,\sigma_1) = (q_2, \sigma_2, R)$:

\begin{center}
  \begin{tikzpicture}[xscale=1.4, yscale=0.8]

    \foreach \al/\x/\y/\conc/\wh in {%
      x0/0/1/D_{\textit{trans}}/above,%
      x1/0/0//below,%
      x2/0/-1//above,%
      x3/0/-2/{A_2,\overline{A}_1, A_0, A_{q_1,\sigma_1}}/left,%
      y0/1/1/X_1/above,%
      y1/1/0//above,%
      y2/1/-1//below,%
      y3/1/-2/{A_2,A_1, \overline{A}_0, A_{q_3,\sigma_3}}/right%
    }{ \node[point, label=\wh:{\scriptsize$\conc$}] (\al) at (\x,\y) {}; }

    \foreach \from/\to/\wh in {%
      x0/x1, x1/x2, x2/x3, y0/y1, y1/y2, y2/y3%
    }{ \draw[role] (\from) -- node[left] {\scriptsize $S$} (\to); }

    \draw[role] (x0) -- node[above] {\scriptsize $R$} (y0);

   \end{tikzpicture}
\end{center}

$D_{\textit{copy}}$ encodes defects in copying the symbols that are not under
the head. It guesses the symbol and its position, and also the position and the
state of the head. The latter is stored using $G$-symbols and needed to know
whether the state is existential or universal.
\[
\begin{array}{r@{~}l}
  D_{\textit{copy}} \sqsubseteq& \textsf{pos}^B \sqcap \textsf{symbol}^B \sqcap \exists S^n.E_B \sqcap \\
  & \textsf{pos}^G \sqcap
          \big((D_{\textit{copy}}^0 \sqcap
          D_{\textit{copy}}^1 \sqcap\textsf{state}_\exists^G) \sqcup
          ((D_{\textit{copy}}^0 \sqcup
          D_{\textit{copy}}^1)\sqcap\textsf{state}_\forall^G)\big)
          \sqcap \exists S^n.E_G \sqcap \\
  & (\textsf{val}_B \neq \textsf{val}_G)\\

  D_{\textit{copy}}^\ell \sqsubseteq& \exists R. (X_\ell \sqcap \exists S^n.E)\\
  B_{\sigma} \sqcap D_{\textit{copy}}^\ell \sqsubseteq &
                                                      \textit{Counter}_{0}^\ell
                                                      \sqcap \bigsqcup_{\sigma'\in\Gamma, \sigma'\neq\sigma} (C_{\sigma'}^\ell \sqcup
                                                      \bigsqcup_{q \in
                                                      Q}C_{q,\sigma'}^\ell)\\
\end{array}
\]
where $(\textsf{val}_B \neq \textsf{val}_G)$ stands for $(B_0 \sqcap \overline{G}_0)
\sqcup (G_0 \sqcap \overline{B}_0) \sqcup \cdots \sqcup (B_{n-1} \sqcap
\overline{G}_{n-1}) \sqcup (G_{n-1} \sqcap \overline{B}_{n-1})$.
Similarly to $B$-symbols, $G_x$ and $\overline{G}_x$ symbols are copied via the
$S$-successors. At $E_B$ we only copy $B$-symbols to $A$-symbols, while at
$E_G$ we only copy $G$-symbols to $A$-symbols.

\[
\begin{array}{r@{~}l}
  G_i \sqsubseteq \forall S. G_i, & \ \overline{G}_i \sqsubseteq \forall S. \overline{G}_i, \ G_{q,\sigma} \sqsubseteq \forall S. G_{q, \sigma},\\
  E_B \sqcap B_i \sqsubseteq A_i, & \ E_B \sqcap B_{q, \sigma} \sqsubseteq A_{q, \sigma},\\
  E_B \sqcap \overline{B}_i \sqsubseteq \overline{A}_i, &\\
  E_G \sqcap G_i \sqsubseteq A_i, & \ E_G \sqcap G_{q, \sigma} \sqsubseteq A_{q, \sigma},\\
  E_G \sqcap \overline{G}_i \sqsubseteq \overline{A}_i, & \text{ for } 0 \leq i \leq n-1,\ q \in Q,\ \sigma \in \Gamma.
\end{array}
\]
A model of a copying defect can be depicted as follows, for $n=3$ and
$q_1 \in Q_\exists$:

\begin{center}
  \begin{tikzpicture}[xscale=1.4, yscale=0.8]

    \foreach \al/\x/\y/\conc/\wh in {%
      x0/0/1/D_{\textit{copy}}/above,%
      x1/-0.5/0//below,%
      x2/-0.5/-1//above,%
      x3/-0.5/-2/{\overline{A}_2,A_1, A_0, A_{q_1,\sigma_1}}/left,%
      xx1/0/-0.5//below,%
      xx2/0/-1.5//above,%
      xx3/0/-2.5/{A_2,A_1, \overline{A}_0, A_{\sigma}}/left,%
      y0/1.5/0.5/X_1/above,%
      y1/1.5/-0.5//above,%
      y2/1.5/-1.5//below,%
      y3/1.5/-2.5/{A_2,A_1, \overline{A}_0, A_{q_2,\sigma}}/right,%
      yy0/2/1.5/X_0/above,%
      yy1/2/0//above,%
      yy2/2/-1//below,%
      yy3/2/-2/{A_2,A_1, \overline{A}_0, A_{\sigma_2}}/right%
    }{ \node[point, label=\wh:{\scriptsize$\conc$}] (\al) at (\x,\y) {}; }

    \foreach \from/\to/\wh in {%
      x0/x1, x1/x2, x2/x3, x0/xx1, xx1/xx2, xx2/xx3, %
      y0/y1, y1/y2, y2/y3, yy0/yy1, yy1/yy2, yy2/yy3%
    }{ \draw[role] (\from) -- node[left] {\scriptsize $S$} (\to); }

    \draw[role] (x0) -- node[above] {\scriptsize $R$} (y0);
    \draw[role] (x0) -- node[above] {\scriptsize $R$} (yy0);

   \end{tikzpicture}
\end{center}

$D_{\textit{conf}}$ is a ``local'' defect that encodes incorrect
configurations, that is, configurations with at least two heads on the tape.
\[
\begin{array}{r@{~}l}
  D_{\textit{conf}} \sqsubseteq& \textsf{pos}^B \sqcap
                                 (\textsf{state}_\exists^B\sqcup\textsf{state}_\forall^B)
                                 \sqcap \exists S^n.E_B \sqcap \\
          & \textsf{pos}^G \sqcap
          (\textsf{state}_\exists^G\sqcup\textsf{state}_\forall^G)
          \sqcap \exists S^n.E_G \sqcap (\textsf{val}_B \neq \textsf{val}_G)\\
\end{array}
\]

Finally, we use $D_\textit{rej}^0$, $D_\textit{rej}^1$ and $D_\textit{rej}$ to
detect the fact that $M$ rejects $w$. \nb{V: picture here?}

\[
\begin{array}{r@{~}l}
  D_{\textit{rej}}^0 \sqsubseteq &  
  \displaystyle\bigsqcap_{\ell\in\{0,1\}} \exists R. (X_\ell \sqcap (D_\textit{rej}^1 \sqcup D_{\textit{rej}})),\\
  D_{\textit{rej}}^1 \sqsubseteq &  \exists R.(D_\textit{rej}^0 \sqcup D_\textit{rej}),\\
  D_{\textit{rej}} \sqsubseteq & \bigsqcup_{\sigma\in\Gamma}\exists S^n.A_{q_r, \sigma}
\end{array}
\]
A model of a rejecting ``defect'' can be depicted as follows:

\begin{center}
  \begin{tikzpicture}[xscale=1.6, yscale=1.2]
    \node[point, constant, label=left:{\scriptsize$D_{\textit{rej}}^0$}] (a) at
    (0,0) {};

    \foreach \al/\x/\y/\conc/\wh in {%
      x0/1/0.5/{X_0,D_{\textit{rej}}^1}/above,%
      x1/1/-0.5/{X_1,D_{\textit{rej}}^1}/below,%
      x20/2/0.5/{D_{\textit{rej}}^0}/above,%
      x21/2/-0.5/{D_{\textit{rej}}}/right,%
      x30/4/0.9/{X_0,D_{\textit{rej}}}/right,%
      x31/3/0.3/{X_1,D_{\textit{rej}}}/right%
    }{ \node[point, label=\wh:{\scriptsize$\conc$}] (\al) at (\x,\y) {}; }

    \foreach \from/\to/\wh in {%
      a/x0/above, a/x1/below, x0/x20/above, x1/x21/below, x20/x30/above,
      x20/x31/below%
    }{ \draw[role] (\from) -- node[\wh] {\scriptsize $R$} (\to); }

    \foreach \at/\i in {%
      x21/2,x30/3,x31/1%
    }{ %
      \node[point,yshift=-0.5cm] (s1) at (\at) {}; 
      \node[point,yshift=-1cm] (s2) at (\at) {};
      \node[point,yshift=-1.5cm, label=right:{\scriptsize$A_{q_r,\sigma_\i}$}] (s3) at (\at) {};
      \draw[role] (\at) -- node[left] {\scriptsize $S$} (s1);
      \draw[role] (s1) -- node[left] {\scriptsize $S$} (s2);
      \draw[role] (s2) -- node[left] {\scriptsize $S$} (s3);
    }

  \end{tikzpicture}
\end{center}

Note that some models of $\K_2$ are infinite paths or trees that do not ``realise'' any
defect. Such models of $\K_2$ will not be $\Sigma$-homomorphically embeddable
into the models of $\K_1$ representing valid accepting runs. 

It follows from what was said above and from Theorem~\ref{rUCQ-full-hom} that
$M$ accepts $w$ iff $\K_1$ does not rUCQ-entail $\K_2$. Thus, we obtain the
lower bound in Theorem~\ref{thm:rucq}. \nb{V: OK?} Moreover, it can be readily verified by comparing (1) and (2) of Theorem~\ref{crit:KB}, that $\K_1$ rUCQ-entails $\K_2$ iff $\K_1$ UCQ-entails $\K_2$. Thus, we also obtain the lower bound in Theorem~\ref{thm:ucq}.

\section{Proof of Thereoms~\ref{thm:first} and \ref{thm:second}}

In this section we prove the semantic characterizations given in Theorem~\ref{thm:first} and
Theorem~\ref{thm:second}.
In what follows we assume that \hALC TBoxes are given in \emph{normal form} with concept inclusions
of the following form:
$$
A\sqsubseteq B, \quad A_{1} \sqcap A_{2} \sqsubseteq B, \quad \exists R.A \sqsubseteq B
$$
and
$$
A \sqsubseteq \bot, \quad \top \sqsubseteq B, \quad A \sqsubseteq \exists R.B, \quad A \sqsubseteq \forall R.B
$$
where $A,B$ range over concept names.
We define the canonical model $\I_{\T,\A}$ of a consistent \hALC KB $\K=(\T,\A)$ with $\T$ in normal form in the standard way using
a chase procedure. Consider the following rules that are applied to ABox $\Amc$:
\begin{enumerate}
\item if $A(a)\in \A$ and $A\sqsubseteq B\in\T$, then add $B(a)$ to $\A$;
\item if $A_{1}(a)\in A$ and $A_{2}(a)\in \A$ and $A_{1} \sqcap A_{2} \sqsubseteq B\in \T$, then add $B(a)$ to $\A$;
\item if $R(a,b)\in A$ and $A(b)\in \A$ and $\exists R.A \sqsubseteq B \in \T$, then add $B(a)$ to $\A$;
\item if $a\in {\sf ind}(\A)$ and $\top \sqsubseteq B\in \T$, then add $B(a)$ to $\A$;
\item if $A(a)\in \A$ and $A \sqsubseteq \exists R.B\in \T$ and there are no $R(a,b),B(b)\in \A$, then
add assertions $R(a,b),B(b)$ to $\A$ for a fresh $b$;
\item if $A(a)\in \A$ and $R(a,b)\in \A$ and $A \sqsubseteq \forall R.B\in \T$, then add $B(b)$ to $\A$.
\end{enumerate}
Denote by $\A^{c}$ the (possibly infinite) ABox resulting from $\A$ by applying these rules
exhaustively to $\A$. Then the canonical model $\I_{\T,\A}$ is the interpretation defined by $\A^c$.

\medskip
\noindent
We now come to the proof of Theorem~\ref{thm:first}.

\begin{theoremnum}{\ref{thm:first}}
Let $\T_{1}$ be an \ALC TBox, $\T_{2}$ a \hALC TBox, and
$\Theta=(\Sigma_{1},\Sigma_{2})$.  Then $\T_{1}$ $\Theta$-rCQ-entails $\T_{2}$
iff for all tree-shaped $\Sigma_{1}$-ABoxes $\A$ of outdegree bounded by
$|\T_{2}|$ and consistent with $\T_{1}$ and $\T_{2}$, $\I_{\T_{2},\A}$ is
con-$\Sigma_{2}$-homomorphically embeddable into any model $\I_{1}$ of
$(\T_{1},\A)$.
\end{theoremnum}
\begin{proof}
It is known that \hALC is unravelling tolerant \cite{lutz2012non}, that is, if $(\T,\A)\models C(a)$ for a \hALC TBox $\T$ and $\mathcal{EL}$-concept $C$, then
$(\T,\A')\models C(a)$ for a finite subABox $\A'$ of the tree-unravelling $\A^{u}$ of $\A$ at $a$. Thus, any witness ABox for non-entailment w.r.t.~$\mathcal{EL}$-instance queries
can be transformed into a tree-shaped witness ABox. By Theorem~\ref{rUCQ-full-hom} it is therefore sufficient to prove
that if $\T_{1}$ does not $\Theta$-rCQ-entail $\T_{2}$, then
this is witnessed by an $\mathcal{EL}$-instance query $C(a)$.

\medskip
\noindent
{\bf Claim}. If $\T_{1}$ does not $\Theta$-rCQ-entail $\T_{2}$, then there exists a $\Sigma_{1}$-ABox $\A$
and an $\mathcal{EL}$-concept $C$ over $\Sigma_{2}$ such that $\T_{2},\A \models C(a)$ and $\T_{1},\A\not\models C(a)$ for some $a\in {\sf ind}(\A)$.

\medskip
\noindent
Assume $\Amc$ is a $\Sigma_{1}$-ABox and $q(\vec{x})$ a $\Sigma_{2}$-rCQ such that
$\Tmc_{2},\Amc\models q(\vec{a})$ but $\Tmc_{1},\Amc\not\models q(\vec{a})$.

First we show that there exists a $\Sigma_{2}$-CQ $q'(\vec{z})$ such that
$\Tmc_{2},\Amc\models q'(\vec{b})$ but $\Tmc_{1},\Amc\not\models q'(\vec{b})$ for some $\vec{b}$ and, moreover,
there exists a match $\pi$ for $q'$ in $\Imc_{\Tmc_{2},\Amc}$ witnessing this such that no quantified
variable in $q'$ is mapped to ${\sf ind}(\Amc)$. Let $\pi$ be a match for $q(\vec{x})$
in $\Imc_{\Tmc_{2},\Amc}$. Assume $q(\vec{x}) = \exists \vec{y}\varphi(\vec{x},\vec{y})$.
Let $\vec{y}_{1}$ be the additional variables mapped by $\pi$ to elements of ${\sf ind}(\Amc)$
and let $q'(\vec{x},\vec{y_{1}})= \exists \vec{y}_{2}\varphi(\vec{x},\vec{y})$, where $\vec{y}_{2}$
are the remaining variables in $\vec{y}$ without $\vec{y}_{1}$. Then $\Tmc_{2},\Amc\models q'(\pi(\vec{x},\vec{y}_{1}))$
but $\Tmc_{1},\Amc\not\models q'(\pi(\vec{x},\vec{y}_{1}))$. Clearly $q'$ is as required.

We can decompose $q'$ into
\begin{itemize}
\item a quantifier-free core $q_{0}$ containing all $A(x)$ and $r(x,y)$ in $q'$ such that $x,y$ are answer variables;
\item queries $q_{1}(x_{1}),\ldots,q_{n}(x_{n})$ that each have exactly one answer variable.
\end{itemize}
We distinguish the following cases:
\begin{itemize}
\item $\Tmc_{1},\Amc\not\models q_{0}(\pi(\vec{x},\vec{y}_{1}))$: in this case we find a single concept
name $A\in \Sigma_{2}$ and $a\in {\sf ind}(\Amc)$ such that $\Tmc_{2},\Amc \models A(a)$ and $\Tmc_{1},\Amc \not\models A(a)$.
\item there exists $1\leq i \leq n$ such that $\Tmc_{1},\Amc \not\models q_{i}(\pi(x_{i}))$.
Let $C_{i}$ be the image of $q_{i}$ under $\pi$ ($\pi$ maps all variables from $q_{i}$ except $x_{i}$
to elements of $\Imc_{\Tmc_{2},\Amc}$ not in $\Amc$. Thus this image is tree-shaped and can be identified with
an $\mathcal{EL}$-concept). Then $\Tmc_{2},\Amc\models C_{i}(\pi(x_{i}))$ and $\Tmc_{1},\Amc \not\models C_{i}(\pi(x_{i}))$.
\end{itemize}
We have thus shown that there exists a $\Sigma_{1}$-ABox $\Amc$ and an $\mathcal{EL}$-concept $C$ and $a\in {\sf ind}(\Amc)$
such that $\Tmc_{2},\Amc \models C(a)$ and $\Tmc_{1},\Amc \not\models C(a)$.
\end{proof}

\begin{theoremnum}{\ref{thm:second}}
Let $\Tmc_{1}$ and $\Tmc_{2}$ be \hALC TBoxes and $\Sigma_{1},\Sigma_{2}$ be
signatures.  Then $\Tmc_{1}$ $(\Sigma_{1},\Sigma_{2})$-CQ (equivalently,
$(\Sigma_{1},\Sigma_{2})$-UCQ) entails $\Tmc_{2}$ iff for all tree-shaped
$\Sigma_{1}$-ABoxes $\Amc$ of outdegree bounded by $|\Tmc_{2}|$ that are
consistent with $\Tmc_{1}$ and $\Tmc_{2}$, $\Imc_{\Tmc_{2},\Amc}$ is
$\Sigma_{2}$-homomorphically embeddable into $\Imc_{\Tmc_{1},\Amc}$.
\end{theoremnum}
\begin{proof}
The proof is similar to the proof of Theorem~\ref{thm:first}. Denote by $\mathcal{EL}^{u}$ the extension of $\mathcal{EL}$
with the universal role $u$.
Unravelling tolerance of \hALC implies also that if $(\T,\A)\models C(a)$ for a \hALC TBox $\T$ and $\mathcal{EL}^{u}$-concept $C$, then
$(\T,\A')\models C(a)$ for a finite subABox $\A'$ of the tree-unravelling $\A^{u}$ of $\A$ at $a$. Thus, any witness ABox for non-entailment
w.r.t.~$\mathcal{EL}^{u}$-instance queries can be transformed into a tree-shaped witness ABox. By the homomorphism criterion for $\Sigma$-CQ
entailment between \hALC-KBs proved in \cite{BotoevaKRWZ14} it is therefore sufficient to prove that if $\T_{1}$ does not $\Theta$-CQ-entail $\T_{2}$, then
this is witnessed by an $\mathcal{EL}^{u}$-instance query $C(a)$. This proof is a straightforward extension of the proof of Theorem~\ref{thm:first}.
\end{proof}

The notion of con-$\Sigmaquery$-homomorphic embeddability is slightly
unwieldy to use in the subsequent definitions and constructions. We
therefore resort to simulations. 
%
The following lemma gives an analysis of
non-con-$\Sigmaquery$-homomorphic embeddability in terms of
simulations that is relevant for the subsequent constructions.
%
%

\begin{lemmanum}{\ref{lem:homtosim}}
  %
  %
  %
Let \Amc be a $\Sigmaabox$-ABox and $\Imc_1$ a model of $(\Tmc_1,\Amc)$. Then
$\Imc_{\Tmc_2,\Amc}$ is not con-$\Sigmaquery$-homomorphically embeddable into
$\Imc_1$ iff there is $a \in \mn{ind}(\Amc)$ such that one of the following
holds:
  \begin{enumerate}\itemsep 0cm

  \item[\rm (1)] there is a $\Sigmaquery$-concept name $A$ with
    $a \in A^{\Imc_{\Tmc_2,\Amc}} \setminus A^{\Imc_1}$;

  \item[\rm (2)] there is an $R$-successor $d$ of $a$ in
    $\Imc_{\Tmc_2,\Amc}$, for some $\Sigmaquery$-role name $R$, such
    that $d \notin \mn{ind}(\Amc)$ and, for all $R$-successors $e$
    of $a$ in $\Imc_1$, we have  $(\Imc_{\Tmc_2,\Amc},d)
    \not\leq_{\Sigmaquery}(\Imc_1,e)$.

  \end{enumerate}
\end{lemmanum}
\begin{proof}(sketch) The ``if'' direction is clear by definition of
  homomorphisms and because of the following: if there is a
  $\Sigmaquery$-homomorphism $h$ from the maximal $\Sigma$-connected
  subinterpretation of $\Imc_{\Tmc_2,\Amc}$ to $\Imc_1$, $a \in
  \mn{ind}(\Amc)$, and $d$ is an $R$-successor of $a$ in $\Imc_{\Tmc_2,\Amc}$ with $R
  \in \Sigmaquery$ and $d \notin \mn{ind}(\Amc)$, then $h(d)$ is an
  $R$-successor of $a$ in $\Imc_1$ and $h$ contains a
  $\Sigmaquery$-simulation from $(\Imc_{\Tmc_2,\Amc},d)$ to $(\Imc_1,h(d))$.

  For the ``only if'' direction, assume that both Point~1 and Point~2 are false
  for all $a \in \mn{ind}(\Amc)$. Then for every $a \in \mn{ind}(\Amc)$,
  $R$-successor $d$ of $a$ in $\Imc_{\Tmc_2,\Amc}$ with $R \in \Sigmaquery$ and
  $d \notin \mn{ind}(\Amc)$, there is an $R$-successor $d'$ of $a$ in $\Imc_1$
  and a simulation $\S_d$ from $\Imc_1$ to $\Imc_{\Tmc_2,\Amc}$ such that
  $(d,d') \in \S_d$. Because the subinterpretation of $\Imc_{\Tmc_2,\Amc}$ rooted at $d$ is
  tree-shaped, we can assume that $\S_d$ is a partial function. Now consider
  the function $h$ defined by setting $h(a)=a$ for all $a \in \mn{ind}(\Amc)$
  and then taking the union with all the simulations $\S_d$. It can be verified
  that $h$ is a $\Sigmaquery$-homomorphism from the maximal $\Sigma$-connected
  subinterpretation of $\Imc_{\Tmc_2,\Amc}$ to $\Imc_1$.
\end{proof}

\section{Proof of Theorem~\ref{thm:exptbox}}

We aim to prove Theorem~\ref{thm:exptbox}, i.e., that it is \ExpTime-complete
to decide whether an \ALC TBox $\Tmc_1$ $(\Sigma_1,\Sigma_2)$-\RCQ entails a
Horn-\ALC TBox $\Tmc_2$. We are going to use automata on finite trees.

\subsection{Tree Automata Preliminaries}

We introduce two-way alternating B\"uchi automata on finite trees
(\TWABAs).  A finite tree $T$ is \emph{$m$-ary} if for every $x \in T$,
the set $\{ i \mid x \cdot i \in T \}$ is of cardinality zero or
exactly $m$.  An \emph{infinite path} $P$ of $T$ is a prefix-closed
set $P \subseteq T$ such that for every $i \geq 0$, there is a unique
$x \in P$ with $|x|=i$.
\begin{definition}[\TWABA]
  A \emph{two-way alternating B\"uchi automaton (\TWABA) on finite
    $m$-ary trees} is a tuple $\Amf=(Q,\Gamma,\delta,q_0,R)$ where $Q$
  is a finite set of \emph{states}, $\Gamma$ is a finite alphabet,
  $\delta: Q \times \Gamma \rightarrow \Bmc^+(\mn{tran}(\Amf))$ is the
  \emph{transition function} with $\mn{tran}(\Amf) = ([m] \times Q)
  \cup \mn{leaf}$ the set of \emph{transitions} of \Amf, $q_0 \in Q$
  is the \emph{initial state}, and $R \subseteq Q$ is a set of
  \emph{recurring states}.
\end{definition}
Transitions have the same intuition as for \TWATAs on infinite trees.
The additional transition $\mn{leaf}$ verifies that the automaton is
currently at a leaf node.
\begin{definition}[Run, Acceptance]
  A \emph{run} of a \TWABA $\Amf = (Q,\Gamma,\delta,q_0,R)$ on a finite
  $\Gamma$-labeled tree $(T,L)$ is a $T \times Q$-labeled tree
  $(T_r,r)$ such that the following conditions are satisfied:
  \begin{enumerate}

  \item $r(\varepsilon) = ( \varepsilon, q_0)$

  \item if $y \in T_r$, $r(y)=(x,q)$, and $\delta(q,L(x))=\vp$, then
    there is a (possibly empty) set $Q = \{ (c_1,q_1),\dots,(c_n,q_n)
    \} \subseteq \mn{tran}(\Amf)$ such that $Q$ satisfies $\vp$ and
    for $1 \leq i \leq n$, $x \cdot c_i$ is defined and  a node
    in $T$, and there is a $y \cdot i \in T_r$ such that $r(y \cdot
    i)=(x \cdot c_i,q_i)$;

  \item if $r(y)=(x,\mn{leaf})$, then $x$ is a leaf in $T$.

  \end{enumerate}
  We say that $(T_r,r)$ is \emph{accepting} if in all infinite paths
  $\varepsilon = y_1 y_2 \cdots$ of $T_r$, the set $\{ i \geq 0 \mid
  r(y_i)=(x,q) \text{ for some } q \in R\}$ is infinite.\nb{E: is it really
    like this? shouldn't it be that only the states in $R$ are allowed to
    appear infinitely often?}
  A finite $\Gamma$-labeled tree $(T,L)$ is \emph{accepted} by \Amf if
  there is an accepting run of \Amf on $(T,L)$. We use $L(\Amf)$ to
  denote the set of all finite $\Gamma$-labeled tree accepted by \Amf.
\end{definition}
Apart from \TWABAs,
we will also use nondeterministic tree automata, introduced next.
\smallskip

A \emph{nondeterministic top-down tree automaton} (NTA) on finite
$m$-ary trees is a tuple $\Amf=(Q, \Gamma, Q^0, \delta, F)$ where $Q$
is a finite set of \emph{states}, $\Gamma$ is a finite alphabet, $Q^0
\subseteq Q$ is the set of \emph{initial states}, $F \subseteq Q$ is a
set of \emph{final states}, and $\delta:Q \times \Gamma
\rightarrow 2^{Q^m}$ is the \emph{transition function}.

Let $(T,L)$ be a $\Gamma$-labeled $m$-ary tree.  A \emph{run} of \Amf
on $(T,L)$ is a $Q$-labeled $m$-ary tree $(T,r)$ such that
$r(\varepsilon) \in Q^0$ and for each node $x \in T$, we have $\langle
r(x \cdot 1), \ldots ,r(x \cdot m)\rangle \in \delta( r(x),
L(x))$. The run is \emph{accepting} if for every leaf $x$ of $T$, we
have $r(x) \in F$. The set of trees accepted by $\mathcal A$ is
denoted by $L(\mathcal A)$.

We will use the following results from automata theory:
\begin{theorem}~\\[-4mm]
  \begin{enumerate}

  \item Every \TWABA $\Amf=(Q,\Gamma,\delta,q_0,R)$ can be converted
    into an equivalent NTA $\Amf'$ whose number of states is (single)
    exponential in $|Q|$; the conversion needs time polynomial in the
    size of $\Amf'$;

  \item Given a constant number of \TWABAs (resp.\ NTAs), we can
    construct in polytime a \TWABA (resp.\ an NTA) that accepts the
    intersection language;

  \item Emptiness of NTAs can be checked in polytime.

  \end{enumerate}
\end{theorem}

\subsection{$\Gamma$-labeled Trees}

For the proof of Theorem~\ref{thm:exptbox}, fix an \ALC TBox $\Tmc_1$ and a
Horn-\ALC TBox~$\Tmc_2$ and signatures $\Sigmaabox,\Sigmaquery$.  Set $m :=
|\Tmc_2|$.  Ultimately, we aim to construct an NTA \Amf such that a tree is
accepted by \Amf if and only if this tree encodes a $\Sigmaabox$-ABox \Amc of
outdegree at most $m$ that is consistent with both $\Tmc_1$ and $\Tmc_2$ and a
model $\Imc_1$ of $(\Tmc_1,\Amc)$ such that the canonical model
$\Imc_{\Tmc_2,\Amc}$ of $(\Tmc_2,\Amc)$ is not
con-$\Sigmaquery$-homomorphically embeddable into $\Imc_1$.  By
Theorem~\ref{thm:first}, this means that \Amf accepts the empty language if and
only if $\Tmc_2$ is $(\Sigma_1,\Sigma_2)$-rCQ entailed by $\Tmc_1$.  In this
section, we make precise which trees should be accepted by the NTA \Amf and in
the subsequent section, we construct \Amf.

As before, we assume that $\Tmc_1$ takes the form $\top \sqsubseteq C_{\T_1}$
with $C_{T_1}$ in NNF and use $\mathsf{cl}(C_{\T_1})$ to denote the set of
subconcepts of $C_{\T_1}$, closed under single negation.
We also assume that $\Tmc_2$ is in the Horn-\ALC normal form introduced
above. We use $\mn{CN}(\Tmc_2)$ to denote the
set of concept names in $\Tmc_2$ and $\mn{sub}(\Tmc_2)$ for the set of
subconcepts of (concepts in) $\Tmc_2$.

Let $\Gamma_0$ denote the set of all subsets of
$\Sigmaabox \cup \{ R^- \mid R \in \Sigmaabox \}$ that contain at most one
role.  Automata will run on $m$-ary $\Gamma$-labeled trees where
$$
\Gamma = \Gamma_0 \times 2^{\mn{cl}(\Tmc_1)} \times 2^{\mn{CN}(\Tmc_2)}
\times \{ 0,1 \} \times 2^{\mn{sub}(\Tmc_2)}.
$$
For easier reference, in a $\Gamma$-labeled tree $(T,L)$ and for a node $x$
from $T$, we write $L_i(x)$ to denote the $i+1$st component of $L(x)$, for each
$i \in \{0,\dots,4\}$. Intuitively, the projection of a $\Gamma$-labeled tree
to
\begin{itemize}

\item the $L_0$-components of its $\Gamma$-labels represents the
  $\Sigmaabox$-ABox \Amc that witnesses non-$\Sigmaquery$-query
  entailment of $\Tmc_2$ by $\Tmc_1$;

\item $L_1$-components (partially) represents a model $\Imc_1$ of
  $(\Tmc_1,\Amc)$;

\item $L_2$-components (partially) represents the
  canonical model $\Imc_{\Tmc_2,\Amc}$ of $(\Tmc_2,\Amc)$;

\item $L_3$-components marks the individual $a$ in \Amc such that
  $(\Imc_{\Tmc_2,\Amc},a)$ is not $\Sigmaquery$-simulated by
  $(\Imc_{\Tmc_1,\Amc},a)$;

\item $L_4$-components contains bookkeeping information that helps to
  ensure that the afore mentioned $\Sigmaquery$-simulation indeed fails.

\end{itemize}
We now make these intuitions more precise by defining certain
properness conditions for $\Gamma$-labeled trees, one for each
component in the labels.  A $\Gamma$-labeled $(T,L)$ tree is
\emph{0-proper} if
it satisfies the following conditions:
\begin{enumerate}

\item for the root $\varepsilon$ of $T$, $L_0(\varepsilon)$ contains no role;

\item every non-root node $x$ of $T$, $L_0(x)$ contains a role.

\end{enumerate}
Every 0-proper $\Gamma$-labeled tree $(T, L)$ represents the
tree-shaped $\Sigmaabox$-ABox
\begin{align*}
\Amc_{( T, L)} = \ & \{ A(x) \mid A \in L_0(x) \}\\[1mm] 
  & \cup \{ R(x,y) \mid R \in L_0(y), y \text{ is a child of } x \}\\  
  & \cup \{ R(y,x) \mid  R^{-} \in L_0(y), y \text{ is a child of } x \}. 
\end{align*}
A $\Gamma$-labeled tree $(T, L)$ is \emph{1-proper} if it
satisfies the following conditions for all $x_1,x_2 \in T$:
\begin{enumerate}

\item there is a model \Imc of $\Tmc_1$ and a $d \in \Delta^\Imc$ such
  that $d \in C^\Imc$ iff $C \in L_1(x_1)$ for all $C \in
  \mn{cl}(\Tmc_1)$;


\item $A \in L_0(x_1)$ implies $A \in L_1(x_1)$;

\item if $x_2$ is a child of $x_1$ and $R \in x_2$, then $\forall R . C
  \in L_1(x_1)$ implies $C \in L_1(x_2)$ for all $\forall R . C \in \mn{cl}(\Tmc_1)$;

\item if $x_2$ is a child of $x_1$ and $R^- \in x_2$, then $\forall R . C
  \in L_1(x_2)$ implies $C \in L_1(x_1)$ for all $\forall R . C \in \mn{cl}(\Tmc_1)$.

\end{enumerate}
A $\Gamma$-labeled tree $(T, L)$ is \emph{2-proper} if for every node
$x \in T$,
\begin{enumerate}

\item  $A \in L_2(x)$ iff $\Amc_{(T,L)},\Tmc_2 \models A(x)$, for all $A
  \in \mn{CN}(\Tmc_2)$;

\item if $A \in L_2(x)$, then $A \sqsubseteq \bot \notin \Tmc_2$.


\end{enumerate}
%
It is \emph{3-proper} if there is
exactly one node $x$ with $L_3(x)=1$.  For defining 4-properness,
we first give some preliminaries.

Let $t \subseteq \mn{CN}(\Tmc_2)$. Then $\mn{cl}_{\Tmc_2(S)} = \{ A \in
\mn{CN}(\Tmc_2) \mid \Tmc_2 \models \midsqcap S \sqsubseteq A \}$. Moreover, we
say that $S = \{ \exists R . A, \forall R . B_1, \dots, \forall R . B_n \}$ is
a \emph{$\Sigmaquery$-successor set for} $t$ if there is
a 
concept name $A' \in t$ such that $A' \sqsubseteq \exists R . A \in \Tmc_2$ and
$\forall R . B_1, \dots, \forall R . B_n$ is the set of all concepts of this
form such that, for some $B \in t$, we have $B \sqsubseteq \forall R . B_i \in
\Tmc_2$. In the following, it will sometimes be convenient to speak about the
canonical model $\Imc_{\Tmc_2,\Smc}$ of $\Tmc_2$ and a finite set of concepts
$C$ that occur on the right-hand side of a CI in $\Tmc_2$. What we mean with
$\Imc_{\Tmc_2,\Smc}$ is the interpretation obtained from the canonical model
for the TBox $\Tmc_i \cup \{ A_C \sqsubseteq C \mid C \in \Smc \}$ and the ABox
$ \{ A_C(a_\varepsilon) \mid C \in \Smc \}$ in which all fresh concept names
$A_C$ are removed.

A $\Gamma$-labeled tree is \emph{4-proper} if it satisfies the
following conditions for all nodes $x_1$, $x_2$:
\begin{enumerate}

\item if $L_3(x_1)=1$, then there is a $\Sigmaquery$-concept name in $L_2(x_1)
  \setminus L_1(x_1)$ or $L_4(x_1)$ is
a $\Sigmaquery$-successor set for~$L_2(x_1)$;

\item if $L_4(x_1) \neq \emptyset$, then there is a model \Imc of
  $\Tmc_1$ and a $d \in \Delta^\Imc$ such that $d \in C^\Imc$ iff $C
  \in L_1(x_1)$ for all $C \in \mn{cl}(\Tmc_1)$ and
  $(\Imc_{\Tmc_2,L_4(x_1)},a_\varepsilon) \not\leq_{\Sigmaquery} (\Imc,d)$;


\item if $x_2$ is a child of $x_1$, $L_0(x_2)$ contains the role name
  $R$, and $L_4(x_1)= \{ \exists R . A, \forall R . B_1, \dots,\forall
  R . B_n \}$, then
  %
there is a $\Sigmaquery$-concept name in
  $\mn{cl}_{\Tmc_2}(\{A,B_1,\dots,B_n\}) \setminus L_1(x_2)$ or
$L_4(x_2)$ is
  a $\Sigmaquery$-successor set for $\mn{cl}_{\Tmc_2}(\{A,B_1,\dots,B_n\})$;

\item if $x_2$ is a child of $x_1$, $L_0(x_2)$ contains the role
  $R^-$, and $L_4(x_2)= \{ \exists R . A, \forall R . B_1,
  \dots,\forall R . B_n \}$, then
  %
%
there is a $\Sigmaquery$-concept name in
  $\mn{cl}_{\Tmc_2}(\{A,B_1,\dots,B_n\}) \setminus L_1(x_1)$ or
%
$L_4(x_1)$ is
  a $\Sigmaquery$-successor set for $\mn{cl}_{\Tmc_2}(\{A,B_1,\dots,B_n\})$.
%


\end{enumerate}
Note how 4-properness addresses Condition~2 of
Lemma~\ref{lem:homtosim}. By that condition, there is a set of
simulations from certain pointed ``source'' interpretations to
certain pointed ``target'' interpretations that should be avoided.
In the $L_4$-component of $\Gamma$-labels, we store the source
interpretations, represented as sets of concepts. 4-properness
then ensures that there is no simulation to the relevant target
interpretations.
\begin{lemma}
\label{lem:gammattrees}
There is an $m$-ary $\Gamma$-labeled tree that is $i$-proper for all
$i \in \{0,\dots,4\}$ iff there is a tree-shaped $\Sigmaabox$-ABox
$\Amc$ of outdegree at most $m$ that is consistent w.r.t.\ $\Tmc_1$
and $\Tmc_2$ and a model $\Imc_1$ of
$(\Tmc_1,\Amc)$ such that the canonical model $\Imc_{\Tmc_2,\Amc}$ of
$(\Tmc_2,\Amc)$ is not con-$\Sigmaquery$-homomorphically embeddable
into $\Imc_1$.
\end{lemma}
\begin{proof}
  ``if''. Let $(T,L)$ be an $m$-ary $\Gamma$-labeled tree that is
  $i$-proper for all $i \in \{0,\dots,4\}$. Then $\Amc_{(T,L)}$ is a
  tree-shaped $\Sigmaabox$-ABox of outdegree at most $m$.  Moreover,
  $\Amc_{(T,L)}$ is consistent w.r.t.\ $\Tmc_2$: because of the second
  condition of 2-properness, the canonical model $\Imc_{\Tmc_2,\Amc}$
  of $(\Tmc_2,\Amc)$ is indeed a model of $\Tmc_2$ and \Amc.

  Since $(T,L)$ is $3$-proper, there is exactly one $x_0 \in T$ with
  $L_3(x_0)=1$. By construction, $x_0$ is also an individual name in
  $\Amc_{(T,L)}$.  To finish this direction of the proof, it suffices
  to construct a model $\Imc_1$ of $(\Tmc_1,\Amc_{(T,L)})$ such that
  $(\Imc_{\Tmc_2,\Amc},x_0) \not\leq_{\Sigmaquery} (\Imc_1,x_0)$.  In
  fact, $\Imc_1$ witnesses consistency of $\Amc_{(T,L)}$ with
  $\Tmc_1$. Moreover, by definition of simulations $\Imc_1$ must
  satisfy one of Points~1 and~2 of Lemma~\ref{lem:homtosim} with $a$
  replaced by $x_0$. Consequently, by that Lemma $\Imc_{\Tmc_2,\Amc}$
  is not con-$\Sigmaquery$-homomorphically embeddable into $\Imc_1$.

  Start with the interpretation $\Imc_0$ defined as follows:
  $$
  \begin{array}{r@{~}c@{~}l}
    \Delta^{\Imc_0} &=& T \\[1mm]
    A^{\Imc_0} &=& \{ x \in T \mid A \in L_1(x) \} \\[1mm]
    R^{\Imc_0} &=& \{ (x_1,x_2) \mid x_2 \text{ child of } x_1 \text{
                   and } R \in L_0(x_2) \} \, \cup{} \\[1mm]
    && \{ (x_2,x_1) \mid x_2 \text{ child of } x_1 \text{
                   and } R^- \in L_0(x_2) \}.
  \end{array}
  $$
  Then take, for each $x \in T$, a model $\Imc_x$ of \Tmc such that
  $x \in C^{\Imc_x}$ iff $C \in L_1(x)$ for all
  $C \in \mn{cl}(\Tmc_1)$.  Moreover, if $L_4(x) \neq \emptyset$,
  then choose $\Imc_x$ such that $(\Imc_{\Tmc_2,L_3(x)},a_\varepsilon)
  \not\leq_{\Sigmaquery} (\Imc_x,x)$. These choices are possible since $(T,L)$ is 1-proper
  and 4-proper. Further assume that $\Delta^{\Imc_0}$ and
  $\Delta^{\Imc_x}$ share only the element $x$. Then $\Imc_1$ is the
  union of $\Imc_0$ and all chosen interpretations $\Imc_x$.  It is
  not difficult to prove that $\Imc_1$ is indeed a model of
  $(\Tmc_1,\Amc_{(T,L)})$.

  We show that $(\Imc_{\Tmc_2,\Amc_{(T,L)}},x_0)
  \not\leq_{\Sigmaquery} (\Imc_1,x_0)$.
  By Point~1 of 4-properness, there is a $\Sigmaquery$-concept name in
  $L_2(x_0) \setminus L_1(x_0)$ or $L_4(x_0)$ is a
  $\Sigmaquery$-successor set for $L_2(x)$. In the former case, we are
  done. In the latter case, it suffices to show the following.
  \\[2mm]
  {\bf Claim.} For all $x \in T$: if $L_4(x) \neq \emptyset$, then
    $(\Imc_{\Tmc_2,L_4(x)},a_\varepsilon) \not\leq_{\Sigmaquery}(\Imc_1,x)$.
  \\[2mm]
  The proof of the claim is by induction on the co-depth of $x$ in
  $\Amc_{(T,L)}$, which is the length $n$ of the longest sequence of
  role assertions $R_1(x,x_1),\dots,R_n(x_{n-1},x_n)$ in
  $\Amc_{(T,L)}$. It uses Conditions~2 to~4 of 4-properness.

  \medskip

  ``only if''. Let \Amc be a $\Sigmaabox$-ABox of outdegree at most $m$ that is
  consistent w.r.t.\ $\Tmc_1$ and $\Tmc_2$, and $\Imc_1$ a model of
  $(\Tmc_1,\Amc)$ such that $\Imc_{\Tmc_2,\Amc}$ is not
  con-$\Sigmaquery$-homomorphically embeddable into $\Imc_1$. By duplicating
  successors, we can make sure that every non-leaf in \Amc has exactly $m$
  successors. We can further assume w.l.o.g.\ that $\mn{ind}(\Amc)$ is a
  prefix-closed subset of $\mathbbm{N}^*$ that reflects the tree-shape of \Amc,
  that is, $R(a,b) \in \Amc$ implies $b=a\cdot c$ or $a = b \cdot c$ for some
  $c \in \mathbbm{N}$.  By Lemma~\ref{lem:homtosim}, there is an $a_0 \in
  \mn{ind}(\Amc)$ such that one of the following holds:
  \begin{enumerate}

  \item there is a $\Sigmaquery$-concept name $A$ with
    $a_0 \in A^{\Imc_{\Tmc_2,\Amc}} \setminus A^{\Imc_1}$;

  \item there is an $R_0$-successor $d_0$ of $a_0$ in
    $\Imc_{\Tmc_2,\Amc}$, for some $\Sigmaquery$-role name $R_0$, such
    that $d_0 \notin \mn{ind}(\Amc)$ and for all $R_0$-successors $d$
    of $a_0$ in $\Imc_1$, we have that $(\Imc_{\Tmc_2,\Amc},d_0)
    \not\leq_{\Sigmaquery}(\Imc_1,d)$.

  \end{enumerate}
  We now show how to construct from \Amc a $\Gamma$-labeled tree
  $(T,L)$ that is $i$-proper for all $i \in \{0,\dots,4 \}$.  For each
  $a \in \mn{ind}(\Amc)$, let $R(a)$ be undefined if $a=\varepsilon$
  and otherwise let $R(a)$ be the unique role $R$ (i.e., role name or
  inverse role) such that $R(b,a) \in \Amc$ and $a= b \cdot c$ for
  some $c \in \mathbbm{N}$.  Now set
  $$
  \begin{array}{rcl}
    T &=& \mn{ind}(\Amc) \\[1mm]
    L_1(x) &=& \{ C \in \mn{cl}(\Tmc_1) \mid x \in C^{\Imc_1} \}
    \\[1mm]
    L_2(x) &=& \{ C \in \mn{CN}(\Tmc_2) \mid \Amc,\Tmc_2 \models C(x)
               \} \\[1mm]
    L_3(x) &=& \left \{
               \begin{array}{rl}
                 1 & \text{ if } x=a_0 \\
                 0 & \text{ otherwise}
               \end{array}
               \right .
  \end{array}
  $$
  It remains to define $L_4$.  Start with setting $L_4(x)=\emptyset$ for
  all~$x$. If Point~1 above is true, we are done. If Point~2 is true, then
  there is a $\Sigmaquery$-successor set $S=\{ \exists R_0 . A, \forall R_0
  . B_1,\dots,\forall R_0 . B_n\}$ for $L_2(a_0)$ such that the restriction of
  $\Imc_{\Tmc_2,\Amc}$ to the subtree-interpretation rooted at $d_0$ is the
  canonical model $\Imc_{\Tmc,\{A,B_1,\dots,B_n\}}$. Set $L_4(a_0)=S$. We
  continue to modify $L_4$, proceeding in rounds. To keep track of the
  modifications that we have already done, we use a set
  $$\Gamma \subseteq
  \mn{ind}(\Amc) \times (\NR \cap \Sigmaquery) \times
  \Delta^{\Imc_{\Tmc_2,\Amc}}
   $$
   such that the following conditions are satisfied:
   \begin{itemize}

   \item[(i)] if $(a,R,d) \in \Gamma$, then $L_4(a)$ has the form $\{ \exists R
     . A, \forall R . B_1,\dots,\forall R . B_n\}$ and the restriction of
     $\Imc_{\Tmc_2,\Amc}$ to the subtree-interpretation rooted at $d$ is the
     canonical model $\Imc_{\Tmc,\{A,B_1,\dots,B_n\}}$;

   \item[(ii)] if $(a,R,d) \in \Gamma$ and $d'$ is an $R$-successor of
     $a$ in $\Imc_1$, then $(\Imc_{\Tmc_2,\Amc},d)
     \not\leq_{\Sigmaquery} (\Imc_1,d')$.


   \end{itemize}
   Initially, set $\Gamma = \{ (a_0,R_0,d_0) \}$.  In each round of
   the modification of $L_4$, iterate over all elements $(a,R,d) \in
   \Gamma$ that have not been processed in previous rounds. Let
   $L_4(a)= \{ \exists R . A, \forall R . B_1,\dots,\forall R . B_n\}$
   and iterate over all $R$-successors $b$ of $a$ in \Amc. By (ii),
   $(\Imc_{\Tmc_2,\Amc},d) \not\leq_{\Sigmaquery} (\Imc_1,b)$. By (i),
   there is thus a top-level $\Sigmaquery$-concept name $A'$ in $A
   \sqcap B_1 \sqcap \cdots \sqcap B_n$ such that $b \notin
   A^{\Imc_1}$ or there is an $R'$-successor $d'$ of $d$ in
   $\Imc_{\Tmc_2,\Amc}$, $R'$ a $\Sigmaquery$-role name, such that for
   all $R'$-successors $d''$ of $b$ in $\Imc_1$,
   $(\Imc_{\Tmc_2,\Amc},d') \not\leq_{\Sigmaquery} (\Imc_1,d'')$. In the
   former case, do nothing. In the latter case, there is a
   $\Sigmaquery$-successor set $S'=\{ \exists R' . A', \forall R'
   . B'_1,\dots,\forall R' . B'_{n'}\}$ for $\{A,B_1,\dots,B_n\}$ such
   that the restriction of $\Imc_{\Tmc_2,\Amc}$ to the
     subtree-interpretation rooted at $d'$ is the canonical model
     $\Imc_{\Tmc,\{A',B'_1,\dots,B'_{n'}\}}$. Set $L_4(b)=S'$ and add $(b,R',d')$ to
   $\Gamma$.

   Since we are following only role names (but not inverse roles)
   during the modification of $L_4$ and since \Amc is tree-shaped, we
   will never process tuples $(a_1,R_1,d_1), (a_2,R_2,d_2)$ from
   $\Gamma$ such that $a_1=a_2$. For any $x$, we might thus only
   redefine $L_4(x)$ from the empty set to a non-empty set, but never
   from one non-empty set to another. For the same reason, the
   definition of $L_4$ finishes after finitely many rounds.

   It can be verified that the $\Gamma$-labeled tree $(T,L)$ just
   constructed is $i$-proper for all $i \in \{0,\dots,4\}$. The
   most interesting point is 4-properness, which consists of four
   conditions. Condition~1 is satisfied by construction of $L_4$.
   Condition~2 is satisfied by ($*$) and Conditions~3 and~4 again
   by construction of $L_4$.
\end{proof}

\subsection{Upper Bound in Theorem~\ref{thm:exptbox}}
\label{subsect:autofull}

By Theorem~\ref{thm:first} and Lemma~\ref{lem:gammattrees}, we can
decide whether $\Tmc_1$ does $(\Sigma_1,\Sigma_2)$-rCQ entail $\Tmc_2$
by checking that there is no $\Gamma$-labeled tree that is $i$-proper
for each $i \in \{0,\dots,4\}$. We do this by constructing automata
$\Amc_0,\dots,\Amc_4$ such that each $\Amc_i$ accepts exactly the
$\Gamma$-labeled trees that are $i$-proper, then intersecting the
automata and finally testing for emptiness. Some of the
constructed automata are \TWABAs while others are NTAs. Before intersecting,
all \TWABAs are converted into equivalent NTAs (which involves an
exponential blowup). Emptiness of NTAs can be decided in time
polynomial in the number of states. To achieve \ExpTime overall
complexity, the constructed \TWABAs should thus have at most
polynomially many states while the NTAs can have at most (single)
exponentially many states.
It is straightforward to construct
\begin{enumerate}

\item an NTA $\Amf_0$ that checks 0-properness and has constantly many
  states;

\item a \TWABA $\Amf_1$ that checks 1-properness and whose number of
  states is polynomial in $|\Tmc_1|$ (note that Conditions~1 and~2 of
  1-properness are in a sense trivial as they could also be guaranteed
  by removing undesired symbols from the alphabet $\Gamma$);

\item an NTA $\Amf_3$ that checks 3-properness and has constantly many
  states;

\item an NTA $\Amf_4$ that checks 4-properness and whose number of
states is (single) exponential in $|\Tmc_2|$ (note that Conditions~1
and~2 of 4-properness could again be ensured by refining $\Gamma$).

\end{enumerate}
Details are omitted. It thus remains to construct an automaton
$\Amf_2$ that checks 2-properness.  For this purpose, it is more
convenient to use a \TWABA than an NTA. In fact, the reason for mixing
\TWABAs and NTAs is that while $\Amf_2$ is more easy to be constructed
as a \TWABA, there is no obvious way to construct $\Amf_4$ as a \TWABA
with only polynomially many states: it seems one needs that one state
is needed for every possible value of the $L_4$-components in
$\Gamma$-labels.

The \TWABA $\Amf_2$ is actually the intersection of two \TWABAs
$\Amf_{2,1}$ and $\Amf_{2,2}$. The \TWABA $\Amf_{2,1}$ ensures one
direction of Condition~1 of 2-properness as well as Condition~2,
that is:
\begin{itemize}

\item[(i)] if $\Amc_{(T,L)}$ is consistent w.r.t.\ $\Tmc_2$, then
  $\Amc_{(T,L)},\Tmc_2 \models A(x)$ implies $A \in L_2(x)$ for all $x
  \in T$ and $A \in \mn{CN}(\Tmc_2)$;

\item[(ii)] if $A \in L_2(x)$, then $A \sqsubseteq \bot \notin \Tmc_2$.

\end{itemize}
It is simple for a \TWABA to verify (ii), alternatively one could
refine $\Gamma$.  To achieve~(i), it suffices to guarantee the
following conditions for all $x_1,x_2 \in T$, which are essentially
just the rules of a chase required for Horn-\ALC TBoxes in normal
form:
\begin{enumerate}

\item $A \in L_0(x_1)$ implies $A \in L_2(x_1)$;

\item if $A_1,\dots,A_n \in L_2(x_1)$ and $\Tmc_2 \models A_1
  \sqcap \cdots \sqcap A_n \sqsubseteq A$, then $A \in L_2(x_1)$;

\item if $A \in L_2(x_1)$, $x_2$ is a successor of $x_1$,
  $R \in L_0(x_2)$, and $A \sqsubseteq \forall R . B \in \Tmc_2$,
  then $B \in L_2(x_2)$;

\item if $A \in L_2(x_2)$, $x_2$ is a successor of $x_1$,
  $R^- \in L_0(x_2)$, and $A \sqsubseteq \forall R . B \in \Tmc_2$,
  then $B \in L_2(x_1)$;

\item if $A \in L_2(x_2)$, $x_2$ is a successor of $x_1$,
  $R \in L_0(x_2)$, and $\exists R . A \sqsubseteq B \in \Tmc_2$,
  then $B \in L_2(x_1)$;

\item if $A \in L_2(x_1)$, $x_2$ is a successor of $x_1$,
  $R^- \in L_0(x_2)$, and $\exists R . A \sqsubseteq B \in \Tmc_2$,
  then $B \in L_2(x_2)$

\end{enumerate}
All of this is easily verified with a \TWABA, details are again
ommitted. Note that Conditions~1 and~2 can again be ensured by
refining $\Gamma$.

The purpose of $\Amf_{2,2}$ is to ensure the converse of (i). Before
constructing it, it is convenient to first characterize the entailment
of concept names at ABox individuals in terms of derivation trees.
A \emph{$\Tmc_2$-derivation tree} for an assertion $A_0(a_0)$ in \Amc
with $A_0 \in \mn{CN}(\Tmc_2)$ is a finite $\mn{ind}(\Amc) \times
\mn{CN}(\Tmc_2)$-labeled tree $(T,V)$ that satisfies
the following conditions:
\begin{itemize}

\item $V(\varepsilon)=(a_0,A_0)$;




\item if $V(x)=(a,A)$ and neither $A(a) \in \Amc$ nor \mbox{$\top
  \sqsubseteq A \in \Tmc_2$}, then one of the following holds:
  \begin{itemize}

  \item $x$ has successors $y_1,\dots,y_n$ with
    $V(y_i)=(a,A_i)$ for $1 \leq i \leq n$ and $\Tmc_2 \models A_1
    \sqcap \cdots \sqcap A_n \sqsubseteq A$;

  \item $x$ has a single successor $y$ with $V(y)=(b,B)$ and there is
    an $\exists R . B \sqsubseteq A \in \Tmc_2$ such that
    $R(a,b) \in \Amc$;

  \item $x$ has a single successor $y$ with $V(y)=(b,B)$ and
    there is a $B \sqsubseteq \forall R . A \in \Tmc_2$ such
    that $R(b,a) \in \Amc$.

  \end{itemize}

\end{itemize}
%
%
\begin{lemma}
\label{lem:derivationtrees}
  %
  %
  %
  %
If $\Amc,\Tmc_2 \models A(a)$ and \Amc is consistent w.r.t.\
    $\Tmc_2$, then there is a derivation tree for $A(a)$ in \Amc, for
    all assertions $A(a)$ with $A \in \mn{CN}(\Tmc_2)$ and $a \in \mn{ind}(\Amc)$.
%
  %
  %
\end{lemma}
A proof of Lemma~\ref{lem:derivationtrees} is based on the chase
procedure, details can be found in \cite{BienvenuLW-ijcai-13}
for the extension $\mathcal{ELI}_\bot$ of Horn-\ALC.

We are now ready to construct the remaining \TWABA~$\Amf_{2,2}$.  By
Lemma~\ref{lem:derivationtrees} and since $\Amf_{2,1}$ ensures
that $\Amc_{(T,L)}$ is consistent w.r.t.\ $\Tmc_2$, it is enough for
$\Amf_{2,2}$ to verify that, for each node $x \in T$ and each concept
name $A \in L_2(x)$, there is a $\Tmc_2$-derivation tree for $A(x)$ in
$\Amc_{(T,L)}$.

For readability, we use $\Gamma^- := \Gamma_0 \times \mn{CN}(\Tmc_2)$ as
the alphabet instead of $\Gamma$ since transitions of $\Amf_{2,2}$
only depend on the $L_0$- and $L_2$-components of $\Gamma$-labels. Let
$\mn{rol}(\Tmc_2)$ be the set of all roles $R,R^-$ such that the role
name $R$ occurs in $\Tmc_2$.  Set $\Amf_2=(Q,\Gamma^-,\delta,q_0,R)$
with
  $$
  \begin{array}{rcl}
  Q &=& \{
  q_0 \} \uplus \{ q_A \mid A \in \mn{CN}(\Tmc_2) \} 
   \uplus \{ q_{A,R}, q_R \mid A \in \mn{CN}(\Tmc_2), 
     R \in \mn{rol}(\Tmc_2) \}
  \end{array}
  $$
  and $R = \emptyset$ (i.e., exactly the finite runs are
  accepting).
  %
  For all $(\sigma_0,\sigma_2) \in \Gamma^-$, set
  \begin{itemize}

  \item  $\delta(q_0,(\sigma_0,\sigma_2)) = \displaystyle \bigwedge_{A
      \in \sigma_2} (0,q_A) \wedge (\mn{leaf} \vee \bigwedge_{i \in 1..m} (i,q_0))$;

\item $\delta(q_A,(\sigma_0,\sigma_2))=\mn{true}$ whenever $A \in \sigma_1$ or $\top
  \sqsubseteq A \in \Tmc_2$;

\item $\delta(q_A,(\sigma_0,\sigma_2))=
  \begin{array}[t]{@{}l}
    \bigvee_{\Tmc_2 \models A_1 \sqcap \cdots \sqcap A_n \sqsubseteq A}
    ((0,q_{A_1}) \wedge \cdots \wedge (0,q_{A_n})) \vee{}\\
    \bigvee_{\exists R . B \sqsubseteq A \in \Tmc, \ R \in \Sigmaabox}
    (((0,q_{R^-}) \wedge (-1,q_{B})) \vee
    \bigvee_{i \in 1..m} (i,q_{B,R})) \vee{}\\
    \bigvee_{B \sqsubseteq \forall R . A \in \Tmc, \ R \in \Sigmaabox}
    ((0,q_{R}) \wedge (-1,q_{B})) \vee \bigvee_{i \in 1..m} (i,q_{B,{R^-}}))
  \end{array}$ \\[2mm]
  whenever $A \notin \sigma_0$ and $\top \sqsubseteq A \notin \Tmc_2$;

\item $\delta(q_{A,R},(\sigma_0,\sigma_2))= (0,q_A)$ whenever $R \in \sigma_0$;

\item $\delta(q_{A,R},(\sigma_0,\sigma_2))= \mn{false}$ whenever $R \notin \sigma_0$;

\item $\delta(q_{R},(\sigma_0,\sigma_2))= \mn{true}$ whenever $R \in \sigma_0$;

\item $\delta(q_{R},(\sigma_0,\sigma_2))= \mn{false}$ whenever $R \notin \sigma_0$.

  \end{itemize}
  Note that the finiteness of runs ensures that $\Tmc_2$-derivation
  trees are also finite, as required.
%

\subsection{Upper Bound in Theorem~\ref{thm:2exptbox}}

The 2\ExpTime upper bound stated in Theorem~\ref{thm:2exptbox} can be
obtained by a modification of the construction given in
Section~\ref{subsect:autofull}. We now have to build in the
characterization given in Theorem~\ref{thm:first} instead of the one
from Theorem~\ref{thm:second}. There are two differences: first, the
theorem refers to the canonical model $\Imc_{\Tmc_1,\Amc}$ instead of
quantifying over all models \Imc of $(\Tmc_1,\Amc)$; and second, we
need to consider $\Sigma_2$-homomorphic embeddability instead of
con-$\Sigma_2$-homomorphic embeddability. The former difference can be
ignored. In fact, Theorem~\ref{thm:second} remains true if we quantify
over all models \Imc of $(\Tmc_1,\Amc)$, as in
Theorem~\ref{thm:first}, because $\Imc_{\Tmc_1,\Amc}$ is
$\Sigma_2$-homomorphically embeddable into any model of $\Tmc_1$ and
\Amc. The second difference, however, does make a difference. To
understand it more properly, we first give the following adaptation of
Lemma~\ref{lem:homtosim}.
\begin{lemma}
\label{lem:homtosim2}
Let \Amc be a $\Sigmaabox$-ABox and $\Imc_1$ a model of
$(\Tmc_1,\Amc)$. Then $\Imc_{\Tmc_2,\Amc}$ is not $\Sigmaquery$-homomorphically
embeddable into $\Imc_1$ iff there is an $a \in \mn{ind}(\Amc)$ such
that one of the following is true:
  \begin{enumerate}

  \item there is a $\Sigmaquery$-concept name $A$ with
    $a \in A^{\Imc_{\Tmc_2,\Amc}} \setminus A^{\Imc_1}$;

  \item there is an $R_0$-successor $d_0$ of $a$ in
    $\Imc_{\Tmc_2,\Amc}$, for some $\Sigmaquery$-role name $R_0$, such
    that $d_0 \notin \mn{ind}(\Amc)$ and for all $R_0$-successors $d$
    of $a$ in $\Imc_1$, we have that $(\Imc_{\Tmc_2,\Amc},d_0)
    \not\leq_{\Sigmaquery}(\Imc_1,d)$.

  \item there is an element $d$ in the subtree of
      $\Imc_{\Tmc_2,\Amc}$ rooted at $a$ (with possibly $d=a$) and $d$ has an
    $R_0$-successor $d_0$, for some role name $R_0 \notin
    \Sigmaquery$, such that for all elements $e$ of $\Imc_1$, we have
    $(\Imc_{\Tmc_2,\Amc},d_0) \not\leq_{\Sigmaquery}(\Imc_1,e)$.

  \end{enumerate}
\end{lemma}
The proof of Lemma~\ref{lem:homtosim2} is very similar to that of
Lemma~\ref{lem:homtosim}, details are omitted.  Note that the
difference between Lemma~\ref{lem:homtosim} and
Lemma~\ref{lem:homtosim2} is the additional Condition~3 in the
latter. This condition needs to be reflected in the definition of
proper $\Gamma$-labeled trees which, in turn, requires a
modification of the alphabet~$\Gamma$.

An important reason for the construction in
Section~\ref{subsect:autofull} to yield an \ExpTime upper bound is
that in the $L_4$-component of $\Gamma$-labels, we only need to store
a single successor set instead of a set of such sets. This is not the
case in the new construction (which only yields 2\ExpTime upper bound)
where we let the $L_4$-component of $\Gamma$-labels range over
$2^{2^{\mn{sub}(\Tmc_2)}}$ instead of over $2^{\mn{sub}(\Tmc_2)}$. We
also add an $L_5$-component to $\Gamma$-labels, which also ranges over
$2^{2^{\mn{sub}(\Tmc_2)}}$. The notion of $i$-properness remains the
same for $i \in \{0,1,2,3\}$. We adapt the notion of 4-properness and
add a notion of 5-properness.

As a preliminary, we need to define the notion of a descendant set.
Let $t \subseteq \mn{CN}(\Tmc_2)$ and define $\Gamma$ to be the
smallest set such that
\begin{itemize}

\item $t \in \Gamma$;

\item if $t' \in \Gamma$, $A \in t'$, and $A' \sqsubseteq \exists R . A
  \in \Tmc_2$, then $\{ A,B_1,\dots,B_n\} \in \Gamma$ where
  $B_1,\dots,B_n$ is the set of all concept names such that, for some
  $B \in t'$, we have $B \sqsubseteq \forall R . B_i \in \Tmc_2$.

\end{itemize}
Note that, in the above definition, $R$ need not be from $\Sigmaquery$
(nor from its complement). A subset $s$ of $\mn{CN}(\Tmc_2)$ is a
\emph{descendant set} for $t$ if there is a $t' \in \Gamma$, an $A \in
t'$, and an $A' \sqsubseteq \exists R . A \in \Tmc_2$ with $R \notin
\Sigmaquery$ such that $s$ consists of $A$ and of all concept names
$B$ such that $B' \sqsubseteq \forall R . B \in \Tmc_2$ for some $B'
\in t'$.

A $\Gamma$-labeled tree $(T,L)$ is \emph{4-proper} if it satisfies
the following conditions for all $x_1,x_2 \in T$:
\begin{enumerate}

\item if $L_3(x_1)=1$, then one of the following is true:
  \begin{itemize}
  \item
  there is a $\Sigmaquery$-concept name in $L_2(x_1)
  \setminus L_1(x_1)$;

  \item or $L_4(x_1)$ contains
  a $\Sigma_2$-successor set for~$L_2(x_1)$;

\item $L_5(y)$ contains a $\Sigma_2$-descendant set for~$L_2(x_1)$;

  \end{itemize}

\item there is a model \Imc of $\Tmc_1$ and a $d \in \Delta^\Imc$ such
  that all of the following are true:
  \begin{itemize}

  \item $d \in C^\Imc$ iff $C \in L_1(x_1)$ for all $C \in
    \mn{cl}(\Tmc_1)$;

  \item if $S \in L_4(x_1)$, then
    $(\Imc_{\Tmc_2,S},a_\varepsilon) \not\leq_{\Sigmaquery} (\Imc,d)$;

  \item if $S \in L_5(x_1)$ and $e \in \Delta^\Imc$, then
    $(\Imc_{\Tmc_2,S},a_\varepsilon)
    \not\leq_{\Sigmaquery} (\Imc,e)$;

  \end{itemize}


\item if $x_2$ is a child of $x_1$, $L_0(x_2)$ contains the role name
  $R$, and $L_4(x_1) \ni \{ \exists R . A, \forall R . B_1, \dots,\forall
  R . B_n \}$, then
  %
there is a $\Sigmaquery$-concept name in
  $\mn{cl}_{\Tmc_2}(\{A,B_1,\dots,B_n\}) \setminus L_1(x_2)$ or
$L_4(x_2)$ contains
  a $\Sigmaquery$-successor set for $\mn{cl}_{\Tmc_2}(\{A,B_1,\dots,B_n\})$;

\item if $x_2$ is a child of $x_1$, $L_0(x_2)$ contains the role
  $R^-$, and $L_4(x_2) \ni \{ \exists R . A, \forall R . B_1,
  \dots,\forall R . B_n \}$, then
  %
%
there is a $\Sigmaquery$-concept name in
  $\mn{cl}_{\Tmc_2}(\{A,B_1,\dots,B_n\}) \setminus L_1(x_1)$ or
%
$L_4(x_1)$ contains
  a $\Sigmaquery$-successor set for $\mn{cl}_{\Tmc_2}(\{A,B_1,\dots,B_n\})$.
%


\end{enumerate}
A $\Gamma$-labeled tree $(T,L)$ is \emph{5-proper} if all $x \in T$
agree regarding their $L_5$-label.

Note how the adapted notion of 4-properness and the $L_5$-component of
$\Gamma$-labels implements the additional third condition of
Lemma~\ref{lem:homtosim2}. That condition gives rise to an additional
set of simulations that have to be avoided. The (pointed)
interpretations on the ``source side'' of these simulations are
described using sets of concepts in $L_5$. In the pointed
interpretations $(\Imc_1,e)$ on the ``target side'', we now have to
consider all possible points $e$. For this reason, 5-properness
distributes elements of $L_5$-labels to everywhere else. The
simulations are then avoided via the additional third item in the
second condition of 4-properness. The proof details of the following
lemma are omitted.
\begin{lemma}
\label{lem:gammattrees}
There is an $m$-ary $\Gamma$-labeled tree that is $i$-proper for all
$i \in \{0,\dots,5\}$ iff there is a tree-shaped $\Sigmaabox$-ABox
$\Amc$ of outdegree at most $m$ that is consistent w.r.t.\ $\Tmc_1$
and $\Tmc_2$ and a model $\Imc_1$ of
$(\Tmc_1,\Amc)$ such that the canonical model $\Imc_{\Tmc_2,\Amc}$ of
$(\Tmc_2,\Amc)$ is not $\Sigmaquery$-homomorphically embeddable
into $\Imc_1$.
\end{lemma}
It is now straightforward to adapt the automaton construction to the
new version of 4-properness, and to add an automaton for 5-properness.
The NTA for 4-properness will now have double exponentially many
states because $L_4$- and $L_5$-components are sets of sets of
concepts instead of sets of concepts. In fact, we could dispense NTAs
altogether and use an \TWABA that has exponentially many
states. Overall, we obtain a 2\ExpTime upper bound.

\subsection{2\ExpTime Lower Bound}

We reduce the word problem of
exponentially space bounded alternating Turing machines (ATMs), see
\cite{ChandraKS81}.  An \emph{Alternating Turing Machine (ATM)}
is of the form $M = (Q,\Sigma,\Gamma,q_0,\Delta)$. The set
of \emph{states} $Q = Q_\exists \uplus Q_\forall \uplus \{q_a\} \uplus
\{q_r\}$ consists of \emph{existential states} in $Q_\exists$,
\emph{universal states} in $Q_\forall$, an \emph{accepting state}
$q_a$, and a \emph{rejecting state} $q_r$; $\Sigma$ is the \emph{input
  alphabet} and $\Gamma$ the \emph{work alphabet} containing a
\emph{blank symbol} $\square$ and satisfying $\Sigma \subseteq
\Gamma$; $q_0 \in Q_\exists \cup Q_\forall$ is the \emph{starting}
state; and the \emph{transition relation} $\Delta$ is of the form
$$
  \Delta \; \subseteq \; Q \times \Gamma \times Q \times \Gamma
  \times \{ L,R \}.
$$
We write $\Delta(q,\sigma)$ to denote $\{ (q',\sigma',M) \mid
(q,\sigma,q',\sigma',M) \in \Delta \}$ and assume w.l.o.g.\ that every
set $\Delta(q,\sigma)$ contains exactly two elements when $q$ is
universal, and that the state $q_0$ is existential and cannot be
reached by any transition.

A \emph{configuration} of an ATM is a word $wqw'$ with $w,w' \in
\Gamma^*$ and $q \in Q$. The intended meaning is that the one-side
infinite tape contains the word $ww'$ with only blanks behind it, the
machine is in state $q$, and the head is on the symbol just after
$w$. The \emph{successor configurations} of a configuration $wqw'$ are
defined in the usual way in terms of the transition
relation~$\Delta$. A \emph{halting configuration} (resp.\
\emph{accepting configuration}) is of the form $wqw'$ with $q \in \{
q_a, q_r \}$ (resp.\ $q=q_a$).

A \emph{computation tree} of an ATM $M$ on input $w$ is a tree whose
nodes are labeled with configurations of $M$ on $w$, such that the
descendants of any non-leaf labeled by a universal (resp.\
existential) configuration include all (resp. one) of the successors
of that configuration. A computation tree is \emph{accepting} if the
root is labeled with the \emph{initial configuration} $q_0w$ for $w$
and all leaves with accepting configurations. An ATM $M$ accepts
input $w$ if there is a computation tree of $M$ on $w$.


\smallskip

There is an exponentially space bounded ATM $M$ whose word
problem is {\sc 2Exp\-Time}-hard and we may assume that the length of
every computation path of $M$ on $w \in \Sigma^n$ is bounded
by $2^{2^{n}}$, and all the configurations $wqw'$ in such computation
paths satisfy $|ww'| \leq 2^{n}$, see \cite{ChandraKS81}. 
We may also assume w.l.o.g.\ that $M$ makes at least one step on every
input, and that it never reaches the last tape cell (which is both not
essential for the reduction, but simplifies it).

\medskip Let $w$ be an input to $M$.  We aim to construct Horn-\ALC
TBoxes $\Tmc_1$ and $\Tmc_2$ and a signature $\Sigma$ such that
the following are equivalent:
\begin{enumerate}

\item there is a tree-shaped $\Sigma$-ABox \Amc\xspace
such that
  \begin{enumerate}
  \item \Amc is consistent w.r.t.\ $\Tmc_1$ and $\Tmc_2$ and

  \item $\Imc_{\Tmc_2,\Amc}$ is
    not $\Sigma$-homomorphically embeddable into
    $\Imc_{\Tmc_1,\Amc}$;

  \end{enumerate}

\item $M$ accepts $w$.

\end{enumerate}
Note that we dropped the outdegree condition from
Theorem~\ref{thm:second}. In fact, it is easy to go through the
proofs of that theorem and verify that this condition is not needed;
we have included it because it makes the upper bounds easier.

When dealing with an input $w$ of length $n$, we represent
configurations of $M$ by a sequence of $2^n$ elements linked by the
role name $R$, from now on called \emph{configuration
  sequences}. These sequences are then interconnected to form a
representation of the computation tree of $M$ on $w$. This is
illustrated in Figure~\ref{fig:red3}, which shows three configuration
sequences, enclosed by dashed boxes. The topmost configuration is
universal, and it has two successor configurations. All solid arrows
denote $R$-edges. We will explain later why successor configurations
are separated by two consecutive edges instead of by a single one.
\begin{figure}[tb]
  \centering
    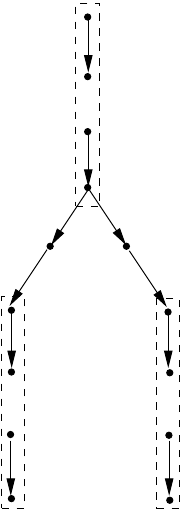
    \caption{Configuration tree (partial)}
    \label{fig:red3}
\end{figure}

The above description is actually an oversimplification. In fact,
every configuration sequence stores two configurations instead of only
one: the current configuration and the previous configuration in the
computation. We will later use the homomorphism condition~(a) above to
ensure that
\begin{itemize}

\item[($*$)] the previous configuration stored in a configuration
sequence is identical to the current configuration stored in its
predecessor configuration sequence.

\end{itemize}
The actual transitions of $M$ are
then represented locally inside configuration sequences.

We next show how to use the TBox $\Tmc_2$ to verify the existence of
a computation tree of $M$ on input $w$ in the ABox,
assuming ($*$).  The signature $\Sigma$ consists of the following
symbols:
\begin{enumerate}

\item concept names $A_0,\dots,A_{n-1}$ and
  $\overline{A}_0,\dots,\overline{A}_{n-1}$ that serve as bits in the
  binary representation of a number between 0 and $2^{n}-1$,
  identifying the position of tape cells inside configuration
  sequences ($A_0$, $\overline{A}_0$ represent the lowest bit);

\item the concept names $A_\sigma$, $A'_\sigma$, $\overline{A}_\sigma$
  for each $\sigma \in \Gamma$;

\item the concept names $A_{q,\sigma}$, $A'_{q,\sigma}$, $\overline{A}_{q,\sigma}$ for
  each $\sigma \in \Gamma$ and $q \in Q$;





\item concept names $X_L,X_R$ that mark left and right successor
  configurations;

\item the role name $R$.

\end{enumerate}
From the above list, concept names $A_\sigma$ and $A_{q,\sigma}$ are
used to represent the current configuration and $A'_\sigma$ and
$A'_{q,\sigma}$ for the previous configuration. The role of the
concept names $\overline{A}_\sigma$ and $\overline{A}_{q,\sigma}$ will
be explained later.

We start with verifying accepting configurations, in a bottom-up manner:
$$
\begin{array}{@{}r@{\,}c@{\,}l}
  A_0 \sqcap \cdots \sqcap A_{n-1} \sqcap A_\sigma \sqcap A'_{\sigma}
  &\sqsubseteq& V \\[1mm]
  A_i \sqcap \exists R . A_i \sqcap \midsqcup_{j < i} \exists R
  . A_j &\sqsubseteq& \mn{ok}_i  \\[1mm]
  \overline{A}_i \sqcap \exists R . \overline{A}_i \sqcap \midsqcup_{j < i} \exists R
  . A_j &\sqsubseteq& \mn{ok}_i  \\[1mm]
  A_i \sqcap \exists R . \overline{A}_i \sqcap \midsqcap_{j < i} \exists R
  . \overline{A}_j &\sqsubseteq& \mn{ok}_i  \\[1mm]
  \overline{A}_i \sqcap \exists R . A_i \sqcap \midsqcap_{j < i} \exists R
  . \overline{A}_j &\sqsubseteq& \mn{ok}_i  \\[1mm]
  \mn{ok}_0 \sqcap  \cdots \sqcap \mn{ok}_{n-1} \sqcap \overline{A}_i \sqcap \exists R . V
  \sqcap A_\sigma \sqcap A'_{\sigma} &\sqsubseteq& V \\[1mm]
  \mn{ok}_0 \sqcap  \cdots \sqcap \mn{ok}_{n-1} \sqcap \overline{A}_i \sqcap \exists R . V
  \sqcap A_\sigma \sqcap A'_{q,\sigma'} &\sqsubseteq& V_{L,\sigma} \\[1mm]
  \mn{ok}_0 \sqcap  \cdots \sqcap \mn{ok}_{n-1} \sqcap \overline{A}_i \sqcap \exists R . V
  \sqcap A_{q_{\mn{acc}},\sigma} \sqcap A'_{\sigma} &\sqsubseteq& V_{R,q_\mn{acc}} \\[1mm]
  \mn{ok}_0 \sqcap  \cdots \sqcap \mn{ok}_{n-1} \sqcap \overline{A}_i \sqcap \exists R . V_{L,\sigma}
  \sqcap A_{q_{\mn{acc}},\sigma'} \sqcap A'_{\sigma'} &\sqsubseteq& V_{L,q_\mn{acc},\sigma} \\[1mm]
  \mn{ok}_0 \sqcap  \cdots \sqcap \mn{ok}_{n-1} \sqcap \overline{A}_i \sqcap \exists R . V_{R,q_{\mn{acc}}}
  \sqcap A_{\sigma} \sqcap A'_{q,\sigma'} &\sqsubseteq& V_{R,q_\mn{acc},\sigma} \\[1mm]
  \mn{ok}_0 \sqcap  \cdots \sqcap \mn{ok}_{n-1} \sqcap \overline{A}_i \sqcap \exists R . V_{M,q_{\mn{acc}},\sigma}
  \sqcap A_{\sigma'} \sqcap A'_{\sigma'} &\sqsubseteq& V_{M,q_\mn{acc},\sigma} \\[1mm]
  \exists R . A_i \sqcap \exists R . \overline{A}_i &\sqsubseteq& \bot
\end{array}
$$
where $\sigma,\sigma'$ range over $\Gamma$, $q$ over $Q$, and
$i$ over $0..n-1$.  The first line starts the verification at the last
tape cell, ensuring that at least one concept name $A_\alpha$ and one
concept name $A'_\beta$ is true. The following lines implement the
verification of the remaining tape cells of the configuration. Lines
two to five implement decrementation of a binary counter and the
conjunct $\overline{A}_i$ in Lines six to eleven prevents the counter
from wrapping around once it has reached zero. We use several kinds
of verification markers:
\begin{itemize}

\item with $V$, we indicate that we have not yet seen the head of the
  TM;

\item $V_{L,\sigma}$ indicates that the TM made a step to the left to
  reach the current configuration, writing $\sigma$;

\item $V_{R,q}$ indicates that the TM made a step to the right to
  reach the current configuration, switching to state $q$;

\item $V_{M,q,\sigma}$ indicates that the TM moved in direction $M$
to reach the current configuration, switching to state $q$ and writing
$\sigma$.

\end{itemize}
In the remaining reduction, we expect that a marker of the form
$V_{M,q,\sigma}$ has been derived at the first cell of the
configuration.  This makes sure that there is exactly one head in the
current and in the previous configuration, and that the head moved
exactly one step between the previous and the current position. Also
note that the above CIs make sure that the tape content does not
change for cells that were not under the head in the previous
configuration. We exploit that $M$ never moves its head to the
right-most tape cell, simply ignoring this case in the CIs above. Note
that it is not immediately clear that lines two to eleven work as
intended since they can speak about different $R$-successors for
different bits. The last line fixes this problem.

We also ensure that relevant concept names are mutually exclusive:
$$
\begin{array}{rcll}
  A_i \sqcap \overline{A}_i &\sqsubseteq& \bot \\[1mm]
  A_{\sigma_1} \sqcap A_{\sigma_2} &\sqsubseteq& \bot & \text{ if }
  \sigma_1 \neq \sigma_2 \\[1mm]
  A_{\sigma_1} \sqcap A_{q_2,\sigma_2} &\sqsubseteq& \bot \\[1mm]
  A_{q_1,\sigma_1} \sqcap A_{q_2,\sigma_2} &\sqsubseteq& \bot & \text{ if }
  (q_1,\sigma_1) \neq (q_2,\sigma_2)
\end{array}
$$
where the $i$ ranges over $0..n-1$, $\sigma_1,\sigma_2$ over $\Gamma$,
and $q_1,q_2$ over~$Q$. We also add the same concept inclusions for
the primed versions of these concept names.  The next step is to
verify non-halting configurations:
$$
\begin{array}{@{}r@{\,}c@{\,}l}
  \exists R. \exists R . (X_L \sqcap \overline{A}_0 \sqcap \cdots \sqcap
  \overline{A}_{n-1} \sqcap (V_{M,q,\sigma} \sqcup V'_{M,q,\sigma})) &\sqsubseteq& L\mn{ok} \\[1mm]
  \exists R . \exists R . (X_R  \sqcap \overline{A}_0 \sqcap \cdots \sqcap
  \overline{A}_{n-1}\sqcap (V_{M,q,\sigma} \sqcup V'_{M,q,\sigma})) &\sqsubseteq& R\mn{ok} \\[1mm]
  A_0 \sqcap \cdots \sqcap A_{n-1} \sqcap A_\sigma \sqcap A'_{\sigma}
  \sqcap L\mn{ok} \sqcap R\mn{ok}   &\sqsubseteq& V' \\[1mm]
  \mn{ok}_0 \sqcap  \cdots \sqcap \mn{ok}_{n-1} \sqcap \overline{A}_i \sqcap \exists R . V'
  \sqcap A_\sigma \sqcap A'_{\sigma} &\sqsubseteq& V' \\[1mm]
  \mn{ok}_0 \sqcap  \cdots \sqcap \mn{ok}_{n-1} \sqcap \overline{A}_i \sqcap \exists R . V'
  \sqcap A_\sigma \sqcap A'_{q,\sigma'} &\sqsubseteq& V'_{L,\sigma} \\[1mm]
  \mn{ok}_0 \sqcap  \cdots \sqcap \mn{ok}_{n-1} \sqcap \overline{A}_i \sqcap \exists R . V_{R,q}
  \sqcap A_{\sigma} \sqcap A'_{q',\sigma'} &\sqsubseteq& V'_{R,q,\sigma} \\[1mm]
  \mn{ok}_0 \sqcap  \cdots \sqcap \mn{ok}_{n-1} \sqcap \overline{A}_i \sqcap \exists R . V'_{M,q,\sigma}
  \sqcap A_{\sigma'} \sqcap A'_{\sigma'} &\sqsubseteq& V'_{M,q,\sigma}
\end{array}
$$
where $\sigma,\sigma',\sigma''$ range over $\Gamma$, $q$ and $q'$ over
$Q$, and $i$ over $0..n-1$. We switch to different verification
markers $V'$, $V'_{L,\sigma}$, $V'_{R,q}$, $V'_{M,q,\sigma}$ to
distinguish halting from non-halting configurations. Note that the
first verification step is different for the latter: we expect to see
one successor marked $X_L$ and one marked $X_R$, both the first cell of an
already verified (halting or non-halting) configuration. For easier
construction, we require two successors also for existential
configurations; they can simply be identical.  The above inclusions do
not yet deal with cells where the head is currently located. We need
some prerequisites because when verifying these cells, we want to
(locally) verify the transition relation. For this purpose, we carry
the transitions implemented locally at a configuration up to its
predecessor configuration:
$$
\begin{array}{r@{~}c@{~}l}
  \exists R . \exists R . (X_M \sqcap \overline{A}_0 \sqcap \cdots \sqcap
  \overline{A}_{n-1} \sqcap  V_{q,\sigma,M'}) & \sqsubseteq&
  S^M_{q,\sigma,M'} \\[1mm]
  \exists R . \exists R . (X_M \sqcap \overline{A}_0 \sqcap \cdots \sqcap
  \overline{A}_{n-1} \sqcap  V'_{q,\sigma,M'}) & \sqsubseteq&
  S^M_{q,\sigma,M'} \\[1mm]
  \exists R . (A_{\sigma} \sqcap S^M_{q,\sigma',M}) &\sqsubseteq&
  S^M_{q,\sigma',M} \\[1mm]
\end{array}
$$
where $q$ ranges over $Q$, $\sigma$ and $\sigma'$ over $\Gamma$, $M$
over $\{L,R\}$, and $i$ over $0..n-1$.  Note that markers are
propagated up exactly to the head position. One issue with the above
is that additional $S_{q\sigma,M}$-markers could be propagated up not
from the successors that we have verified, but from surplus
(unverified) successors. To prevent such undesired markers, we put
$$
S^M_{q_1,\sigma_1,M_1} \sqcap S^M_{q_2,\sigma_2,M_2} \sqsubseteq \bot
$$
for all $M \in \{L,R\}$ and all distinct
$(q_1,\sigma_1,M_1),(q_2,\sigma_2,M_2) \in Q \times \Gamma \times
\{L,R\}$.  We can now implement the verification of cells under the
head in non-halting configurations. Put
$$
\begin{array}{rcl}
  \mn{ok}_0 \sqcap  \cdots \sqcap \mn{ok}_{n-1} \sqcap \overline{A}_i \sqcap \exists R . V'
  \sqcap
  A_{q_1,\sigma_1} \sqcap A'_{\sigma_1} \sqcap S^L_{q_2,\sigma_2,M_2}
  \sqcap S^R_{q_3,\sigma_3,M_3}&\sqsubseteq& V'_{R,q_1} \\[2mm]
  \mn{ok}_0 \sqcap  \cdots \sqcap \mn{ok}_{n-1} \sqcap \overline{A}_i \sqcap \exists R . V'_{L,\sigma}
  \sqcap
  A_{q_1,\sigma_1} \sqcap A'_{\sigma_1} \sqcap S^L_{q_2,\sigma_2,M_2}
  \sqcap S^R_{q_3,\sigma_3,M_3}&\sqsubseteq& V'_{L,q_1,\sigma}
\end{array}
$$
for all $(q_1,\sigma_1) \in Q \times \Gamma$ with $q_1$ a universal
state and
$\Delta(q_1,\sigma_1)=\{(q_2,\sigma_2,M_2),(q_3,\sigma_3,M_3)\}$, $i$
from $0..n-1$, and $\sigma$ from $\Gamma$; moreover, put
$$
\begin{array}{rcl}
   \mn{ok}_0 \sqcap  \cdots \sqcap \mn{ok}_{n-1} \sqcap \overline{A}_i \sqcap \exists R . V'
   \sqcap
   A_{q_1,\sigma_1} \sqcap A'_{\sigma_1} \sqcap S^L_{q_2,\sigma_2,M_2}
   \sqcap S^R_{q_2,\sigma_2,M_2}&\sqsubseteq& V'_{R,q_1} \\[2mm]
   \mn{ok}_0 \sqcap  \cdots \sqcap \mn{ok}_{n-1} \sqcap \overline{A}_i \sqcap \exists R . V'_{L,\sigma}
   \sqcap
   A_{q_1,\sigma_1} \sqcap A'_{\sigma_1} \sqcap S^L_{q_2,\sigma_2,M_2}
   \sqcap S^R_{q_2,\sigma_2,M_2}&\sqsubseteq& V'_{L,q,\sigma}
\end{array}
$$
for all $(q_1,\sigma_1) \in Q \times \Gamma$ with $q_1$ an existential
state, for all $(q_2,\sigma_2,M_2) \in \Delta(q_1,\sigma_1)$, all $i$
from $0..n-1$, and all $\sigma$ from~$\Gamma$.  It remains to verify
the initial configuration. Let $w=\sigma_0 \cdots \sigma_{n-1}$, let
$(C=i)$ be the conjunction over the concept names $A_i$,
$\overline{A}_i$ that expresses $i$ in binary for $0 \leq i < n$, and
let $(C \geq n)$ be the Boolean concept over the concept names $A_i$,
$\overline{A}_i$ which expresses that the counter value is at least
$n$. Then put
$$
\begin{array}{r@{\,}c@{\,}l}
  A_0 \sqcap \cdots \sqcap A_{n-1} \sqcap A_\Box \sqcap A'_\sigma
  \sqcap L\mn{ok} \sqcap R\mn{ok}   &\sqsubseteq& V^I \\[1mm]
\mn{ok}_0 \sqcap  \cdots \sqcap \mn{ok}_{n-1}
\sqcap (C \geq n) \sqcap \exists R . V^I
  \sqcap A_\Box \sqcap A'_\sigma &\sqsubseteq& V^I \\[1mm]
 \mn{ok}_0 \sqcap  \cdots \sqcap \mn{ok}_{n-1} \sqcap (C=i) \sqcap \exists R . V^I
  \sqcap A_{\sigma_i}\sqcap A'_\sigma &\sqsubseteq& V^I \\[1mm]
 \mn{ok}_0 \sqcap  \cdots \sqcap \mn{ok}_{n-1} \sqcap (C=1) \sqcap \exists R . V^I
  \sqcap A_{\sigma_1}\sqcap A'_{q,\sigma'} &\sqsubseteq& V^I_{R,q}
\end{array}
$$
where $i$ ranges over $2..n-1$ and $\sigma,\sigma'$ over $\Gamma$.  This
verifies the initial conditions except for the left-most cell, where the head
must be located (in initial state $q_0$) and where we must verify the
transition, as in all other configurations.  Recall that we assume
$q_0$ to be an existential state. We can thus add
$$
\begin{array}{rcl}
   \mn{ok}_0 \sqcap  \cdots \sqcap \mn{ok}_{n-1} \sqcap (C=0) \sqcap \exists R . V^I_{R,q}
   \sqcap
   A_{q_0,\sigma_0} \sqcap A'_\sigma \sqcap S^L_{q,\sigma,M}
   \sqcap S^R_{q,\sigma,M}&\sqsubseteq& I
\end{array}
$$
for all $(q,\sigma,M) \in \Delta(q_0,\sigma_0)$ and $\sigma \in \Gamma$.

At this point, we have finished the verification of the computation
tree, except that we have assumed but not yet established ($*$). To
achieve ($*$), we use both $\Tmc_1$ and $\Tmc_2$. Let
$\alpha_0,\dots,\alpha_{k-1}$ be the elements of $\Gamma \cup (Q
\times \Gamma)$. We use concept names $A^\ell_i$,
$\overline{A}^\ell_i$, $\ell \in \{0,\dots,k-1\}$, to implement $k$
additional counters. This time, we have to count up to $2^n+1$
(because successor configuration sequences are separated by two
edges), so $i$ ranges from $0$ to $m := \lceil \mn{log}(2^n+1)
\rceil$.  We first add to $\Tmc_2$:
$$
\begin{array}{rcl}
  \exists R . I & \sqsubseteq & \exists S .  \midsqcap_{\ell <k}
  \exists R . (A_{\alpha_\ell} \sqcap (C^\ell = 0)) \\[1mm]
  \overline{A}^\ell_i &\sqsubseteq& \exists R . \top \\[1mm]
  A^\ell_i \sqcap \midsqcap_{j < i} A^\ell_j & \sqsubseteq& \forall R
  . \overline{A}^\ell_i \\[1mm]
  \overline{A}^\ell_i \sqcap \midsqcap_{j < i} A^\ell_j & \sqsubseteq& \forall R
  . A^\ell_i \\[1mm]
  A^\ell_i \sqcap \midsqcup_{j < i} \overline{A}^\ell_j & \sqsubseteq& \forall R
  . A^\ell_i \\[1mm]
  \overline{A}^\ell_i \sqcap \midsqcup_{j < i} \overline{A}^\ell_j & \sqsubseteq& \forall R
  . \overline{A}^\ell_i \\[1mm]
  (C^\ell = 2^n) & \sqsubseteq & \overline{A}_{\alpha_\ell}
\end{array}
$$
where $\ell$ ranges over $0..k-1$, $i$ over $0..m-1$, and $(C^\ell = j)$
denotes the conjunction over $A^\ell_i$, $\overline{A}^\ell_i$ which
expresses that the value of the $\ell$-th counter is $j$.  We will
explain shortly why we need to travel one more $R$-step (in the first
line) after seeing~$I$.

The above inclusions generate, after the verification of the
computation tree has ended successfully, a tree in the canonical model
of the input ABox and of $\Tmc_2$ as shown in Figure~\ref{fig:mod}.
\begin{figure}[t!]
  \begin{center}
    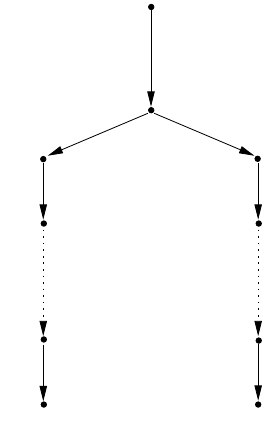
    \caption{Tree gadget.}
    \label{fig:mod}
  \end{center}
\end{figure}
Note that the topmost edge is labeled with the role name $S$, which is
\emph{not} in~$\Sigma$. By Condition~(b) above and since, up to now,
we have always only used non$-\Sigma$-symbols on the right-hand side
of concept inclusions, we must not (homomorphically) find the subtree
rooted at the node with the incoming $S$-edge \emph{anywhere} in the
canonical model of the ABox and $\Tmc_1$. We use this effect which we
to ensure that ($*$) is satisfied \emph{everywhere}. Note that, the
paths in Figure~\ref{fig:mod} have length $2^n+1$ and that we do not
display the labeling with the concept names $A^\ell_i$,
$\overline{A}^\ell_i$.  These concept names are not in $\Sigma$ anyway
and only serve the purpose of achieving the intended path
length. Intuitively, every path in the tree represents one possible
\emph{copying defect}. The concept names of the form
$\overline{A}_\alpha$ need not occur in the input ABox and
stand for the disjunction over all $\overline{A}_\beta$ with $\beta
\neq \alpha$. They need to be in $\Sigma$, though, because we want
them to be taken into account in $\Sigma$-homomorphisms.

\smallskip

We
next extend $\Tmc_1$ as follows:
$$
\begin{array}{rcl}
  A_{\alpha} &\sqsubseteq& \overline{A}_\beta \\[1mm]
  \exists R . A_{\alpha_i} & \sqsubseteq& \midsqcap_{\ell \in \{0,\dots,k-1\}
    \setminus \{i\}} \exists R . (A_{\alpha_\ell} \sqcap
(C^\ell=0)) \\[1mm]
  \overline{A}^\ell_i &\sqsubseteq& \exists R . \top \\[1mm]
  A^\ell_i \sqcap \midsqcap_{j < i} A^\ell_j & \sqsubseteq& \forall R
  . \overline{A}^\ell_i \\[1mm]
  \overline{A}^\ell_i \sqcap \midsqcap_{j < i} A^\ell_j & \sqsubseteq& \forall R
  . A^\ell_i \\[1mm]
  A^\ell_i \sqcap \midsqcup_{j < i} \overline{A}^\ell_j & \sqsubseteq& \forall R
  . A^\ell_i \\[1mm]
  \overline{A}^\ell_i \sqcap \midsqcup_{j < i} \overline{A}^\ell_j & \sqsubseteq& \forall R
  . \overline{A}^\ell_i \\[1mm]
  (C^\ell = 2^n) & \sqsubseteq & \overline{A}_{\alpha_\ell}
\end{array}
$$
where $\ell$ ranges over $0..k-1$, $i$ over $0..m-1$, and
$\alpha,\beta$ over distinct elements of $\Gamma \cup (Q\times
\Gamma)$. Note that it is not important to use the same counter
concepts $A^\ell_i$, $\overline{A}^\ell_i$ in $\Tmc_1$ and $\Tmc_2$:
since they are not in $\Sigma$, one could as well use different ones.
Also note that the intended behaviour of the concept names
$\overline{A}_\alpha$ is implemented in the first line.

\smallskip

The idea for achieving ($*$) is as follows: the tree shown in
Figure~\ref{fig:mod} contains all possible copying defects, that is,
all paths of length $2^n+1$ such that $A_{\alpha_i}$ is true at the
beginning, but some $A_{\alpha_j}$ with $j \neq i$ is true at the end.
At each point of the computation tree where some $A_{\alpha_i}$ is
true at an $R$-predecessor, the above inclusions in $\Tmc_1$ generate
a tree in the canonical model of the ABox and of $\Tmc_1$ which is
similar to that in Figure~\ref{fig:mod}, except that the initial
$S$-edge and the path representing an $A_{\alpha_i}$-defect are
missing. Consequently, if $A_{\alpha_i}$ is not properly copied to
$A'_{\alpha_i}$ at all nodes that are $2^n+1$ $R$-steps away, then we
homomorphically find the tree from Figure~\ref{fig:mod} in the
canonical model of the ABox and of $\Tmc_1$. Consequently, not finding
the tree anywhere in that model means that all copying is done
correctly.

We need to avoid that the inclusions in $\Tmc_1$ enable a homomorphism
from the tree in Figure~\ref{fig:mod} due to an ABox where some node
has two $R$-successors labeled with different concepts $A_\alpha$,
$A_\beta$:
$$
  \exists R. A_\alpha \sqcap \exists R . A_\beta \sqsubseteq \bot.
$$
This explains why we need to separate successor configurations by two
$R$-steps. In fact, the mid point needs not make true any of the
concept names $A_\alpha$ and thus we are not forced to violate the
above constraint when branching at the end of configuration
sequences. Also note that copying the content of the first cell of the
initial configuration requires traveling one more $R$-step after
seeing $I$, as implemented above.
\begin{lemma}
\label{lem:2expcorr}
  The following conditions are equivalent:
\begin{enumerate}

\item there is a tree-shaped $\Sigma$-ABox \Amc\xspace such that
  \begin{enumerate}
  \item \Amc is consistent w.r.t.\ $\Tmc_1$ and $\Tmc_2$ and

  \item $\Imc_{\Tmc_2,\Amc}$ is not
    $\Sigma$-homomorphically embeddable into
    $\Imc_{\Tmc_1,\Amc}$;

  \end{enumerate}

\item $M$ accepts $w$.

\end{enumerate}
\end{lemma}
\begin{proof}(sketch) For the direction ``$2 \Rightarrow 1$'', assume
  that $M$ accepts $w$. An accepting computation tree of $M$ on $w$
  can be represented as a $\Sigma$-ABox as detailed above alongside
  the construction of the TBoxes $\Tmc_2$ and $\Tmc_1$. The
  representation only uses the role name $R$ and the concept names of
  the form $A_i$, $\overline{A}_i$, $A_\sigma$, $A_{q,\sigma}$,
  $A'_\sigma$, $A'_{q,\sigma}$, $X_L$, and $X_R$, but not the concept
  names of the form $\overline{A}_\sigma$ and
  $\overline{A}_{q,\sigma}$. As explained above, we need to duplicate
  the successor configurations of existential configurations to ensure
  that there is binary branching after each configuration. Also, we
  need to add one additional incoming $R$-edge to the root of the tree
  as explained above. The resulting ABox \Amc is consistent w.r.t.\
  $\Tmc_1$ and $\Tmc_2$. Moreover, since there are no copying defects,
  there is no homomorphism from $\Imc_{\Tmc_2,\Amc}$ to
  $\Imc_{\Tmc_1,\Amc}$.

  \smallskip

  For the direction ``$1 \Rightarrow 2$'', assume that there is a
  tree-shaped $\Sigma$-ABox \Amc that satisfies Conditions~(a)
  and~(b).  Because of Condition~(b), $I$ must be true somewhere in
  $\Imc_{\Tmc_2,\Amc}$: otherwise, $\Imc_{\Tmc_2,\Amc}$ does not
  contain anonymous elements and the identity is a homomorphism from
  $\Imc_{\Tmc_2,\Amc}$ to $\Imc_{\Tmc_1,\Amc}$, contradicting
  (b). Since $I$ is true somewhere in $\Imc_{\Tmc_2,\Amc}$ and by
  construction of $\Tmc_2$, the ABox must contain the representation
  of a computation tree of $M$ on $w$, except satisfaction of ($*$).
  For the same reason, $\Imc_{\Tmc_2,\Amc}$ must contain a tree as
  shown in Figure~\ref{fig:mod}.  As has already been argued during
  the construction of $\Tmc_2$ and $\Tmc_1$, however, condition ($*$)
  follows from the existence of such a tree in $\Imc_{\Tmc_2,\Amc}$
  together with (b).
\end{proof}
We remark that the above reduction also yields 2\ExpTime hardness for
CQ entailment in $\mathcal{ELI}$. In fact, universal restrictions on
the right-hand sides of concept inclusions can easily be simulated
using universal roles and disjunctions on the left-hand sides can be
removed with only a polynomial blowup (since there are always only two
disjuncts). It thus remains to eliminate $\bot$, which only occurs
non-nested on the right-hand side of concept inclusions. With the
exception of the inclusions
$$
S^M_{q_1,\sigma_1,M_1} \sqcap S^M_{q_2,\sigma_2,M_2} \sqsubseteq \bot,
$$
this can be done as follows: include all concept inclusions with
$\bot$ on the right-hand side in $\Tmc_1$ instead of in $\Tmc_2$; then
replace $\bot$ with $D$ and add the following concept inclusions to
$\Tmc_1$:
$$
\begin{array}{rcl}
  D& \sqsubseteq & \exists S .  \midsqcap_{\ell <k}
  \exists R . (A_{\alpha_\ell} \sqcap (C^\ell = 0)) \\[1mm]
  \overline{A}^\ell_i &\sqsubseteq& \exists R . \top \\[1mm]
  A^\ell_i \sqcap \midsqcap_{j < i} A^\ell_j & \sqsubseteq& \forall R
  . \overline{A}^\ell_i \\[1mm]
  \overline{A}^\ell_i \sqcap \midsqcap_{j < i} A^\ell_j & \sqsubseteq& \forall R
  . A^\ell_i \\[1mm]
  A^\ell_i \sqcap \midsqcup_{j < i} \overline{A}^\ell_j & \sqsubseteq& \forall R
  . A^\ell_i \\[1mm]
  \overline{A}^\ell_i \sqcap \midsqcup_{j < i} \overline{A}^\ell_j & \sqsubseteq& \forall R
  . \overline{A}^\ell_i \\[1mm]
  (C^\ell = 2^n) & \sqsubseteq & \overline{A}_{\alpha_\ell}
\end{array}
$$
where $\ell$ ranges over $0..k-1$ and $i$ over $0..m$.  The effect of
these additions is that any ABox which satisfies the left-hand side of
a $\bot$-concept inclusion in the original $\Tmc_2$ cannot satisfy
Condition~(b) from Lemma~\ref{lem:2expcorr} and thus needs not be
considered.

For the inclusions excluded above, a different approach needs to be
taken. Instead of introducing the concept names
$S^M_{q_1,\sigma_1,M_1}$, one would propagate transitions inside the
$V'$-markers. Thus, $S^{L}_{q_1,\sigma_1,M_1}$,
$S^R_{q_2,\sigma_2,M_2}$, and $V'$ would be integrated into a single
marker $V'_{q_1,\sigma_1,M_1,q_2,\sigma_2,M_2}$, and likewise for
$V_{R,q}$. The concept inclusion excluded above can then simply be
dropped.

\end{document}